\documentclass[9pt,twocolumn,twoside]{pnas-new}
\templatetype{pnasresearcharticle} 
\articletype{} 

\setboolean{displaywatermark}{false}
\makeatletter
\patchcmd{\@maketitle}{\includegraphics[height=24.5pt,width=66.4pt]{PNAS_Logo.pdf}}{}{}{\PackageWarning{arxiv}{logo patch failed}}
\patchcmd{\@maketitle}{\put(28.9,9.05){\color{typecol}\sffamily\bfseries\footnotesize RESEARCH ARTICLE}}{}{}{\PackageWarning{arxiv}{label patch failed}}
\patchcmd{\@maketitle}{\put(102,3.55){\color{gray80}{\rule{0.5pt}{16pt}}}}{}{}{\PackageWarning{arxiv}{rule patch failed}}
\patchcmd{\@maketitle}{\put(110,9.05){\@articletype}\par\vskip11\p@}{}{}{\PackageWarning{arxiv}{type patch failed}}
\fancypagestyle{firststyle}{\fancyhf{}\fancyfoot[R]{\footerfont\textbf{\thepage\ of \pageref*{LastPage}}}}
\makeatother

\newcommand{\Kt}{\mathsf{Kt}}

\newcommand{\twscore}{\mathrm{tw}_{\mathrm{score}}}

\providecommand{\dropcap}[1]{#1}

\title{The program is still there, and searching for it obeys a conservation law}

\author[a,1]{Jorge Miguel Silva}
\affil[a]{Institute of Electronics and Informatics Engineering of Aveiro (IEETA) and Department of Electronics, Telecommunications and Informatics (DETI), University of Aveiro, 3810-193 Aveiro, Portugal}

\leadauthor{Silva}

\significancestatement{A short program that generates a sequence compresses it and predicts its continuation. No algorithm can find the shortest such program for every sequence, so we ask what finding a short one costs when one exists. For searches that see only how well candidate programs reproduce the sequence, we prove that, under stated conditions, the worst-case cost grows exponentially with the length of the shortest program. By a conservation law, structural knowledge replaces search at most bit for bit. Our system recovered certified programs for 2,383 of 3,914 candidate integer sequences and cellular automata, and it failed only where the program it found did not compress, where a short program lay outside its reach, or where time ran out.}

\authorcontributions{J.M.S. designed and performed the research, developed the system, analyzed the data, and wrote the paper.}
\authordeclaration{The author declares no competing interest.}
\correspondingauthor{\textsuperscript{1}To whom correspondence should be addressed. E-mail: jorge.miguel.ferreira.silva@ua.pt}


\begin{abstract}
The length of the shortest program that generates a sequence, its Kolmogorov complexity, is uncomputable, so we ask what finding a short program costs when one exists. Many search algorithms, including Levin search, evolutionary methods, simulated annealing and the cross-entropy method, learn about a candidate program only through its score, how much of the sequence it reproduces. For such algorithms we define the coupling width of a search problem, the largest number of program fields that must change together before the score moves. We prove an unconditional worst-case lower bound on the number of candidates scored, exponential in that width. Under stated conditions, a conservation law follows, in which each bit of structural knowledge injected into a search removes at most one bit of search. The structure plus the logarithm of the remaining search is at least a fixed fraction of the sequence's Kolmogorov complexity. That fraction tends to one as the instruction set grows, so the lower bound approaches Levin's 1973 upper bound, apart from the program's running time. Search by score must pay this cost to be complete. We prove that our engine's structural inference avoids it only by giving up coverage, as does, if one-way functions exist, every method that runs in polynomial time. Our deterministic engine, built on this theory, recovers a generating program, certified by compressing its data and predicting an unseen continuation, for 2,383 of 3,914 candidate sequences in four populations of integer sequences and cellular automata, including the chaotic Rule 30.
\end{abstract}

\dates{This manuscript was compiled on \today}

\begin{document}

\maketitle
\thispagestyle{firststyle}
\ifthenelse{\boolean{shortarticle}}{\ifthenelse{\boolean{singlecolumn}}{\abscontentformatted}{\abscontent}}{}

\dropcap{A} sequence can be thought of as the output of a program. The length of the shortest such program, on a fixed universal machine, is the Kolmogorov complexity of the sequence \cite{kolmogorov, chaitin}. For example, the center column of the cellular automaton Rule 30 is chaotic, yet its first 500 binary terms are the output of a program of 74.5 bits (Section~\ref{sec:scale} and SI Appendix~B). Solomonoff's theory of inductive inference \cite{solomonoff} predicts the continuation of a sequence from all the programs that generate it, each weighted by two to the power of minus its length, so that the shortest programs count most.

Three questions arise. Can a short generating program be found, can a search tell when it has found one, and what does finding it cost? Kolmogorov complexity is not computable, so no general procedure can certify that a program is the shortest. This fact rules out certifying minimality but not finding a short generating program when one exists, and since Levin~\cite{levin} we know how to find one. The program is still there, and the open question is what it costs to reach it.

Finding a generating program does not require minimality, and neither does certifying one. A program predicts when it keeps producing the sequence correctly past the part we have seen. Such a program has captured structure of the process behind the sequence, because a short program cannot reproduce an unseen continuation by accident, and a program long enough to encode that continuation is no longer short. If it is also shorter than the data it accounts for, it compresses, so it cannot be a lookup table or a stored copy. We certify a program as a generator when it does both, and we call this test the gate. We measure program length as its description length, the number of bits needed to write it in one fixed code (Materials and Methods). This length is a computable upper bound on the Kolmogorov complexity of the sequence up to a logarithmic term.

Many program searches learn about a candidate only through its score, a measure of how much of the target sequence it reproduces. We call this setting the score-oracle model. Levin search and evolutionary methods belong to it. In this model a search problem has a coupling width, the largest number of program fields (the choices left open, such as an operation or a register) that must all change at once before the score moves. This quantity sets how much search is required.

We prove, unconditionally, that in the worst case the number of score queries grows exponentially with the coupling width, and that this worst case is attained on a class of register programs. On that class, reading the data instead recovers the program with a few arithmetic operations per unknown field (SI Appendix~A, Proposition~14).

We also prove a conservation law for this cost. Structure built into a search, such as a reformulation that reencodes the problem so that fewer fields are coupled, trades at best one for one against the search that remains. Under conditions stated in Section~\ref{sec:law}, the two together, counted in bits, are at least a fixed fraction of the length of the shortest program for the sequence. The no free lunch theorems \cite{wolpertmacready1997} are also a conservation result, but for performance averaged over all objective functions, while our law is a budget for each instance, priced in bits of Kolmogorov complexity.

We build a system, which we call OMNIS, whose strategies make this trade (Fig.~\ref{fig:arch}). We bring together three things not combined before, a test that certifies a discovered program, a barrier that explains why passing the test is hard, and a run of that test on thousands of sequences. The system recovers a generating program for 2,383 of 3,914 candidate sequences in four populations of integer sequences and cellular automata, covering 606 of the 2,095 distinct targets they hold. It fails where the program it finds is no shorter than the observed terms, where a short program lies outside the search's coverage, and where the search for a long program runs out of time.

We prove that the system's structural inference, which works out part of each program without consulting the score, avoids the exponential cost only by giving up coverage. We also prove that, if one-way functions (functions easy to compute but hard to invert) exist, every method that runs in polynomial time gives up coverage (Section~\ref{sec:law}).

Levin \cite{levin} showed that universal search (Levin search) solves every search problem in time optimal up to a constant factor and an additive term linear in the input length. In the standard analysis of this search (Li and Vit\'anyi \cite{livitanyi}, Chapter 7), it finds a generating program within a constant factor of $2^{|p|}\,t(p)$ steps, for every program $p$ of length $|p|$ bits that generates the sequence in $t(p)$ steps, and the time-bounded complexity $\Kt$ is the least value of $|p| + \log t(p)$ over the programs that generate the sequence. The search therefore takes at most a constant times $2^{\Kt}$ steps. Levin's bound is the canonical upper bound. Like Levin search, Hutter's algorithm \cite{hutter} gives an upper bound (SI Appendix~D).

Our new result is the lower bound. Corollary~4 of SI Appendix~A applies to every algorithm that reads only a candidate's score and to every reformulation that meets the conditions of Section~\ref{sec:law}. It shows that, counted in coupled fields instead of bits, the structure that the reformulation injects plus the logarithm of the remaining search is at least the coupling width, up to a logarithmic term. In this count, each field's worth of structure removes at most one field's worth of search. The closest earlier lower bounds concern different objects. One is the inversion of random functions priced in advice bits (information computed before any query) rather than in coupling width (Yao \cite{yao}). Corrigan-Gibbs and Kogan \cite{corrigangibbs} showed barriers to improving that bound. The other is the cost of computing $\Kt$ itself rather than of program search (Brandt \cite{brandt}). Neither gives a conserved quantity for score-oracle search over structured program spaces.

\begin{figure*}[t]
\centering
\includegraphics[width=\linewidth]{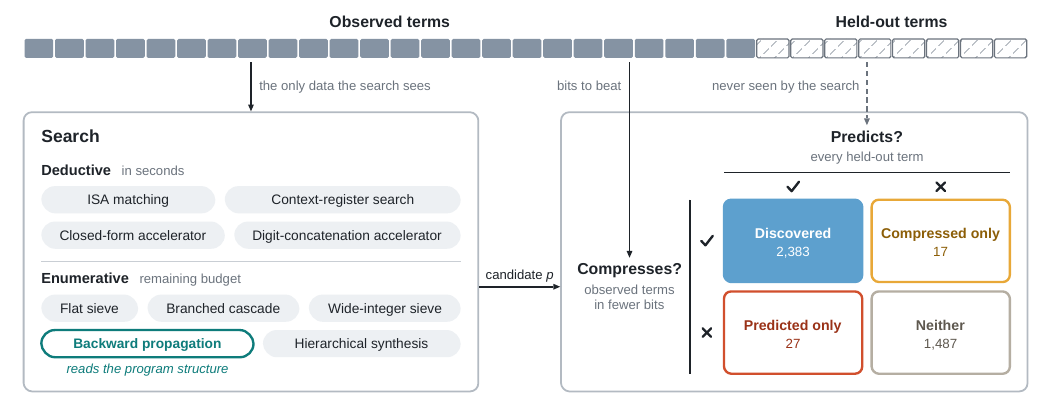}
\caption{The OMNIS search pipeline (left) and the gate (right). The observed terms of a sequence are the only data the search sees; the held-out terms are set aside before the search starts and are read only by the prediction test. A deductive stage completes within about 30\,s. It runs ISA matching, which tests single instructions of the instruction set architecture (ISA) against a table of the observed transitions; context-register search, over programs that compute each term from earlier ones; and the closed-form and digit-concatenation accelerators. An enumerative stage spends the remaining budget on the flat sieve, branched cascade, wide-integer sieve, backward propagation, and hierarchical synthesis, which is idle here because the library it composes programs from is empty. Backward propagation (outlined) is the structural inference of Section~\ref{sec:law}. It builds the wiring of each operation sequence from its dataflow, so that only a bounded number of candidates remain to be tested. It finds only the programs it can build. The best candidate $p$ faces two separate tests, compression (its description length is below the bits of the observed terms) and prediction (the vote of the fitting candidates that the search records, weighted toward the shortest, reproduces every held-out term). The two answers place it in one cell of the table. Counts are over the 3,914 candidate sequences.}
\label{fig:arch}
\end{figure*}

\section{Coupling Width and the Conservation Law}\label{sec:law}

We formulate the search for a generating program as a constraint satisfaction problem. A candidate program has a fixed control structure whose slots hold operations, and its fields are the choices left open, such as the operation in each slot or the register each operand reads or writes. The fields are the variables, each drawn from a finite domain of $d$ values, and the constraints require the program to reproduce the target sequence. An assignment gives every field a value and so specifies one candidate program. The score of an assignment, its match score, is the number of target symbols it reproduces, so a perfect score marks a program that reproduces all of the observed data. A search algorithm proposes assignments and reads their scores. We ask how many such reads are required, and what the answer depends on.

The answer depends on one structural quantity, which we call the coupling width. Consider a search that begins from an assignment scoring poorly, the baseline. A set of fields is coupled if, starting from the baseline, no change to any proper subset of them alters the score, while some simultaneous change to all of them improves it. The coupling width, written $\kappa$, is the size of the largest such set. When $\kappa$ is small, the score can reveal progress a few fields at a time, and the bound below does not force the search to be hard. When $\kappa$ is large, the score is flat across every partial correction and moves only when all $\kappa$ fields are set correctly at once. The search is then in a flat basin, where nothing it can measure tells it which way to go.

From this definition we obtain the barrier. An algorithm in the score-oracle model accesses the problem only through the scores of the assignments it proposes, as do evolutionary search, simulated annealing and Markov chain Monte Carlo over program space, the cross-entropy method, and Levin search. For every such algorithm, the number of assignments it must score before it finds an improving one is, in the worst case,
\[
\Omega\!\left(\frac{(d-1)^{\kappa}}{m}\right),
\]
where $d$ is the domain size, $m$ is the number of assignments of the $\kappa$ coupled fields that improve the score, and $\Omega(x)$ means at least a constant multiple of $x$ (SI Appendix~A, Theorem~1 and Corollary~1).

The argument is adversarial. Until all $\kappa$ coupled fields differ from the baseline, every score equals the baseline score, so the algorithm learns nothing. Among the $(d-1)^{\kappa}$ assignments that change every coupled field, an adversary places the $m$ improving ones so that a deterministic strategy must score at least $(d-1)^{\kappa} - m + 1$ of them. By Yao's principle \cite{yao1977}, a randomized strategy still needs $\Omega((d-1)^{\kappa}/m)$ on average.

Concretely, for register inference, which chooses the register each operand reads or writes, the domain is the eight registers, $d = 8$, and the barrier grows as $7^{\kappa}$. For type discovery, which chooses the operation in each slot, the domain is the operation catalog of $53$ entries plus the baseline type, the identity operation that does nothing, so $d = 54$ and the barrier grows as $53^{\kappa}$ (SI Appendix, Table~S1).

The bound is the worst case over the scores consistent with a given coupling width (SI Appendix~A, Corollary~1). How large the width is depends on the problem and on the score. For type discovery, when the score records only whether some register assignment reproduces the whole target exactly, the width is at least the program's own minimal size, its count of nonidentity slots (SI Appendix~A, Theorem~7). Under the match score used above, the width can be smaller (SI Appendix~A, Theorem~11). For register inference the width is that of the register constraint problem, the quantity our reformulations reduce, and we prove a lower bound on it for a class of register programs (SI Appendix~A, Proposition~14).

In that class, registers $R_0$ and $R_1$ hold $0$ and $1$, and each of $k$ loop steps doubles an accumulator and subtracts the register named by its own field, so the target encodes a hidden binary number with one digit per field. With arithmetic modulo a large prime and a suitable starting value, every assignment of these fields but the correct one scores zero, the baseline included. The $k$ fields are therefore coupled, and a search by score over them needs at least $(2^k + 1)/2$ queries on average for a uniformly random target, while the first output term gives the hidden number after a few arithmetic operations per field.

A second class uses the same registers and operations, with $k$ unknown fields and a coupling width $\kappa$ that divides $k$. For a uniformly random target in this class, the least expected number of queries on these fields is ${1 + (k/\kappa)(2^\kappa - 1)/2}$. This number therefore rises from $1 + k/2$ to $(2^k + 1)/2$ as $\kappa$ rises from $1$ to $k$ (SI Appendix~A, Proposition~15). Up to a relabeling, a shift and a scale, the score of this second class on $\{R_0, R_1\}^k$ is the royal staircase fitness of van Nimwegen and Crutchfield's theory of epochal evolutionary search \cite{vannimwegen2000}, with block size $\kappa$ (SI Appendix~A, Remark~15).

The coupling width is bounded by the structure of the score,
\[
\kappa \le \twscore + 1,
\]
where $\twscore$ is the treewidth of the interaction graph of the score, which joins two variables when changing one alters how changing the other affects the score (SI Appendix~A, Theorem~2). Treewidth measures how far a graph is from a tree, and a constraint problem can be solved exactly in time exponential only in the treewidth of its constraint graph \cite{dechterpearl}. The bound ties search difficulty to this familiar graph parameter, and it guarantees that the coordinated change a search must discover is never wider than this treewidth plus one. When the score counts satisfied constraints, as the match score of program discovery does, two variables interact only if some constraint contains both. The bound then also holds with the treewidth of the constraint graph itself (SI Appendix~A, Lemma~3). Lowering this treewidth to ease exact inference therefore also lowers the ceiling on the coupling width.

The coupling width can be much smaller than this bound. Some problems have a constraint graph of large treewidth and a small coupling width, and, for scores that do not count constraints, others show the reverse (SI Appendix~A, Proposition~3). Coupling width is therefore a different quantity from treewidth. It is the part of the structure that the score exposes to search.

A search need not pay the full exponential. It can instead spend effort reformulating the problem, reencoding it so that fewer fields are coupled. Each bit of structure it injects this way removes at most one bit of coupled search, that is, each such bit can at most halve the number of candidates that must still be scored. We prove this under conditions on the reformulation and on the target. The reformulation has polynomial size and preserves the score, its coupling structure and its single improving assignment, and no other assignment improves the score of the reencoded problem. The target is scored only by whether a candidate reproduces it exactly. It has a single generating program as complex as the search space, i.e., describing that program takes as many bits as picking one candidate out of that space.

For type discovery on a target sequence $s$ of length $N$, if the program's control structure and register assignment cost only $O(\log N)$ bits to describe, the cost of solving obeys the conservation law (SI Appendix~A, Theorem~6 and Remark~9),
\[
K(\varphi, \psi) + \log_2(\mathrm{search}) \ge \frac{\log_2(d-1)}{\log_2 d}\,K(s) - O(\log N),
\]
where $K(\varphi, \psi)$ is the Kolmogorov complexity of the reformulation (the pair of maps that translate assignments between the original and the reencoded problem), $\mathrm{search}$ is the worst-case number of candidates that must still be scored, and $K(s)$ is the Kolmogorov complexity of $s$.
In words, the structure put into a search plus the search that remains, both counted in bits, costs at least the length of the shortest program for the sequence, up to one computable factor. That factor, $\log_2(d-1)/\log_2 d$, is $0.995$ for type discovery ($d = 54$) and approaches $1$ as the catalog grows. With no structure injected, as in Levin search, the bound on search then approaches $2^{K(s)}$, the factor $2^{|p|}$ of Levin's bound at a shortest generating program (SI Appendix~A, Remarks~9 and~16). At the other extreme, a reformulation that fully exposes the program's structure pays in knowledge and leaves only polynomial search.

Without the assumption on control structure and register assignment, the law holds in units of coupling width. Counted in units of $\log_2(d-1)$ bits, the structure plus the logarithm of the remaining search is at least the coupling width, up to a logarithmic term, so each unit of structure removes at most one variable's worth of search (SI Appendix~A, Corollary~4).

The conservation law also explains why passing the gate of the next section is not cheap. For score-oracle search, the worst-case cost in bits of finding a program that both compresses and predicts is at least the Kolmogorov complexity of the sequence, up to the factor above.

Two points set the scope of the bound. First, the lower bound is about search by score, and our system also contains a method that infers structure instead. Backward propagation (Materials and Methods, and SI Appendix~A) builds the wiring of each operation sequence from its dataflow without consulting the score, and only then tests the bounded number of remaining candidates against the target.

Search by score pays in time, at a cost exponential in $\kappa$. Backward propagation pays in coverage. It infers structure in time polynomial in the sequence length but solves only the instances its representation can reach (SI Appendix~A, Theorem~18 and Remark~22). We prove it incomplete on generic targets, where it misses almost all the sequences a fixed program structure can produce when enough of its register wiring is free (SI Appendix~A, Theorem~17 and Corollary~12). If one-way functions exist, every method that runs in polynomial time pays in coverage. Each such method fails to find a short, fast program for almost all of a family of targets that each have one (SI Appendix~A, Proposition~16). Our system uses both kinds of search, so its failures fall in identifiable places.

Second, the bound is a worst case. A target whose score gives a usable gradient toward the optimum is easier, and our system exploits such gradients. The theorem therefore describes the hardest case, and the system works below it. The experiments first test the gate itself (Section~\ref{sec:test}) and then two consequences of the bound under a fixed time budget, a sharp edge between fast discoveries and timeouts (Section~\ref{sec:cliff}) and failures where no short program is within reach (Section~\ref{sec:scale}).

\section{Results}\label{sec:results}

\subsection{The System and the Corpus}\label{sec:corpus}

We ran the system on 3,914 candidate sequences in four populations. The first is a mixed set of integer sequences from the On-Line Encyclopedia of Integer Sequences (OEIS) together with the step counts of generalized Collatz maps. These maps halve an even number and send an odd number to a fixed odd multiple of it plus a fixed odd constant. The second is the 256 elementary cellular automata, the two-state rules on a line in which a cell's next state depends on itself and its two neighbors. The third is a larger OEIS set, the sequences tagged core or base, the latter meaning dependent on the number base. The fourth is the 2,187 three-state totalistic cellular automata, in which a cell's next state depends only on the sum of its neighborhood.

Each OEIS sequence enters at two or more of the alphabet sizes 2, 3, 4, 5 and 7, reduced modulo the alphabet size. Each entry, like each automaton rule, is one candidate sequence. Different candidate sequences can coincide, mostly among the automata (Section~\ref{sec:scale}), so the 3,914 hold 2,095 distinct targets. The same engine ran on all four, with an empty library and the same parameters for every candidate sequence. It recovered a generating program for 2,383 of the 3,914 candidate sequences (61\%), covering 606 of the 2,095 distinct targets.

Every discovery we report is obtained under four conditions. First (C1), it must compress the observed terms and predict an isolated held-out block. Second (C2), every program is written in one universal instruction set or in one of three fixed control forms (Materials and Methods), the same for every sequence. Third (C3), the library is empty, so no program may call another and no knowledge is carried between candidate sequences. Fourth (C4), no outside knowledge is imported, whether as a closed-form fitter, a recurrence solver, a learned policy, or programs found for other sequences. The deductive stage of OMNIS tests candidates from fixed families against the target. It imports no outside solver, and the programs it returns use the same instructions and control forms, the same description length and gate as every other program.

To our knowledge no prior system reports results under all four. LODA \cite{loda} mines the OEIS at the scale of 150,000 programs, but its programs may call previously discovered programs and part of its search mutates existing ones. The search of Gauthier and Urban \cite{gauthier2023} works in a fixed language of basic operators, but trains its policy on the programs found for other sequences and accepts a program that reproduces the given terms. Both relax C1, C3, and C4. The block decomposition methods of Zenil and collaborators \cite{zenil2018} estimate complexity from precomputed libraries rather than recover a single generating program. None of the three requires an isolated prediction, and the first two carry knowledge across sequences, so their coverage figures and ours measure different things (SI Appendix, Table~S13).

\subsection{Compression and Prediction Are Logically Independent}\label{sec:test}

A discovery, in this work, is a program that passes the two tests of the gate at once, compression and prediction (Fig.~\ref{fig:arch} and Materials and Methods). The obvious worry is that this definition is circular. If discovery is defined by passing a test, then reporting that our discoveries pass the test says nothing.

It would say nothing if the two conditions were the same condition. They are logically independent, and in the corpus they come apart. Of 3,914 candidate sequences, 2,383 both compress and predict, 27 predict their held-out continuation without compressing, and 17 compress without predicting it (Fig.~\ref{fig:landscape}). All 2,383 discoveries satisfy both conditions, but that holds by construction, since we label a program discovered only when both hold. If the two conditions were one, the off-diagonal cells of the table in Fig.~\ref{fig:arch} would be empty.

Passing the gate by luck is unlikely. If the held-out terms were uniformly random and independent of the observed ones, the gate would accept with probability at most two to the power of minus their number of bits. For every discovery this probability is at most $2^{-20}$ (SI Appendix~B, Proposition~17). Summed over all 3,914 candidate sequences, the expected number of such chance acceptances is below $4 \times 10^{-4}$. As a control, we shuffled the terms of 321 sequences from the four populations and ran the engine on them at the same budget (SI Appendix, Table~S2). The gate accepted none of the 201 shuffled sequences of high entropy, where shuffling destroys the generating order. This zero gives a one-sided 95\% upper bound of 1.48\% on the rate of false acceptances. Of the other 120, the gate accepted six, all nearly constant or strongly biased, which keep the structure of a short program.

\begin{figure}[t]
\centering
\includegraphics[width=\linewidth]{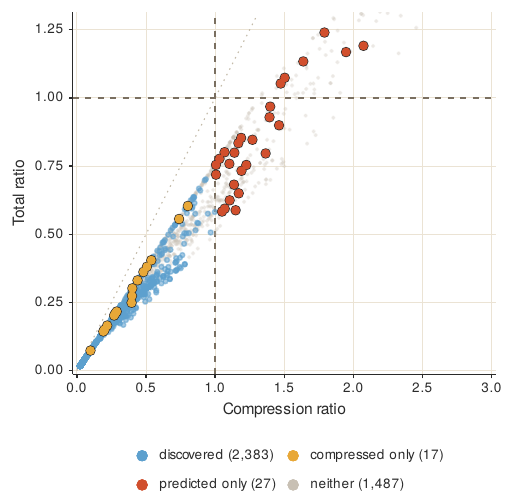}
\caption{Compression and prediction can disagree. Each point is one of the 3,914 candidate sequences of the four populations, colored by gate outcome and placed by two ratios. The compression ratio is the best program's description length over the bits of the observed terms, and the total ratio is that length over the bits of the whole sequence, held-out block included. The dashed lines mark the compression boundary and the total boundary, where these ratios equal 1, and the dotted diagonal marks equal ratios. A program compresses only if it also reproduces every observed term, so the gray points left of the compression boundary are partial fits that do not compress. The classes compressed only (17) and predicted only (27) are the cases where the two gate conditions disagree. The 37 failures whose ratios lie beyond the axes are not shown.}
\label{fig:landscape}
\end{figure}

\subsection{Discovery Times and Timeouts Separate}\label{sec:cliff}

If finding a program costs search exponential in the coupling width, then a search under a fixed time budget should show a sharp edge. Below some width the program should be found quickly, above it the program should never be found at all, and almost nothing should land in between. The data show this pattern. Across all four populations, discovery times and timeouts separate into two clouds with a band between them that holds no discovery and only six failures (Fig.~\ref{fig:cliff}). The slowest discovery takes 378\,s in population one, 587\,s in population two, 314\,s in population three and 497\,s in population four, while the failures pile up against the wall at 600\,s.

In the two cellular-automaton populations, 2,443 candidate sequences combined, not one candidate sequence, discovered or failed, occupies the gap between the last discovery and the wall. In the two OEIS populations, six candidate sequences in total fall in the gap. Four are near misses that predict without compressing (three divisor sums found as nested-loop programs by backward propagation, and the divisor-counting sequence). The other two fail both conditions.

The barrier predicts this cliff. A search whose cost rose gradually with program length would spread discoveries up to the budget, while under an exponential barrier the slowest discoveries stop short of it, as they do here. The time axis itself is shaped by the engine's collection window and phase deadlines (Section~\ref{sec:price} and Materials and Methods), so we make no claim about the absolute width of the gap. Two findings are robust. No discovery reaches the budget, and the narrowest margin is the 13\,s between the slowest discovery among the elementary automata and the wall. Only these six candidate sequences lie between the last discovery and the wall.

\begin{figure}[t]
\centering
\includegraphics[width=\linewidth]{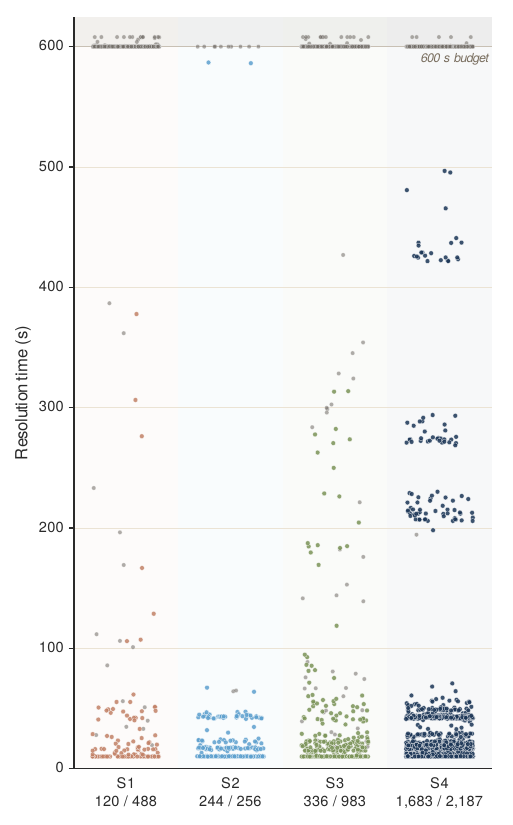}
\caption{The cliff between discoveries and timeouts. Resolution time for every candidate sequence in each of the four populations. The labels S1 to S4 mark the mixed set of OEIS sequences and Collatz step counts, the elementary cellular automata, the larger OEIS set, and the three-state totalistic cellular automata, in the order introduced in Section~\ref{sec:corpus}. Colored points are discoveries and gray points are failures. The count below each label is discoveries over population size. Discovery times and timeout failures separate into a lower cloud and a wall at the 600\,s budget. Only six failures lie between the slowest discovery of each population and the wall.}
\label{fig:cliff}
\end{figure}

\subsection{Discovery Time Rises with Program Length}\label{sec:price}

Discovery time rises with program description length across the 2,383 discoveries (Spearman correlation $0.76$), and the rise holds in all four populations (SI Appendix, Table~S11). It is mostly a rise between solver families, the classes of program that the strategies return, rather than within them (SI Appendix, Fig.~S3). The deductive families find short programs early, more than half of them within $15$\,s. The backward propagation family, which runs after the sieves and the branched cascade, finds the longest programs, up to $117$ bits, after a few minutes. Within each of the four largest families, time depends only weakly on length ($R^2$ below $0.2$). The rise thus reflects the order in which the strategies run.

Wall time also has a floor, the collection window of at least $10$\,s that follows the first full match, and a ceiling, the $600$\,s budget. Between these limits $\log_2$ time spans fewer than six doublings while program length spans $110$ bits, so wall time cannot show a cost that doubles with each added bit. We therefore do not read the slope as the engine's search cost (SI Appendix~B).

\subsection{Prediction without Compression Comes from Short Data}\label{sec:offdiag}

The cases of Section~\ref{sec:test} that predict without compressing follow a simple mechanism. Each has a correct program that is no shorter than the data it explains, which happens when the data are short. All 27 such cases come from the two OEIS populations; the cellular-automaton populations produce none, because their 375 observed terms carry more bits than any program found for them (SI Appendix, Table~S8).

The divisor-counting sequence, whose $n$th term is the number of divisors of $n$, is the sharpest illustration. At the binary alphabet, where each term is reduced modulo 2, the wide-integer sieve finds a program that predicts every held-out term and misses compression by 0.45 bits. At larger alphabets the same 104 terms carry more bits, and a nested-loop program of 97.7 bits compresses them, so the sequence is recovered outright. What changed across the alphabets was whether the data carried enough bits to pay for the program.

\subsection{Discovery Rates Fall as the Alphabet Grows}\label{sec:alphabet}

Reducing a sequence modulo a larger alphabet keeps more of its information, so the program that generates the reduced sequence tends to be longer and, under the barrier, harder to find. We see this effect on the same sequences encoded at each alphabet, the base keyword set of the third population. On these sequences, discovery rates fall from 54\% at the binary alphabet to 33\% at ternary to 27\% at quaternary. This decay is in the direction the exponential form predicts. It does not isolate the barrier, however, because changing the alphabet also changes the coupling width, by an amount that depends on the representation. For the 26 sequences of the first population that the OEIS tags hard, recovery falls from 27\% at the binary alphabet to none at either larger alphabet. All seven binary recoveries are lost at both larger alphabets (exact McNemar test, $p \approx 0.016$).

The decay is monotone only on average. About 11\% of the third population's multialphabet sequences, 30 of 261, invert the trend and become easier at a larger alphabet. Most of these inversions are digit sequences from the base keyword set, where an alphabet can align with the base in which the sequence is written. Several are near misses at the binary alphabet that compress once each term carries more bits (Section~\ref{sec:offdiag}). We therefore do not extract a single exponent from the alphabet data. The trend agrees with the barrier, and the inversions show where the representation matters more than the difficulty.

\subsection{Success Depends on Reaching a Short Program}\label{sec:scale}

The bound also predicts that failures fall where no short program is within reach. The best-known test is the elementary cellular automata, each presented as the center column of its evolution from a single live cell. The engine recovers a program for 244 of the 256 rules from raw output, with an empty library, inside the time budget (SI Appendix, Fig.~S1). The 256 rules give only 35 distinct sequences, because a rule and its mirror image always give the same column and many rules give a trivial one. For example, the column of Rule 110 is all ones, since its pattern grows only to the left of the seed. The 244 recovered rules cover 29 of the 35. The hardest rule recovered is Rule 30, whose column is chaotic. The engine recovers it with a program of 74.5 bits that updates a word of 512 bits by bitwise operations and reads one of its bits. The prediction for Rule 30 reproduces all 125 held-out terms, which a uniformly random continuation would match with probability at most $2^{-125}$ (SI Appendix~B).

The 12 failures are six sequences, each shared by a mirror pair. Two are chaotic and unreached. Two come from the complements of Rules 30 and 110, whose columns are not relabelings of their parents' because the seed is not complemented with the rule. The last two are simple, one compressing without predicting and one running out of time (SI Appendix, Fig.~S1). Success therefore depends on whether the search reaches a short program, and complexity alone does not decide it.

The totalistic automata show the same pattern. On the 2,187 three-state totalistic cellular automata, the engine recovers 1,683 (77\%), covering 180 of the 666 distinct sequences these rules produce (SI Appendix, Fig.~S2). The failures are almost entirely all-or-nothing. A rule either reduces to a short program or produces output the engine cannot compress at all, with very few partial fits in between. On OEIS sequences, which are shorter than the automata, the failures often reproduce 30\% to 70\% of the observed terms before the budget runs out, while on chaotic rules almost no failure is a partial fit. We record partial fits as failures.

The engine compresses the output it is given at the length it is given. For a rule like Rule 30 this output is finite, and recovering it makes no claim about the rule's behavior over infinite time. At each alphabet the engine recovers the shortest description it finds for the projected signal, the sequence reduced modulo the alphabet size. The same sequence is often recovered by different solver families at different alphabets (SI Appendix~C). For sequences tied to deep open questions, such as the exponents of the Mersenne and Wagstaff primes, the engine recovers only the binary projection, which is nearly constant, and fails at every larger alphabet.

\section{Discussion}

We have shown that, for the search methods that learn only from a candidate's score, the worst-case cost of finding a program is at least exponential in the coupling width and, under the conditions of Section~\ref{sec:law}, at least exponential in the length of the shortest program for the sequence. Structure injected into a search trades against the remaining search at best one for one. With no structure injected and a growing instruction set, the second bound approaches $2^{K(s)}$, the factor $2^{|p|}$ of Levin's upper bound at a shortest generating program.

Three other lines of work also study the cost of search. Bounds for constraint satisfaction (Razgon \cite{razgon}; Traxler \cite{traxler}; Lokshtanov, Marx, and Saurabh \cite{lokshtanov}; Marx \cite{marx2010, marx2013}) are either upper bounds or conditional on the Exponential Time Hypothesis, the conjecture that 3-SAT on $n$ variables cannot be solved in time $2^{o(n)}$, or on its strong form. All are parameterized by the explicit constraints, while ours is unconditional for algorithms that read only the score and is governed by the score itself. Black-box complexity (Droste, Jansen, and Wegener \cite{djw2006}; Lehre and Witt \cite{lehrewitt2012}) restricts an algorithm to querying the value of its objective, as we restrict it to the score. The barrier of Section~\ref{sec:law} is a query bound of that kind. For a single improving assignment it reduces to the bound of Droste, Jansen, and Wegener for finding a needle, and it extends that bound to coupled fields with any number of improving assignments. Measures that read difficulty off the fitness landscape (Kauffman \cite{kauffman1993}; Davidor \cite{davidor1991}; Jones and Forrest \cite{jonesforrest1995}) are descriptive and admit counterexamples \cite{he2007}. None of these lines of work defines a quantity that trades search against reformulation. Coupling width does so through a proven lower bound and a conservation law that prices reformulation in bits. SI Appendix~D compares the bounds one by one.

The system puts this trade into practice. Each of its search strategies trades structural knowledge of a region of program space for search, and the gate certifies each result as a generator. The cliff between discoveries and timeouts agrees with the barrier.

The score counts only exact matches, and the system works on deterministic integer sequences. Extending it to noisy data means replacing exact matching with a statistical threshold. The barrier holds for every scalar score, but we prove the conservation law here for the score that records whether the whole target is reproduced exactly. The coverage rates are properties of one instruction set and one budget, and a different machine would draw the boundary in a different place.

We have also distinguished search, which must pay the exponential to be complete, from structural inference, which need not pay it and gives up coverage instead. We prove this price for backward propagation on generic targets and, if one-way functions exist, for every method that runs in polynomial time (SI Appendix~A, Theorem~17 and Proposition~16). Since the engine orders programs by description length, it favors targets with short descriptions. Whether backward propagation is also incomplete on these targets is open.

\matmethods{OMNIS is a deterministic search over programs written for a single register machine. The system does not learn, uses no randomness, and carries no state from one sequence to the next, so no result comes from training. The phases stop at wall-clock deadlines, so a target solved only late in a phase can end differently on another run (SI Appendix~B).

The instruction set consists of 18 canonical operations (arithmetic, bitwise, memory, and a single loop) and one further operation that calls a stored program as a subroutine. We hold that operation unused in the experiments here, so that no program can draw on knowledge of another. A program is a short sequence of these instructions over at most eight registers, and the interpreter runs it in one of four modes that differ only in how the program is driven (SI Appendix~B).

Three further program types are fixed control forms with a few parameters. A digit recursion (DARY) applies one instruction to the digits of the index in a chosen base, a digit concatenation (CONCAT) writes the indices in a chosen base one after another, and a step count (STEP) counts the iterations of a branched body until it halts. The description length charges the first two for their type tag and parameters, the instruction of a digit recursion included.

The description length comes from one code over all program types, built from a universal integer code with a short tag that selects the subdecoder, so that the lengths of different program types are directly comparable. It is computed exactly for each candidate and is an upper bound on the Kolmogorov complexity of the sequence up to $O(\log N)$ bits (SI Appendix~B).

Search proceeds through the nine enumeration strategies of Fig.~\ref{fig:arch}, four deductive subphases and five enumerative phases. The Phase~0 library lookup is not counted among them. Hierarchical synthesis and the lookup both draw on the program library, so in these empty-library runs they stay idle. ISA matching contributes no reported program, and every reported program comes from the remaining seven strategies.

Each strategy covers a different region of program space (closed-form arithmetic, context-register programs, digit-concatenation forms, nested-loop programs, and several smaller families). Each replaces search with the structural knowledge that defines its region, the kind of trade that Section~\ref{sec:law} bounds for reformulations. Backward propagation, which builds the wiring of each program from its dataflow, recovers the nested-loop programs, and its limit is the incompleteness proved in SI Appendix~A, Theorem~17. The other strategies inject some structure and search the rest by score.

During the search a candidate's score is the number of leading observed terms it reproduces before its first mismatch, a prefix form of the match score of Section~\ref{sec:law}. The barrier holds for every score.

Each candidate sequence is searched once, with a budget of $600$\,s of wall time. Three phases also stop at earlier deadlines of their own (SI Appendix~B). Once a program matches all the observed data, the search continues for a collection window of at least $10$\,s to look for a shorter program, so no recorded time falls below $10$\,s. No discovered program exceeded its budget (SI Appendix, Table~S11).

We define a discovery by the gate, which has two conditions. A candidate compresses if it reproduces every term of the observed data and, under the description length above, encodes those data in strictly fewer bits than the data themselves take. It predicts if a vote of the fitting candidates that the search records reproduces every term of a held-out continuation that the search never saw. Each candidate in the vote is weighted by two to the power of minus its length, as in Solomonoff's mixture. A discovery requires both. The gate is one-sided by design. It refuses a correct program that is longer than its data and records it as a near miss, because a false refusal costs a missed result while a false acceptance costs a wrong claim. Section~\ref{sec:test} shows why the two conditions are distinct.

The held-out block is set aside before the search begins and isolated from it. It holds the larger of 20 terms and a quarter of the available terms. For the context-register programs it is generated autoregressively, each program reading its own output forward, so that a single early error propagates and the block fails. We fixed this size rule in advance for every candidate sequence and never tuned it to outcomes. We developed the engine on 14 reference targets. Some of them also occur in the corpus, including Rule 30, the Fibonacci sequence and the divisor functions. Every candidate sequence is searched with the same parameters.

Before the search, we excluded the sequences of length below 40, which are too short to admit a meaningful split into observed terms and a held-out block. SI Appendix~B gives the criterion, the per-population counts, and the raw pool from which the analyzed 3,914 are drawn. The OEIS sequences of the corpus have 40 to 133 terms, the Collatz step counts 200, and the cellular automata 500. The full instruction set, the description-length code, the four modes, the nine strategies, and the construction of each population are described in SI Appendix~B and specified exactly by the deposited source code. The sweeps ran one candidate sequence at a time on one Apple M2 computer (SI Appendix~B).

\subsection*{Use of AI Tools} Claude (Anthropic), principally Claude Opus 4.6 and Claude Opus 5.5 used through Claude Code, assisted with writing and testing the OMNIS source code and the analysis scripts, with developing and checking proofs and definitions, with editing the manuscript, with producing figures, and with checking the references. No other AI tool was used. The author reviewed and verified all AI-assisted output and takes full responsibility for the content.

\subsection*{Data, Materials, and Software Availability} The OMNIS engine source and the complete sweep data for all four populations \cite{omnis2026} are deposited and available at \href{https://github.com/jorgeMFS/omnis}{github.com/jorgeMFS/omnis} and permanently archived at Zenodo (\href{https://doi.org/10.5281/zenodo.23863377}{doi.org/10.5281/zenodo.23863377}).\par}

\showmatmethods{}

\acknow{This work was supported by the Foundation for Science and Technology (FCT) through contract UID/00127/2025 (\href{https://doi.org/10.54499/UID/00127/2025}{doi.org/10.54499/UID/00127/2025}).}
\showacknow{}

\bibliography{paper1}

\end{document}


\maketitle

\SItext

This Supporting Information accompanies the main text. It contains four
appendices:
\begin{itemize}[leftmargin=2em]
\item SI Appendix A, the Coupling Barrier Theorem and wire-space register inference.
\item SI Appendix B, the OMNIS architecture, full specification.
\item SI Appendix C, sweep data tables.
\item SI Appendix D, prior bounds and prior systems (D.1, prior lower
bounds and how the coupling barrier differs; D.2, comparison with prior
program-discovery systems).
\end{itemize}

The reference list for this Supporting Information is separate from the
reference list of the main text.

\bigskip

\section*{SI Appendix A. The Coupling Barrier Theorem and Wire-Space Register Inference}

This appendix develops the coupling barrier, the reformulation and conservation
analysis, and the register-machine instantiation, in which an actual program has the flat
score the lower bound assumes. It then analyzes wire-space
backward propagation (WSBP), the register inference algorithm that the main text calls backward propagation, and proves its
soundness, polynomial running time, and incompleteness. The last section proves the
Search-Inference Coupling (\cref{thm:coupling-capstone}).

\clearpage
\subsection*{Notation}

\begin{center}
\small
\begin{tabular}{@{}ll@{}}
\toprule
Symbol & Meaning \\
\midrule
\multicolumn{2}{@{}l}{\emph{General constraint satisfaction problem (CSP)}}\\
$X, D, C$ & variables, domain, constraints of a CSP \\
$n, d, m$ & $|X|$, $|D|$, $|C|$ \\
$\mathcal{P},\ \mathrm{fw}(\mathcal{P})$ & a CSP instance; its framework data (variables, domain, baseline) \\
$\alpha_0,\ \vbase$ & baseline assignment; $\vbase = \score(\alpha_0)$ \\
$S,\ k$ & coupled set of variables; its size \\
$\kappa$ & coupling width (rigid; relaxed $\kappa(\delta,\Delta)$) \\
$\delta, \Delta$ & coupling parameters (partial bound; signal threshold); $\delta_0, \Delta_0$ those of \Cref{thm:k-growth} \\
$Q$ & witness cube $\{\beta \in D^k : \beta_j \ne \alpha_0(x_{i_j})\ \forall j\}$, size $(d-1)^k$ \\
$\msig$ & number of signal witnesses in $Q$ \\
$\mathcal{A}(S,\msig),\ \mathcal{A}_{\delta,\Delta}(S,\msig)$ & score functions with $S$ coupled (rigid; relaxed) and $\msig$ witnesses \\
$\tw;\ \Gscore,\ \twscore$ & treewidth of the primal graph; score-interaction graph and its treewidth \\
$\varphi, \psi;\ \tilde S, k'$ & reformulation maps; coupled set of the reformulated CSP, its size \\
$B,\ \operatorname{search}$ & coupling budget $\kappa \log_2(d-1) = \log_2 |Q|$; worst-case number of score queries \\
$K(\cdot)$ & Kolmogorov complexity on a fixed universal prefix machine \\
\addlinespace
\multicolumn{2}{@{}l}{\emph{Register machine}}\\
$F^*, \tau^*, \rho^*$ & control frame; ground-truth type assignment; register assignment \\
$\exptr(P, N)$ & first $N$ output symbols of $P$, also written $\exptr(F^*, \tau, \rho, N)$ \\
$T,\ W$ & catalog size $|D|=T$; number of physical registers $|R|=W$ \\
$\Knt$ & syntactic minimality (min nonidentity slots generating $s$) \\
$N,\ L$ & target length; number of operation slots (leaves of the frame) \\
$[N]$ & output positions $\{0, \dots, N-1\}$ \\
$d^*$ & distinguishing depth of register assignments (SI Appendix A, Approach A) \\
$P_0$ & baseline program, $(F, \tau^*, \rho_0)$ for register inference and $(F^*, \tau_0, \rho^*)$ for type discovery \\
$s \in \Sigma^N,\ \Sigma$ & target sequence; output alphabet \\
$y \in \Sigma^{n_y}$ & held-out block of the gate \\
$\mu(s)$ & maximum single-symbol frequency in $s$ \\
$p,\ \xi_j$ & prime modulus and shifts of the $\mathbb{Z}/p$ class (\Cref{prop:zp}) \\
\addlinespace
\multicolumn{2}{@{}l}{\emph{Structural and dynamical coupling}}\\
$\mathcal{D}(P_0,\tau)$ & typed dataflow graph \\
$\mathrm{dep}(t)$ & type-conservative dependency set at iteration $t$ \\
$\varepsilon,\ \gamma(N)$ & essential-engagement parameter; adversarial-decorrelation residual \\
$\kappa_{\mathrm{bin}},\ \kappa_{\mathrm{match}}$ & binary; match-score coupling width (superscripts $\rho^*$, $\max_\rho$, $(\delta,\Delta)$) \\
$\score_{\mathrm{match}}^{\rho^*},\ \score_{\mathrm{match}}$ & match count with $\rho^*$ fixed; with the best $\rho$ for each typing \\
$\mathrm{mc},\ \mathrm{mc}_{\mathrm{struct}},\ \mathrm{coinc}$ & match count; its structural part; its coincidental part \\
$c_T$ & Levin base relative to the instruction set architecture (ISA), $(T-1)^{1/\log_2 T} < 2$ \\
$M(s)$ & Solomonoff algorithmic probability \\
\addlinespace
\multicolumn{2}{@{}l}{\emph{Wire-space backward propagation (WSBP)}}\\
$\mathcal{T} = (\eta, U)$ & wire topology (port-to-wire map; wire equivalence) \\
$J,\ b,\ n_{\mathrm{nc}},\ n_c$ & loop-body operations; \textsc{divmod}$_C$ operations; noncommutative operations; wire classes \\
$\theta,\ u,\ o,\ \iota$ & routing; demand resolution; operand ordering; preloop initializer \\
$N_U,\ N_P,\ N_{\mathrm{cand}}$ & numbers of demand resolutions, of initializers, of realized candidates \\
$[\rho],\ \mathcal{B}_\rho$ & functional class of $\rho$; classes WSBP can realize \\
$\mathcal{F}(F^*,\tau^*),\ h$ & distinct output sequences the frame realizes; one of them \\
$j;\ \lambda,\ \nu,\ H_\infty$ & independently free ports; a measure on $\mathcal{F}$, the uniform one, min-entropy \\
\bottomrule
\end{tabular}
\end{center}

\section{Coupling Width and the Coupling Barrier}

We define the objects of the theory (the CSP, the score, and the
coupled set) and prove the coupling barrier, the unconditional lower bound on
score-oracle search that the rest of the appendix prices and instantiates.

\subsection{Definitions}

\begin{definition}[Constraint satisfaction problem]\label{def:csp}
A \emph{constraint satisfaction problem} (CSP) is a triple $(X, D, C)$ where:
\begin{enumerate}[label=(\roman*)]
\item $X = \{x_1, \dots, x_n\}$ is a finite set of variables.
\item $D$ is a finite domain with $|D| = d$. Each variable takes values in $D$.
\item $C = \{C_1, \dots, C_m\}$ is a finite set of constraints. Each constraint
$C_j$ has a scope $S_j \subseteq X$ and a relation $\mathrm{Rel}_j \subseteq D^{|S_j|}$.
An assignment $\alpha : X \to D$ satisfies $C_j$ if $\alpha|_{S_j} \in \mathrm{Rel}_j$.
\end{enumerate}
\end{definition}

\begin{definition}[Score function]\label{def:score}
For an assignment $\alpha : X \to D$, the score is
\[
  \score(\alpha) = \bigl|\{\, j \in \{1, \dots, m\} : \alpha \text{ satisfies } C_j \,\}\bigr|.
\]
An assignment $\alpha^*$ is \emph{satisfying} if $\score(\alpha^*) = m$.
\end{definition}

More generally, an instance may carry an arbitrary score function
$\score : D^X \to \mathbb{N}$ in place of the constraint count; every coupling
notion below is stated relative to a pair $(\score, \alpha_0)$ and applies
verbatim, with the constraint-count score as the canonical case. When the
explicit constraints play no role we write such an instance as a quadruple
$(X, D, \score, \alpha_0)$.

\begin{definition}[Primal graph]\label{def:primal}
The primal graph (or constraint graph) $G = (X, E)$ has an edge $(x_i, x_j)$ if and only if there
exists a constraint $C_k$ whose scope $S_k$ contains both $x_i$ and $x_j$.
\end{definition}

\begin{definition}[Coupling set]\label{def:coupling-set}
Let $\alpha_0$ be a baseline assignment, let $\vbase = \score(\alpha_0)$, and let
$S = \{x_{i_1}, \dots, x_{i_k}\}$ be a set of variables. Write $\restr{\alpha_0}{\beta}$ for the
assignment that agrees with $\alpha_0$ outside $S$ and assigns $\beta$ to the
variables in $S$.

The set $S$ is \emph{$k$-coupled} with
respect to $(\score, \alpha_0)$ if both conditions hold:
\begin{enumerate}[label=(\alph*)]
\item For every proper subset
$S' \subsetneq S$ with $0 < |S'| < k$, and for every assignment
$\beta : S' \to D$,
\[
  \score(\restrS{\alpha_0}{S'}{\beta}) = \vbase.
\]
\item There exists an assignment
$\beta^* : S \to D$ such that
\[
  \score(\restr{\alpha_0}{\beta^*}) > \vbase.
\]
\end{enumerate}
\end{definition}

\begin{definition}[$(\delta, \Delta)$-coupling set, relaxed]\label{def:coupling-set-relaxed}
Let $S = \{x_{i_1}, \dots, x_{i_k}\}$ be a set of variables, let $\alpha_0$, $\vbase$, and $\restr{\alpha_0}{\beta}$ be as in
\Cref{def:coupling-set}, and let $\delta, \Delta$ be nonnegative reals with
$\Delta > \delta$. The set $S$ is
\emph{$(\delta, \Delta)$-$k$-coupled} with respect to $(\score, \alpha_0)$ if both conditions hold:
\begin{enumerate}[label=(\alph*$'$)]
\item For every proper subset
$S' \subsetneq S$ with $0 < |S'| < k$ and every $\beta : S' \to D$,
\[
  \score(\restrS{\alpha_0}{S'}{\beta}) \le \vbase + \delta.
\]
\item There exists $\beta^* : S \to D$ such that
\[
  \score(\restr{\alpha_0}{\beta^*}) \ge \vbase + \Delta.
\]
\end{enumerate}
The gap $\Delta - \delta$ is the \emph{correction margin} of the coupled set.
\end{definition}

\begin{proposition}[Relationship between \Cref{def:coupling-set} and \Cref{def:coupling-set-relaxed}]\label{prop:def-equiv}
Let $S$ be a candidate coupled set with baseline $\alpha_0$ and an integer-valued
score function.
\begin{enumerate}[label=(\roman*)]
\item If $S$ is $k$-coupled in the
sense of \Cref{def:coupling-set}, then $S$ is $(0,1)$-$k$-coupled in the sense
of \Cref{def:coupling-set-relaxed}.
\item If $S$ is $(0,1)$-$k$-coupled and
$\alpha_0$ is a local score minimum on $S$-perturbations, i.e.,
\begin{equation}
  \score(\restrS{\alpha_0}{S'}{\beta}) \ge \vbase
  \quad\text{for every proper } S' \subsetneq S,\ 0<|S'|<k,\ \text{and every }
  \beta : S' \to D, \tag{LMH}\label{eq:lmh}
\end{equation}
then $S$ is $k$-coupled in the sense of \Cref{def:coupling-set}.
\end{enumerate}
\end{proposition}

\begin{proof}
(i) Condition (a) of \Cref{def:coupling-set} states
$\score(\restrS{\alpha_0}{S'}{\beta}) = \vbase$, which implies that
$\score(\restrS{\alpha_0}{S'}{\beta}) \le \vbase + 0$, giving (a$'$) at
$\delta = 0$. Condition (b) gives $\beta^*$ with
$\score(\restr{\alpha_0}{\beta^*}) > \vbase$; since the score is
integer-valued, this is equivalent to
$\score(\restr{\alpha_0}{\beta^*}) \ge \vbase + 1$, giving (b$'$) at
$\Delta = 1$.

(ii) Condition (a$'$) at $\delta = 0$ gives $\score \le \vbase$. Hypothesis
\textup{(\ref{eq:lmh})} gives $\score \ge \vbase$. Together these force
$\score = \vbase$, which is (a). Condition (b$'$) at $\Delta = 1$ gives
$\score \ge \vbase + 1 > \vbase$, which is (b).
\end{proof}

\begin{remark}
Hypothesis~\textup{(\ref{eq:lmh})} holds whenever $\vbase = 0$, since scores
are nonnegative. The case $\vbase = 0$ is the regime of the binary score of \cref{thm:binary-coupling}
(the all-\textsc{identity} baseline generates no
nontrivial target, so $\vbase^{\mathrm{bin}} = 0$). The hypothesis also holds whenever the
baseline attains the pointwise infimum of the score over $S$-perturbations, a
property of the graded-translation class of \cref{prop:zp}
(every proper perturbation achieves exactly $\vbase = \lfloor N/p \rfloor$
matches, so \textup{(\ref{eq:lmh})} holds with equality).
\end{remark}

\begin{definition}[Coupling width]\label{def:coupling-width}
The \emph{coupling width} of the CSP instance with respect to
$(\score, \alpha_0)$ is
\[
  \kappa = \max\{\, k : \text{there exists a $k$-coupled set with respect to }
  (\score, \alpha_0) \,\}.
\]
If no coupled set exists (the baseline is a local optimum of the score), then
$\kappa = 0$ by convention and the search lower bound of clause~(III) of \Cref{thm:barrier} is vacuous.
\end{definition}

\begin{definition}[$(\delta, \Delta)$-coupling width]\label{def:coupling-width-relaxed}
The \emph{$(\delta, \Delta)$-coupling width} is
\[
  \kappa(\delta, \Delta)
  = \max\{\, k : \text{there exists a $(\delta, \Delta)$-$k$-coupled set
  with respect to } (\score, \alpha_0) \,\}.
\]
\end{definition}

\begin{proposition}[Computing the coupling width]\label{prop:kappa-hard}
When the score is given as a Boolean circuit, deciding whether $\kappa \ge k$ is
$\Sigma_2^{\mathrm{P}}$-complete, already for $d = 3$, and NP-hard already for
$k = 1$. Deciding whether a $k$-coupled set exists takes at most $|X|^k d^k$ score
evaluations.
\end{proposition}
\begin{proof}
By \Cref{def:coupling-set}, $S$ is coupled when some assignment to $S$ scores
above $\vbase$ while every assignment to $S$ that agrees with $\alpha_0$ on some
variable of $S$ scores $\vbase$; the improving assignment therefore differs from
$\alpha_0$ everywhere on $S$. So $\kappa \ge k$ is an $\exists\forall$ condition
over a set of at least $k$ variables and two assignments, checked by the circuit,
and lies in $\Sigma_2^{\mathrm{P}}$. Trying every set of size $k$ and every
assignment to it gives the count.

Hardness is by reduction from deciding $\exists u\, \forall v\, \psi(u, v)$ for a
Boolean formula $\psi$ with $u$ nonempty, which is $\Sigma_2^{\mathrm{P}}$-complete
(Stockmeyer, 1976). Take $D = \{0, 1, 2\}$ and $\alpha_0 = 0$. For each $u_i$ take
two variables $a_i, b_i$, for each $v_j$ one variable $c_j$, and one more
variable $z$. Call an assignment \emph{proper} if exactly one of $a_i, b_i$ is
nonzero for every $i$ and every $c_j$ is nonzero; it encodes $u_i = 1$ when
$a_i \ne 0$ and $v_j = 1$ when $c_j = 2$. The score is $2$ on proper assignments
with $z \ne 0$, is $1$ on proper assignments with $z = 0$ whose $(u, v)$ falsifies
$\psi$, and is $0$ otherwise, so $\vbase = 0$. Let $k$ be the number of variables
$u_i$ and $v_j$ plus one.

A coupled set $S$ needs a proper assignment that is nonzero exactly on $S$, so
$S$ contains every $c_j$ and exactly one of $a_i, b_i$ for each $i$, which fixes
a value of $u$. If $|S| \ge k$, it also contains $z$, and its all-nonzero
assignments score $2$. An assignment to such an $S$ that agrees with $\alpha_0$
somewhere is proper only when the agreement is at $z$ alone, and then it scores
$1$ exactly when its $v$ falsifies $\psi(u, \cdot)$; otherwise it scores $0$.
Hence $S$ is coupled if and only if $\psi(u, v)$ holds for every $v$, and
$\kappa \ge k$ if and only if $\exists u\, \forall v\, \psi(u, v)$. For $k = 1$,
let the score of an assignment to the variables of a CNF formula $\varphi$ be $1$
when every variable is nonzero and $\varphi$ holds under $x = [x = 2]$, and $0$
otherwise. Then the only candidate coupled set is the whole set, and
$\kappa \ge 1$ if and only if $\varphi$ is satisfiable.
\end{proof}

\subsection{The Coupling Barrier Theorem}

The barrier has three clauses. Verification and inference are classical; the
content is clause~(III), the flat-basin reduction. 

\begin{theorem}[Coupling Barrier Theorem]\label{thm:barrier}
Let $(X, D, C)$ be a CSP with $|X| = n$, $|D| = d$, and $|C| = m$, and for a coupled set $S$ write
\[
  Q = \{\, \beta \in D^k : \beta_j \ne \alpha_0(x_{i_j}) \text{ for all } j \,\},
  \qquad |Q| = (d-1)^k .
\]
\begin{enumerate}[label=\textup{(\Roman*)}]
\item (Verification) Checking $\score(\alpha) = m$ for a given $\alpha$
takes $O\bigl(\sum_j |S_j|\bigr)$ time, which is $O(m)$ for bounded arity.
\item (Inference) Given a tree decomposition of the primal graph of width
$\tw$, a satisfying assignment is found in $O(n\, d^{\tw+1})$ time for fixed
$\tw$ (Dechter and Pearl, 1989; Freuder, 1990).
\item (Flat-basin reduction) If $S$ is $k$-coupled with respect to
$(\score, \alpha_0)$, then $\score(\restr{\alpha_0}{\beta}) = \vbase$ for every
$\beta \in D^k$ having at least one baseline coordinate; consequently every
$\beta$ with $\score(\restr{\alpha_0}{\beta}) > \vbase$ lies in $Q$, and the
score is constant on the flat basin $D^k \setminus Q$ of size
$d^k - (d-1)^k$.
\item[\textup{(III$'$)}] (Relaxed form) If $S$ is
$(\delta, \Delta)$-$k$-coupled, then $\score(\restr{\alpha_0}{\beta}) \le \vbase
+ \delta$ for every $\beta$ with a baseline coordinate, and every $\beta$ with
$\score(\restr{\alpha_0}{\beta}) \ge \vbase + \Delta$ lies in $Q$; clause~(III)
is the case $\delta = 0$ and $\Delta = 1$.
\end{enumerate}
\end{theorem}

Clauses (I) and (II) are the cited classical results. Clauses (III) and
(III$'$) follow from the next two lemmas.

\begin{lemma}[Flat region]\label{lem:flat}
Let $S$ be $k$-coupled (resp.\ $(\delta,\Delta)$-$k$-coupled). For every
$\beta \in D^k$ with at least one baseline coordinate,
$\score(\restr{\alpha_0}{\beta}) = \vbase$ (resp.\ $\le \vbase + \delta$).
\end{lemma}
\begin{proof}
Let $S' = \{\, x_{i_j} : \beta_j \ne \alpha_0(x_{i_j}) \,\}$. Since at least one
coordinate is at baseline, $|S'| < k$. If $S' = \emptyset$ then
$\restr{\alpha_0}{\beta} = \alpha_0$ and the score is $\vbase$. Otherwise $S'$
is a proper subset of $S$, and writing $\beta'$ for the restriction of $\beta$
to $S'$ we have $\restr{\alpha_0}{\beta} = \restrS{\alpha_0}{S'}{\beta'}$
(the coordinates outside $S'$ are at baseline). Condition~(a) gives
$\score = \vbase$; condition~(a$'$) gives $\score \le \vbase + \delta$.
\end{proof}

\begin{lemma}[Uninformative queries and signal concentration]\label{lem:uninformative}
Under the hypotheses of \Cref{lem:flat}, every $\beta \notin Q$ satisfies
$\score(\restr{\alpha_0}{\beta}) = \vbase$ (rigid case) or $\le \vbase + \delta$
(relaxed case); equivalently, every assignment with score exceeding $\vbase$
(rigid) or reaching $\vbase + \Delta$ (relaxed, using $\Delta > \delta$) lies in
$Q$. In the relaxed case, if in addition $\alpha_0$ is a local score minimum on
$S$-perturbations \textup{(\ref{eq:lmh})}, then $\score(\restr{\alpha_0}{\beta})
\in [\vbase, \vbase + \delta]$ for every $\beta \notin Q$.
\end{lemma}
\begin{proof}
A $\beta \notin Q$ has a baseline coordinate, so \Cref{lem:flat} gives the upper
bound, and its contrapositive (with $\Delta > \delta$ in the relaxed case)
places all signal in $Q$. The two-sided bracket under \textup{(\ref{eq:lmh})} is the
upper bound together with the lower bound $\score \ge \vbase$ supplied by
\textup{(\ref{eq:lmh})} on the proper subset $S'$ defined as in \Cref{lem:flat}.
\end{proof}

The two lemmas prove clause~(III). The score is $\vbase$ off $Q$, and all signal is
concentrated in $Q$, a set of size $(d-1)^k$. The relaxed clause~(III$'$) is
identical with the one-sided bound $\vbase + \delta$. The lower bound below uses
only the upper-bound (signal-concentration) direction, which holds
unconditionally, and \textup{(\ref{eq:lmh})} is needed only for the two-sided bracket.

\begin{corollary}[Classical query lower bound via promise-OR]\label{cor:yao}
Let $\msig = \bigl|\{\, \beta \in Q : \score(\restr{\alpha_0}{\beta}) > \vbase
\,\}\bigr|$ under~(III), or
$\msig = \bigl|\{\, \beta \in Q : \score(\restr{\alpha_0}{\beta}) \ge \vbase +
\Delta \,\}\bigr|$ under~(III$'$). Clauses~(III) and~(III$'$) reduce locating a
signal element via the score oracle to promise-OR on $(d-1)^k$ inputs with
$\msig$ ones. For the coupled set $S$, write $\mathcal{A}(S, \msig)$ for the
family of score functions having $S$ as a $k$-coupled set with exactly $\msig$
witnesses in $Q$, i.e., the score functions consistent with~(III). Write
$\mathcal{A}_{\delta, \Delta}(S, \msig)$ for the family
having $S$ as a $(\delta, \Delta)$-$k$-coupled set with exactly $\msig$
witnesses in $Q$, i.e., those consistent with~(III$'$). The arguments $(S, \msig)$ are dropped when clear. Over the corresponding family:
\begin{enumerate}[label=(\alph*)]
\item For every deterministic algorithm accessing the score through point
queries, there exists a score function in the family forcing at least
$(d-1)^k - \msig + 1$ evaluations of $Q$.
\item For every randomized algorithm, by Yao's minimax principle (Yao, 1977), the expected number of evaluations on the worst-case
score function in the family is $\Omega\bigl((d-1)^k / \msig\bigr)$.
\end{enumerate}
\end{corollary}

\begin{remark}[Robustness and the search-rate tradeoff]
For a $k$-coupled set $S$ with $\msig$ full-set witnesses for condition~(b), the
\emph{robustness} of $S$ is $r(S) := \msig / (d-1)^k$, the fraction of full-set
perturbations producing signal. The classical search lower bound of \Cref{cor:yao}(b), which rests on clause~(III) of \Cref{thm:barrier},
is $\Omega\bigl((d-1)^k / \msig\bigr) = \Omega(1/r(S))$, so fragile sets
(small $r$) require more queries than robust sets (large $r$). Robustness
ranges from $1/(d-1)^k$ (one witness) to $1$ (every full-set perturbation
produces signal).
\end{remark}

\subsection{Coupling Width and Score-Function Treewidth}

\begin{proposition}[Incomparability with treewidth of explicit constraints]\label{prop:incomp-tw}
Let $\tw(\mathcal{P})$ denote the primal treewidth of $\mathcal{P}$'s constraint graph (constraints
only). For every $n \ge 2$, there exist:
\begin{enumerate}[label=(\alph*)]
\item A CSP $\mathcal{P}_n$ with $\tw(\mathcal{P}_n) = 0$ and coupling width $n$ at some baseline.
\item A CSP $\mathcal{P}'_n$ with $\tw(\mathcal{P}'_n) = n-1$ and coupling width $1$ at some baseline.
\end{enumerate}
\end{proposition}

\begin{proof}
(a) Let $X_n = \{x_1, \dots, x_n\}$ and $D = \{0,1\}$, and take the CSP
$\mathcal{P}_n = (X_n, D, \emptyset)$, with empty constraint set and
$\score(\alpha) = \mathrm{AND}(x_1, \dots, x_n)$. The baseline is
$\alpha_0 = (0, \dots, 0)$. The primal constraint graph is edgeless, so
$\tw(\mathcal{P}_n) = 0$. The full set $\{x_1, \dots, x_n\}$ is $n$-coupled at the
baseline and no proper subset is, so the coupling width is $n$.

(b) Take $\mathcal{P}'_n = (X_n, D, \{C\})$ with $D = \{0,1\}$, a single
$n$-ary constraint $C$ requiring $x_1 = 1$, and the constraint-count score
$\score(\alpha) = x_1$. The baseline is
$\alpha_0 = (0, \dots, 0)$. The primal graph of $C$ is $K_n$, so
$\tw(\mathcal{P}'_n) = n-1$. The singleton $S = \{x_1\}$ is $1$-coupled (condition~(a) is
vacuous; setting $x_1 = 1$ gives $\score = 1 > 0 = \vbase$). No larger set can
be coupled. For every larger set containing $x_1$, condition~(a) requires
$\score(\alpha_0[\{x_1\}\!\to\!1]) = 0$, but this score equals $1 > \vbase$. A
set without $x_1$ never changes the score, so condition~(b) fails for it. Hence
$\kappa(\mathcal{P}'_n) = 1$.
\end{proof}

\begin{definition}[Score-function interaction graph]\label{def:Gscore}
The \emph{interaction graph} $\Gscore$ of a score function $\score$ on
$(X, D, \alpha_0)$ has vertex set $X$ and an edge $(x_i, x_j)$ whenever $x_i$ and
$x_j$ \emph{interact}. The variables $x_i$ and $x_j$ interact when there exist an assignment $\chi$ and values
$a_i \ne b_i$ for $x_i$ and $a_j \ne b_j$ for $x_j$ for which the mixed second difference
\[
  D_{ij}(\chi) :=
  \bigl(\score(\chi[x_i{\mapsto}b_i, x_j{\mapsto}b_j]) - \score(\chi[x_i{\mapsto}b_i, x_j{\mapsto}a_j])\bigr)
  - \bigl(\score(\chi[x_i{\mapsto}a_i, x_j{\mapsto}b_j]) - \score(\chi[x_i{\mapsto}a_i, x_j{\mapsto}a_j])\bigr)
\]
is nonzero, where $\chi[x_i{\mapsto}u, x_j{\mapsto}v]$ is the assignment agreeing with $\chi$
except that $x_i = u$ and $x_j = v$. Equivalently, changing $x_i$ alters the marginal effect of changing
$x_j$ with all other coordinates held fixed at $\chi$. (Equivalently still,
$\score$ does not split as $f + g$ with $f$ independent of $x_i$ and $g$
independent of $x_j$, so an additively separable score has no edges.) Its
treewidth is $\twscore$.
\end{definition}

\begin{theorem}[Coupling width is bounded by score-function treewidth]\label{thm:coupling-tw}
Let $\mathcal{P} = (X, D, \score, \alpha_0)$ be a CSP whose $(\delta, \Delta)$-coupled sets
satisfy the local-minimum hypothesis~\textup{(\ref{eq:lmh})} (automatic when
$\vbase = 0$). If $\Delta > 2\delta$ and the score-interaction graph $\Gscore$ has treewidth
$\twscore$, then
\[
  \kappa(\delta, \Delta) \le \twscore + 1.
\]
For coupling in the sense of \Cref{def:coupling-set} the bound holds at every
baseline with no further hypothesis. The threshold is necessary. An additively separable score has
$\twscore = 0$ yet admits a $(\delta, \Delta)$-coupled set of every size $k$ at
$\delta = k - 1$ and $\Delta = k$ (which satisfy $\Delta \le 2\delta$ for
$k \ge 2$), so the bound fails without $\Delta > 2\delta$.
\end{theorem}
\begin{proof}
Let $S = \{x_{i_1}, \dots, x_{i_k}\}$ be a $(\delta, \Delta)$-coupled set of
maximum size $k = \kappa(\delta, \Delta)$, with witness $\beta^*$ satisfying (b$'$). First,
$\beta^* \in Q$. If some coordinate of $\beta^*$ were at baseline, then
$\restr{\alpha_0}{\beta^*}$ would equal a proper-subset perturbation
$\restrS{\alpha_0}{S \setminus \{x_{i_j}\}}{\cdot}$ ($j$ the index of that coordinate), whose score is at most $\vbase + \delta <
\vbase + \Delta$ by (a$'$), which contradicts (b$'$).

We show that every pair in $S$ is
an edge of $\Gscore$. Fix $x_i, x_j \in S$ and let $\chi$ be the assignment
that sets the coordinates $S \setminus \{x_i, x_j\}$ to
$\beta^*|_{S \setminus \{x_i, x_j\}}$ and leaves all remaining coordinates,
including $x_i$ and $x_j$, at baseline. With $\chi[x_i{\mapsto}u, x_j{\mapsto}v]$
as in \Cref{def:Gscore}, form the four corner scores
\begin{align*}
  A_{11} &:= \score(\chi[x_i{\mapsto}\beta^*_i, x_j{\mapsto}\beta^*_j]),\\
  A_{10} &:= \score(\chi[x_i{\mapsto}\beta^*_i, x_j{\mapsto}\alpha_0(x_j)]),\\
  A_{01} &:= \score(\chi[x_i{\mapsto}\alpha_0(x_i), x_j{\mapsto}\beta^*_j]),\\
  A_{00} &:= \score(\chi[x_i{\mapsto}\alpha_0(x_i), x_j{\mapsto}\alpha_0(x_j)]),
\end{align*}
which set $x_i$ and $x_j$ either to baseline or to their $\beta^*$ values. Then
$A_{11} = \score(\restr{\alpha_0}{\beta^*})$ is the score of the full-set witness, so
$A_{11} \ge \vbase + \Delta$ by (b$'$). The scores $A_{10}$ and $A_{01}$ are those of the proper-subset
perturbations of $S \setminus \{x_j\}$ and $S \setminus \{x_i\}$ (each with
$k - 1$ nonbaseline coordinates), so $A_{10}, A_{01} \le \vbase + \delta$ by (a$'$). Finally,
$A_{00}$ is the score of the proper-subset perturbation of $S \setminus \{x_i, x_j\}$ ($k - 2$
coordinates), so $A_{00} \ge \vbase$ by~\textup{(\ref{eq:lmh})}. The mixed second difference
of \Cref{def:Gscore} at $(x_i, x_j)$, all other coordinates held at $\chi$, is
therefore
\[
  D_{ij}(\chi) = (A_{11} - A_{10}) - (A_{01} - A_{00}) = A_{11} - A_{10} - A_{01} + A_{00}
  \ge (\vbase + \Delta) - 2(\vbase + \delta) + \vbase = \Delta - 2\delta > 0,
\]
using $\Delta > 2\delta$. (When $k = 2$ the set $S \setminus \{x_i, x_j\}$ is
empty, $A_{00} = \vbase$ exactly, and~\textup{(\ref{eq:lmh})} is not needed.) Hence changing
$x_i$ alters the marginal effect of changing $x_j$, so $(x_i, x_j)$ is an edge
of $\Gscore$ by \Cref{def:Gscore}. As the pair was arbitrary, $\Gscore$
restricted to $S$ is the complete graph $K_k$, so
$\twscore \ge \tw(K_k) = k - 1$, giving $\kappa = k \le \twscore + 1$.

For coupling in the sense of \Cref{def:coupling-set}, a $k$-coupled set is $(0, \Delta)$-coupled
for its own margin $\Delta > 0$, and condition~(a) gives every proper-subset
perturbation the score $\vbase$, so \textup{(\ref{eq:lmh})} holds,
$A_{10} = A_{01} = A_{00} = \vbase$, and the argument above applies with $\delta = 0$.
\end{proof}

\begin{lemma}[Constraint-count scores interact only along constraints]\label{lem:gscore-primal}
For the constraint-count score of \Cref{def:score}, every edge of $\Gscore$ is an edge of the primal graph $G$. Hence $\twscore \le \tw$, and, for coupling in the sense of \Cref{def:coupling-set} or under the hypotheses of \Cref{thm:coupling-tw},
\[
  \kappa \le \twscore + 1 \le \tw + 1 .
\]
\end{lemma}
\begin{proof}
The score is given by $\score(\alpha) = \sum_{\ell=1}^{m} \mathbf{1}[\alpha|_{S_\ell} \in \mathrm{Rel}_\ell]$, where the $\ell$th term depends only on the variables in $S_\ell$. The mixed second difference of \Cref{def:Gscore} is linear in $\score$, so $D_{ij}(\chi)$ is the sum of the mixed second differences of the $m$ terms. If $x_i \notin S_\ell$, the $\ell$th term does not change with $x_i$, so the two brackets coincide for it; if $x_j \notin S_\ell$, each bracket is zero for it. In both cases its contribution is zero. Hence $D_{ij}(\chi) \ne 0$ requires a scope $S_\ell$ containing both $x_i$ and $x_j$, i.e., an edge of $G$. Treewidth does not increase under taking subgraphs on the same vertex set, so $\twscore \le \tw$, and \Cref{thm:coupling-tw} gives the chain of inequalities.
\end{proof}

The bound is one-sided. The coupling width need not grow with $\tw$, even for a constraint-count score (\Cref{prop:incomp-tw}(b)), and for a score that is not the count of the instance's own constraints, $\kappa$ can exceed $\tw + 1$ (\Cref{prop:incomp-tw}(a), where $\kappa = \twscore + 1$).

The register match score of \Cref{def:match-score} is a constraint count, with one constraint $C_t$ per position $t$ requiring $\exptr(P, N)[t] = s[t]$ (an undefined output counting as a mismatch). Its scope $S_t$ is the set of fields $x$ for which some register assignment exists where changing $x$ alone changes the output at $t$. Dependence is taken over all assignments and not along the data flow of one program. The output at $t$ is then a function of the assignment on $S_t$, since one can pass between two assignments that agree on $S_t$ by changing one field outside $S_t$ at a time. We take the $C_t$ as the constraints of the register CSP, so \Cref{lem:gscore-primal} applies to it with $G$ the primal graph of the scopes $S_t$.

\subsection{Overview of the Search-Inference Coupling}

The primal graph $G$ relates to two costs. Inference tractability is governed by
its treewidth $\tw$ (clause~(II)), and the structure of the search barrier is governed by
the coupling width $\kappa$, which fixes the size $(d-1)^\kappa$ of the witness
cube $Q$ in the score landscape (clause~(III)). \Cref{thm:coupling-tw} and \Cref{lem:gscore-primal} tie these together.
For constraint-count scores they give $\kappa \le \twscore + 1 \le \tw + 1$,
unconditionally for coupling in the sense of \Cref{def:coupling-set} and under
the local-minimum and $\Delta > 2\delta$ hypotheses for the relaxed coupling.
A coupling width $\kappa$ therefore requires a $\kappa$-clique of score
interactions, which for such scores is a $\kappa$-clique of $G$, so
$\tw \ge \kappa - 1$. The consequence
is a conservation law, made precise in \Cref{thm:cons-budget} and
\Cref{cor:cw-unit} below.

Every structural relaxation that reduces treewidth
(making inference tractable) also lowers the ceiling $\tw + 1$ on
the coupling width of the relaxed instance, when that instance is still a
constraint count satisfying the hypotheses. Such a relaxation therefore cannot leave a search
barrier wider than $\tw + 1$. The price of relaxation is paid
in completeness.

\section{Reformulation and the Conservation Law}

The barrier can be paid down by reencoding the problem. This section defines
reformulations, prices them in bits through the Kolmogorov chain rule, and
proves the conservation law, under which structure and search trade at best one for
one against the program sought.

\subsection{Reformulations}

\begin{definition}[Reformulation]\label{def:reformulation}
Let $\mathcal{P}_1 = (X_1, D_1, C_1)$ and $\mathcal{P}_2 = (X_2, D_2, C_2)$ be CSPs with score
functions $\score_1, \score_2$. A \emph{reformulation} of $\mathcal{P}_1$ as $\mathcal{P}_2$ is a
pair of maps $(\varphi, \psi)$ with $\varphi : D_1^{X_1} \to D_2^{X_2}$ and
$\psi : D_2^{X_2} \to D_1^{X_1}$ (assignments to assignments) that preserve
score in both directions, i.e., $\score_1(\alpha) = \score_2(\varphi(\alpha))$ for
every $\alpha \in D_1^{X_1}$ and $\score_2(\beta) = \score_1(\psi(\beta))$ for
every $\beta \in D_2^{X_2}$. These two equalities are the score-preservation
clause invoked below. In particular, satisfying assignments correspond to
satisfying assignments in both directions.
\end{definition}

\begin{definition}[Cost-constrained reformulation]\label{def:cost-constrained}
A reformulation $\mathcal{P}_1 \to \mathcal{P}_2$ via $(\varphi, \psi)$ is \emph{cost-constrained}
if (i) $|D_2| \le \operatorname{poly}(|D_1|)$ independent of
$|X_1|$, (ii) $K(\varphi, \psi) \le \operatorname{poly}(|X_1|, \log |D_1|)$, and
(iii) $|X_2| \le \operatorname{poly}(|X_1|)$.
\end{definition}

\begin{theorem}[Coupling width is not invariant under reformulation]\label{thm:cw-not-invariant}
For every $k \ge 2$ there exist CSPs $\mathcal{P}_1$ and $\mathcal{P}_2$ such that $\mathcal{P}_2$ is a
reformulation of $\mathcal{P}_1$, with coupling width $k$ at a baseline in $\mathcal{P}_1$ and
coupling width $1$ at the image baseline in $\mathcal{P}_2$.
\end{theorem}
\begin{proof}
Take $\mathcal{P}_1$ with $X_1 = \{x_1, \dots, x_k\}$, $D_1 = \{0,1\}$,
$\score_1 = x_1 \wedge \cdots \wedge x_k$, baseline $0^k$. The full set is
$k$-coupled (every proper subset leaves a coordinate at $0$, so the AND is $0 =
\vbase$, while $1^k$ scores $1$), and $|X_1| = k$ forces $\kappa(\mathcal{P}_1) = k$. Take
$\mathcal{P}_2$ with $X_2 = \{z\}$, $D_2 = \{0,1\}$, $\score_2(z) = z$, baseline $0$, so
$\kappa(\mathcal{P}_2) = 1$. The maps $\varphi(x_1, \dots, x_k) = x_1 \wedge \cdots \wedge
x_k$ and $\psi(z) = (z, \dots, z)$ preserve score in both directions and are
$O(k)$-computable, so $\mathcal{P}_2$ is a reformulation of $\mathcal{P}_1$.
\end{proof}

\begin{remark}[The AND reformulation as a cost-constrained instance]
The example in the proof is the AND score on $k$ literals. In $\mathcal{P}_1$ the full set
$S = \{x_1, \dots, x_k\}$ is $k$-coupled, while the reformulation $\mathcal{P}_2$ collapses
$S$ to a single variable $z = x_1 \wedge \cdots \wedge x_k$ whose singleton is
$1$-coupled. The collapse costs $O(\log k)$ bits to describe.
\end{remark}

\begin{proposition}[Universal collapse]\label{prop:universal-collapse}
For every CSP $\mathcal{P} = (X, D, C)$ with score function $\score$ and baseline, there
is a reformulation $\mathcal{P}'$ that inflates the domain, with coupling width at most $1$. It has a single variable $z$ of domain
$D' = D^{|X|}$ and $\score'(z) := \score(\alpha_z)$, where $\alpha_z$ is the assignment
that the value $z$ names. This reformulation is not cost-constrained (\Cref{def:cost-constrained}),
since $|D'|$ is exponential in $|X|$.
\end{proposition}
\begin{proof}
Identify each assignment $\alpha \in D^{|X|}$ with a value $z_\alpha$ of the single
variable, so that $\alpha_z$ is the assignment named by $z$, and set $\score'(z) = \score(\alpha_z)$; the maps
$\varphi(\alpha) = z_\alpha$ and $\psi(z) = \alpha_z$ preserve score. A single
variable has coupling width at most $1$, and the inflation $|D'| = |D|^{|X|}$ is
exponential in $|X|$, so condition~(i) of \Cref{def:cost-constrained} fails.
\end{proof}

\subsection{The Kolmogorov Chain Rule and the Counting Bound}

Throughout, $K(\cdot)$ is Kolmogorov complexity on a fixed universal prefix
machine, and $\mathrm{fw}(\mathcal{P})$ denotes the framework data of $\mathcal{P}$ (variables, domain,
baseline) excluding the score function.

\begin{theorem}[Reformulation chain rule]\label{thm:chain-rule}
Let $\mathcal{P}' = (X', D', \score', \alpha_0')$ be a reformulation of
$\mathcal{P} = (X, D, \score, \alpha_0)$ via maps $(\varphi, \psi)$. Then
\[
  K(\score' \mid \mathrm{fw}(\mathcal{P}')) \ge K(\score \mid \mathrm{fw}(\mathcal{P})) - K(\varphi, \psi) - O(1).
\]
\end{theorem}
\begin{proof}
By the score-preservation clause of \Cref{def:reformulation},
$\score(\alpha) = \score'(\varphi(\alpha))$ for all $\alpha \in D^{|X|}$, so
$\score$ is computable from $\score'$ and $\varphi$. Descriptions of
reformulations are self-delimiting and carry the framework data $\mathrm{fw}(\mathcal{P}')$ of the
target instance, so the decoder recovers $\mathrm{fw}(\mathcal{P}')$ from the description of
$(\varphi, \psi)$ at $O(1)$ additional cost. Consider the program that reads
$\mathrm{fw}(\mathcal{P})$, reads $\varphi$ (part of $(\varphi, \psi)$) and forms
$\alpha_0' = \varphi(\alpha_0)$, reads a description of $\score'$ of length
$K(\score' \mid \mathrm{fw}(\mathcal{P}'))$, and outputs $\alpha \mapsto \score'(\varphi(\alpha))$.
It computes $\score$ and has total length $K(\score' \mid \mathrm{fw}(\mathcal{P}')) + K(\varphi, \psi) + O(1)$. Hence $K(\score \mid \mathrm{fw}(\mathcal{P})) \le K(\score' \mid \mathrm{fw}(\mathcal{P}')) + K(\varphi, \psi)
+ O(1)$, which rearranges to the claim. (The symmetric argument with $\psi$
gives $|K(\score \mid \mathrm{fw}(\mathcal{P})) - K(\score' \mid \mathrm{fw}(\mathcal{P}'))| \le K(\varphi, \psi) +
O(1)$.)
\end{proof}

\begin{theorem}[Counting bound for coupling-structured scores]\label{thm:counting}
Fix $(X, D, \alpha_0)$, a $k$-subset $S \subseteq X$ with $k \ge 1$, and
$\msig \ge 1$. Let $\mathcal{A}(S, \msig)$ be the family of \Cref{cor:yao}. Then, for every $c \ge 0$, for all but a
$2^{-c}$ fraction of the canonical indicator scores in $\mathcal{A}(S, \msig)$ (those supported on $Q(S)$),
\[
  K(\score \mid \mathrm{fw}(\mathcal{P}), S, \msig) \ge \log_2 \binom{(d-1)^k}{\msig} - c .
\]
\end{theorem}
\begin{proof}
Restrict to the canonical family, in which $\score$ vanishes off the witness cube
$Q(S)$ and is the $\{0,1\}$-indicator of an $\msig$-subset $Z \subseteq Q(S)$.
The map $Z \mapsto \score_Z$ is a bijection onto this family, and distinct $Z$
give distinct indicators. Each such
$\score_Z$ has $S$ as a $k$-coupled set. Condition~(a) holds because every proper
subperturbation leaves a baseline coordinate and so lands off $Q(S)$, where the score is
$0 = \vbase$, and condition~(b) holds because $Z \ne \emptyset$. Hence the family has size
$\binom{(d-1)^k}{\msig}$, and given $(\mathrm{fw}(\mathcal{P}), S, \msig)$ its members are
enumerable in lexicographic order, so each is addressed by its rank in
$\log_2 \binom{(d-1)^k}{\msig} + O(1)$ bits. By the counting argument for incompressibility,
at least a $1 - 2^{-c}$ fraction have
$K(\score \mid \mathrm{fw}(\mathcal{P}), S, \msig) \ge \log_2 \binom{(d-1)^k}{\msig} - c$.
\end{proof}

\begin{proposition}[Counting upper bound]\label{thm:counting-upper}
For every $\score \in \mathcal{A}(S, \msig)$ whose support is restricted to $Q(S)$ with
indicator-style values,
\[
  K(\score \mid \mathrm{fw}(\mathcal{P}), S, \msig) \le \log_2 \binom{(d-1)^k}{\msig} + O(\log |X|).
\]
\end{proposition}
\begin{proof}
Such a $\score$ is determined by its witness set $Z_{\score} = \{\, \alpha \in
Q(S) : \score(\alpha) = 1 \,\}$, an $\msig$-subset of $Q(S)$. Given
$(\mathrm{fw}(\mathcal{P}), S, \msig)$, specifying $\score$ reduces to specifying $Z_{\score}$ by its
rank among the $\binom{(d-1)^k}{\msig}$ such subsets, which costs
$\log_2 \binom{(d-1)^k}{\msig} + O(\log |X|)$ bits.
\end{proof}

\begin{corollary}[The coupling-entropy tradeoff]\label{cor:coupling-entropy}
Let $\mathcal{P}$ have $S$ as a $k$-coupled set with $\msig$ witnesses and $\score$
Kolmogorov-typical in $\mathcal{A}(S,\msig)$ for a constant $c \ge 0$, i.e., $K(\score \mid \mathrm{fw}(\mathcal{P}), S, \msig) \ge \log_2 \binom{(d-1)^k}{\msig} - c$. Then every reformulation $(\varphi,\psi)$ to
$\mathcal{P}'$ obeys
\[
  K(\varphi, \psi) + K(\score' \mid \mathrm{fw}(\mathcal{P}'))
  \ge \log_2 \binom{(d-1)^k}{\msig} - c - O(1).
\]
\end{corollary}
\begin{proof}
The bound follows from \Cref{thm:chain-rule} (which gives $K(\score' \mid \mathrm{fw}(\mathcal{P}')) \ge
K(\score \mid \mathrm{fw}(\mathcal{P})) - K(\varphi, \psi) - O(1)$) and the typicality hypothesis, since
$K(\score \mid \mathrm{fw}(\mathcal{P})) \ge K(\score \mid \mathrm{fw}(\mathcal{P}), S, \msig) - O(1)$. By \Cref{thm:counting}, typicality holds for all but a $2^{-c}$ fraction of the canonical scores.
\end{proof}

\subsection{Cost-Constrained Reformulations}

\begin{definition}[Coupling-structure-preserving reformulation]\label{def:coupling-preserving}
A cost-constrained reformulation $\mathcal{P} \to \mathcal{P}'$ via $(\varphi, \psi)$ is
\emph{coupling-structure-preserving} with respect to a coupled set $S$ if there
is a $k'$-subset $\tilde S \subseteq X'$ and an integer $m'_{\mathrm{signal}}$ such that
$\score'$ has $\tilde S$ as a $k'$-coupled set at $\alpha_0' = \varphi(\alpha_0)$ with
$m'_{\mathrm{signal}}$ witnesses, and
\[
  K(\tilde S, m'_{\mathrm{signal}} \mid S, \msig, \mathrm{fw}(\mathcal{P}), \mathrm{fw}(\mathcal{P}'), \varphi, \psi)
  = O(\log |X| + \log |X'|),
\]
i.e., the coupled set of $\mathcal{P}'$ and its witness count can be recovered from those of $\mathcal{P}$, the two frameworks and the maps with $O(\log |X| + \log |X'|)$ bits.
\end{definition}

\begin{corollary}[Coupling width after a cost-constrained reformulation]\label{cor:3-2f}
Let $S$ be a $k$-coupled set for $\score$ at $\alpha_0$ with $\msig$ witnesses,
and let $\score$ be Kolmogorov-typical in $\mathcal{A}(S, \msig)$ for a constant $c \ge 0$, i.e.,
\[
  K(\score \mid \mathrm{fw}(\mathcal{P}), S, \msig) \ge \log_2 \binom{(d-1)^k}{\msig} - c .
\]
Let $\mathcal{P}'$ be a cost-constrained reformulation, with $|D'| \le \operatorname{poly}(|D|)$ and
$K(\varphi, \psi) \le q(|X|, \log |D|)$ for a polynomial $q$. Then:
\begin{enumerate}[label=(\roman*)]
\item $\displaystyle K(\score' \mid \mathrm{fw}(\mathcal{P}')) \ge \log_2 \binom{(d-1)^k}{\msig}
  - q(|X|, \log |D|) - c - O(1).$
\item If moreover $\score'$ has a $k'$-coupled set $\tilde S$ with $m'_{\mathrm{signal}}$
witnesses and is \emph{canonical} on it, i.e., its support is restricted to $Q(\tilde S)$ with
indicator-style values, then in the
unique-witness specialization $\msig = m'_{\mathrm{signal}} = 1$,
\[
  k' \ge \frac{k \log_2(|D|-1)}{\log_2(|D'|-1) + \log_2 |X'|}
  - O\!\left(\frac{q(|X|, \log |D|)}{\log_2(|D'|-1) + \log_2 |X'|}\right).
\]
The hypothesis that $\score'$ is canonical on $\tilde S$ is the one under which \Cref{thm:counting-upper}
gives $K(\score' \mid \mathrm{fw}(\mathcal{P}'), \tilde S, m'_{\mathrm{signal}}) \le \log_2
\binom{(|D'|-1)^{k'}}{m'_{\mathrm{signal}}} + O(\log |X'|)$.
\end{enumerate}
If moreover $q(|X|, \log|D|) = o(k \log_2(|D|-1))$ and if $\log_2(|D'|-1) = (1+o(1))\log_2(|D|-1)$
and $\log_2|X'| = (1+o(1))\log_2|X|$, the bound of part~(ii) reads $k' \ge k - o(k)$ in the
regime dominated by the domain, $\log_2(|D|-1) \gg \log_2 |X|$, and $k' \ge k/2 - o(k)$ in the
balanced regime $\log_2(|D|-1) \sim \log_2|X|$. These conditions are needed. When $\score$ is the indicator of a single witness, the
reformulation that replaces $S$ by one binary variable $z$, with $\varphi(\alpha) = 1$
exactly when $\alpha$ agrees with the witness on $S$ and $\psi(1)$ the witness, is
cost-constrained and gives $k' = 1$. Its maps cost about $k \log_2(|D|-1)$ bits.

\smallskip
\noindent\textup{(iii)} If the reformulation $\mathcal{P} \to \mathcal{P}'$ is coupling-structure-preserving with respect to
$S$ (\Cref{def:coupling-preserving}, so $(\tilde S, m'_{\mathrm{signal}})$ is
recoverable from $(S, \msig, \varphi, \psi)$ at cost $O(\log|X| + \log|X'|)$),
and $\score'$ is canonical on $\tilde S$, then the $k' \log_2 |X'|$ overhead of part~(ii) is
absent. For $\msig =
m'_{\mathrm{signal}} = 1$ with $|D'| = \Theta(|D|)$,
\[
  k' \log_2(|D'|-1) \ge k \log_2(|D|-1) - K(\varphi,\psi) - O(\log|X| + \log|X'|),
  \quad\text{i.e.}\quad
  k' \ge k - O\!\left(\frac{q(|X|,\log|D|)}{\log_2(|D|-1)}\right).
\]
\end{corollary}
\begin{proof}
(i) Chaining \Cref{thm:chain-rule} with the typicality hypothesis gives $K(\score' \mid
\mathrm{fw}(\mathcal{P}')) \ge K(\score \mid \mathrm{fw}(\mathcal{P})) - K(\varphi, \psi) - O(1) \ge \log_2
\binom{(d-1)^k}{\msig} - c - q(|X|, \log |D|) - O(1)$, using $K(\varphi, \psi)
\le q$. (ii) Apply \Cref{thm:counting-upper} to $\mathcal{P}'$ together with the conditional chain
rule $K(\score' \mid \mathrm{fw}(\mathcal{P}')) \le K(\score' \mid \mathrm{fw}(\mathcal{P}'), \tilde S, m'_{\mathrm{signal}})
+ K(\tilde S, m'_{\mathrm{signal}} \mid \mathrm{fw}(\mathcal{P}')) + O(1)$. The first term is at most
$\log_2 \binom{(|D'|-1)^{k'}}{m'_{\mathrm{signal}}} + O(\log |X'|)$; the second is at
most $k' \log_2 |X'| + O(\log m'_{\mathrm{signal}}) + O(\log |X'|)$ (encode $k'$ in
$O(\log |X'|)$ bits, the $k'$-subset $\tilde S$ in at most $k' \log_2 |X'|$ bits, and the
integer $m'_{\mathrm{signal}}$ in $O(\log m'_{\mathrm{signal}})$ bits). Combining with (i) and specializing to
$\msig = m'_{\mathrm{signal}} = 1$ (so the binomials collapse to $(d-1)^k$ and
$(|D'|-1)^{k'}$) gives $k' \log_2(|D'|-1) + k' \log_2 |X'| + O(\log |X'|) \ge
k \log_2(|D|-1) - q - O(1)$, which rearranges to the stated bound. The two
regimes follow by comparing $\log_2|X'|$ to $\log_2(|D'|-1)$.

(iii) Work conditionally on $(S, \msig)$, as the typicality hypothesis does. As in
\Cref{thm:chain-rule}, $\score$ is computed from $\score'$ and $(\varphi, \psi)$,
whose description carries $\mathrm{fw}(\mathcal{P}')$, and, by
\Cref{def:coupling-preserving}, $(\tilde S, m'_{\mathrm{signal}})$ is recovered from
$(S, \msig, \varphi, \psi)$ and the two frameworks with $O(\log|X| + \log|X'|)$ bits.
Hence
$K(\score \mid \mathrm{fw}(\mathcal{P}), S, \msig) \le K(\score' \mid
\mathrm{fw}(\mathcal{P}'), \tilde S, m'_{\mathrm{signal}}) + K(\varphi, \psi) + O(\log|X| +
\log|X'|)$. For $\msig = m'_{\mathrm{signal}} = 1$ the left side is at least
$k \log_2(|D|-1) - c$ by typicality, and the first term on the right is at most
$k' \log_2(|D'|-1) + O(\log|X'|)$ by \Cref{thm:counting-upper}, so
$k \log_2(|D|-1) - c \le k' \log_2(|D'|-1) + K(\varphi,\psi) +
O(\log|X|+\log|X'|)$. With $K(\varphi,\psi) \le q$ and $|D'| = \Theta(|D|)$
(so $\log_2(|D'|-1) = \log_2(|D|-1) + O(1)$), this gives $k' \log_2(|D'|-1) \ge
k \log_2(|D|-1) - q - O(\log|X|+\log|X'|)$ and hence $k' \ge k -
O(q/\log_2(|D|-1))$.
\end{proof}

\subsection{The Conservation Law}

\begin{definition}[Coupling budget]\label{def:budget}
For a CSP $\mathcal{P}$ with coupling width $\kappa$ (\Cref{def:coupling-width}) and $|D| = d$, the
\emph{coupling budget} is $B(\mathcal{P}) := \kappa \log_2(d-1)$ bits, i.e.,
$\log_2 |Q|$, the logarithm of the volume of the witness cube.
\end{definition}

\begin{theorem}[Coupling budget bound]\label{thm:cons-budget}
Let $\mathcal{P}$ have a binary score, $\score \in \{0, 1\}$ with $\vbase = 0$,
coupling width $\kappa$, budget $B = \kappa \log_2(d-1)$, and a $\kappa$-coupled set $S$ with a single witness;
let $\score$ be Kolmogorov-typical in $\mathcal{A}(S, 1)$ for a constant $c \ge 0$, i.e.,
$K(\score \mid \mathrm{fw}(\mathcal{P}), S) \ge B - c$. Let $\mathcal{P}'$ be a
cost-constrained, coupling-structure-preserving reformulation via
$(\varphi, \psi)$ whose guaranteed $k'$-coupled set $\tilde S$ also has a single
witness, with $\score'$ canonical on $\tilde S$
(support restricted to $Q(\tilde S)$), and let
$\operatorname{search}(\mathcal{P}')$ be the worst-case number of score queries over the
score functions consistent with the coupling of $\tilde S$ (\Cref{cor:yao}). Then
\begin{equation}\label{eq:cons-unique}
  K(\varphi, \psi) + \log_2 \operatorname{search}(\mathcal{P}') \ge B - O(\log |X| + \log |X'|). \tag{$*_B$}
\end{equation}
No reformulation reduces the search exponent by more bits than it spends on describing
its maps. Since $|X'| \le \operatorname{poly}(|X|)$
(\Cref{def:cost-constrained}(iii)), the overhead is $O(\log |X|)$, and the bound
is informative whenever $\log |X| = o(B)$.
\end{theorem}

\begin{proof}
Write $Q'$ for the witness cube of $\tilde S$, so $|Q| = (d-1)^{\kappa} = 2^B$ and
$|Q'| = (|D'|-1)^{k'}$. All descriptions below are conditional on
$(\mathrm{fw}(\mathcal{P}), S)$; the framework data $\mathrm{fw}(\mathcal{P}')$ are not assumed, since the codomain
$(X', D')$ and the image baseline $\alpha_0' = \varphi(\alpha_0)$ are recovered
from $\varphi$ at $O(1)$ cost.

First, a program that reads $(\varphi, \psi)$, then a
description of $\score'$ given $(\mathrm{fw}(\mathcal{P}), S, \varphi, \psi)$, and outputs
$\alpha \mapsto \score'(\varphi(\alpha))$ computes $\score$
(\Cref{def:reformulation}), so
\[
  K(\score \mid \mathrm{fw}(\mathcal{P}), S) \le K(\varphi, \psi)
  + K(\score' \mid \mathrm{fw}(\mathcal{P}), S, \varphi, \psi) + O(1).
\]

Second, typicality bounds the left side below by
$B - c$. By \Cref{def:coupling-preserving}, $\tilde S$ is
recoverable from $(\mathrm{fw}(\mathcal{P}), S, \varphi, \psi)$ in $O(\log |X| + \log |X'|)$ bits,
and \Cref{thm:counting-upper}, applied to the canonical $\score'$ with its single
witness, bounds the last term above by $\log_2 |Q'| + O(\log |X| + \log |X'|)$. Substituting both bounds into the inequality of the first step gives
\[
  \log_2 |Q'| \ge B - K(\varphi, \psi) - O(\log |X| + \log |X'|).
\]

Finally, clause~(III) and \Cref{cor:yao} applied to
$\mathcal{P}'$ at $\tilde S$ give $\log_2 \operatorname{search}(\mathcal{P}') \ge \log_2 |Q'| - O(1)$,
and chaining this inequality with the bound of the second step gives the conservation bound in \eqref{eq:cons-unique}.
\end{proof}

\begin{corollary}[Form in units of coupling width]\label{cor:cw-unit}
Under the hypotheses of \Cref{thm:cons-budget},
let $|r|_{\mathrm{cw}} := K(\varphi, \psi) / \log_2(d-1)$. Then
\[
  |r|_{\mathrm{cw}} + \log_{d-1} \operatorname{search}
  \ge \kappa - O\!\left(\frac{\log |X| + \log |X'|}{\log_2(d-1)}\right).
\]
Each unit of reformulation, measured in coupling width, eliminates at most one variable's
worth of search from $Q$. With several witnesses the search term alone still obeys
$\operatorname{search} = \Omega((d-1)^{\kappa}/\msig)$ (\Cref{cor:yao}); the trade against
reformulation is proved here for a single witness.
\end{corollary}
\begin{proof}
Divide \eqref{eq:cons-unique} by $\log_2(d-1)$; the $O(\log|X| + \log|X'|)$ overhead becomes
$O((\log|X| + \log|X'|)/\log_2(d-1))$.
\end{proof}

\begin{remark}[When the witness count is preserved]\label{rem:witness-count}
The hypothesis of \Cref{thm:cons-budget} that $\tilde S$ keeps a single witness holds
whenever $\psi$ maps the $\tilde S$-local perturbations of $\alpha_0'$ injectively to
$S$-local perturbations of $\alpha_0$. Score preservation then sends distinct
witnesses of $\tilde S$ to distinct witnesses of $S$, since
$\score(\psi(\beta)) = \score'(\beta)$ and $\score'(\alpha_0') =
\score(\alpha_0) = \vbase$. The identity reformulation satisfies this condition. A
reformulation that multiplies witnesses is not covered, and whether the bound
survives there is open.
\end{remark}

\subsection{Irreducibility of Coupling Width to Known Priors}

\begin{proposition}[Incomparability with Kolmogorov complexity of the score]\label{prop:incomp-K}
For every $n \ge 2$ there exist:
\begin{enumerate}[label=(\alph*)]
\item a CSP $\mathcal{P}_n$ with $K(\score_n \mid \mathrm{fw}(\mathcal{P}_n)) = O(\log n)$ and coupling width $n$;
\item a CSP $\mathcal{P}'_n$ with $K(\score'_n \mid \mathrm{fw}(\mathcal{P}'_n)) = \Omega(2^n)$ and coupling width $1$.
\end{enumerate}
\end{proposition}

\begin{proof}
(a) Take $\score_n = \mathrm{AND}(x_1, \dots, x_n)$ on $n$ binary variables,
baseline $0^n$; given $(n, 2, 0^n)$ the AND circuit is describable in
$O(\log n)$ bits, and the coupling width is $n$.

(b) The family of Boolean functions $\{0,1\}^n \to \{0,1\}$ has cardinality
$2^{2^n}$; by incompressibility counting, all but a $2^{-c}$ fraction satisfy
$K(\score \mid \mathrm{fw}(\mathcal{P})) \ge 2^n - c - O(1)$. Fix such an incompressible $f$ and
modify it on the points of Hamming weight at most one to force coupling width
$1$. Set
\[
  g(\alpha) := f(\alpha) \text{ for } |\alpha| \ge 2, \quad
  g(0^n) := 0, \quad g(e_i) := 1 \text{ for every } i,
\]
where $e_i$ is the $i$th unit vector. Then $\vbase = 0$ and every singleton
$\{x_i\}$ is $1$-coupled, since condition~(a) is vacuous for singletons and
$g(e_i) = 1 > 0$ gives condition~(b). No set $S$ with $|S| \ge 2$ is coupled, since
$S$ contains some singleton $\{x_i\}$ as a proper subset and the perturbation
setting $x_i$ to $1$ scores $g(e_i) = 1 > \vbase$, which violates condition~(a).
Hence $\kappa = 1$, by the same local mechanism that drives the
OR landscape of \Cref{prop:incomp-sens} below. Since $f$ and $g$
differ only on the $n+1$ points of weight at most one, $f$ is recoverable from
$g$ together with the $n+1$ bits $f(0^n), f(e_1), \dots, f(e_n)$ listed in
canonical order, so $K(f \mid g, \mathrm{fw}(\mathcal{P}'_n)) \le n + O(\log n)$ and
$K(g \mid \mathrm{fw}(\mathcal{P}'_n)) \ge K(f \mid \mathrm{fw}(\mathcal{P}'_n)) - n - O(\log n) = \Omega(2^n)$.
\end{proof}

\begin{proposition}[Incomparability with single-coordinate sensitivity]\label{prop:incomp-sens}
Let $\mathrm{sens}(\score, \alpha_0) := |\{\, i : \score(\alpha_0 \oplus e_i) \ne
\score(\alpha_0) \,\}|$. For every $n \ge 2$ there is a CSP with
$\mathrm{sens}(\score, \alpha_0) = 0$ and coupling width $n$ (AND at $0^n$), and a CSP with
$\mathrm{sens}(\score, \alpha_0) = n$ and coupling width $1$ (OR at $0^n$).
\end{proposition}
\begin{proof}
For AND at $0^n$, every single-coordinate flip keeps a $0$, so the score stays
$0$ and the sensitivity is $0$, while the coupling width is $n$. For OR at $0^n$,
every single flip raises the score to $1$, so the sensitivity is $n$. The
singleton $\{x_1\}$ is $1$-coupled (and no larger set is, since condition~(a)
already fails at $\{x_1\}$), so the coupling width is $1$.
\end{proof}

\begin{corollary}[Coupling width is an independent prior]\label{cor:independent-prior}
Coupling width at a baseline is not a function of primal treewidth
(\Cref{prop:incomp-tw}), of the Kolmogorov complexity of the score
(\Cref{prop:incomp-K}), or of single-coordinate sensitivity
(\Cref{prop:incomp-sens}); for each of the three there are CSPs with equal
value of that quantity and different coupling widths.
\end{corollary}
\begin{proof}
Take $n \ge 3$ binary variables, empty constraint set and baseline $0^n$. The
score $\mathrm{AND}(x_1, \dots, x_n)$ has coupling width $n$
(\Cref{prop:incomp-tw}(a)). The score $x_1$ has coupling width $1$, since
$\{x_1\}$ is $1$-coupled, every larger set containing $x_1$ violates
condition~(a) at $\{x_1\}$, and a set without $x_1$ never changes the score.
The score $\mathrm{AND}(x_1, x_2)$ has coupling width $2$, since $\{x_1, x_2\}$
is $2$-coupled, every larger set containing both violates condition~(a) at
$\{x_1, x_2\}$, and a set missing one of them never changes the score. The
first two scores both have primal treewidth $0$ and both have
$K(\score \mid \mathrm{fw}(\mathcal{P})) = O(\log n)$, but coupling widths $n$
and $1$. The first and third both have single-coordinate sensitivity $0$ at
$0^n$, but coupling widths $n$ and $2$.
\end{proof}

\section{The Register Machine Model}

The theory so far is stated for abstract CSPs. We now fix the
concrete machine, the syntax and semantics of register programs and the score
functions over them. The structural and dynamical analyses of the next two
sections use this machine.

\subsection{Syntax and Semantics}

\begin{definition}[Register machine program]\label{def:program}
A program is a triple $P = (F, \tau, \rho)$ where:
\begin{enumerate}[label=(\roman*)]
\item $F$ is a \emph{control frame}, a rooted tree whose internal nodes are
control structures (\textsc{while}, \textsc{if\_else}, \textsc{repeat},
\textsc{seq}) and whose $L$ leaves are operation slots $l_1, \dots, l_L$.
\item $\tau : \{l_1, \dots, l_L\} \to \{1, \dots, T\}$ assigns each leaf one of
$T$ operation types (e.g., \textsc{add}, \textsc{mod}, \textsc{divmod},
\textsc{copy}, \textsc{load}, \textsc{out}).
\item $\rho : \mathrm{Fields}(F, \tau) \to R$ assigns each register field a
physical register in $R = \{R_0, \dots, R_{W-1}\}$.
\end{enumerate}
\end{definition}

\begin{definition}[Execution]\label{def:execution}
Given $P$ and a length $N$, the output $\exptr(P, N) \in \Sigma^N$ is produced by initializing
registers from preloop \textsc{load} operations, iterating the outermost loop,
and emitting one output symbol each time the \textsc{out} operation executes, once
per iteration unless stated otherwise. The output $\exptr(P, N)$ is the first $N$ symbols
emitted.
\end{definition}

\begin{definition}[Functional equivalence]\label{def:func-equiv}
Programs $P_1 = (F, \tau, \rho_1)$ and $P_2 = (F, \tau, \rho_2)$ sharing frame
and types are \emph{functionally equivalent}, $P_1 \equiv P_2$, if
$\exptr(P_1, N) = \exptr(P_2, N)$ for all $N \ge 1$.
\end{definition}

\begin{definition}[Match score]\label{def:match-score}
For program $P$ and target $s \in \Sigma^N$,
\[
  \score(P, s) = \bigl|\{\, i \in \{0, \dots, N-1\} : \exptr(P, N)[i] = s[i]
  \,\}\bigr|.
\]
\end{definition}

\begin{definition}[Coupling set, register-machine form]\label{def:coupling-rm}
Let $P^* = (F, \tau^*, \rho^*)$ achieve $\score(P^*, s) = N$, let $\rho_0$ be a
baseline register assignment, and write $P_0 = (F, \tau^*, \rho_0)$ with
$P_0[S\!\to\!\alpha]$ the program agreeing with $P_0$ outside $S$ and assigning
$\alpha$ on $S$; let $\vbase = \score(P_0, s)$. A set of register fields
$S = \{f_1, \dots, f_k\}$ is $k$-coupled with respect to $(s, \rho_0)$ if
\begin{enumerate}[label=(\alph*)]
\item for every proper $S' \subsetneq S$ with $0 < |S'| < k$ and every
$\alpha' : S' \to R$ we have $\score(P_0[S'\!\to\!\alpha'], s) = \vbase$;
\item there exists $\alpha^* : S \to R$ with
$\score(P_0[S\!\to\!\alpha^*], s) > \vbase$.
\end{enumerate}
\end{definition}

\begin{remark}[On the coupling witness]\label{rem:witness}
Condition~(b) requires only some $S$-local assignment exceeding
$\vbase$; it does not require $\alpha^*$ to be the restriction $\rho^*|_S$. The
restriction $\rho^*|_S$ together with $\rho_0$ elsewhere may score below
$\vbase$. The coupling definition certifies a local improvement direction on
$S$ and says nothing about the extendibility of a correct solution restricted to $S$. Every
witness $\alpha^*$ lies in $Q = \{\, \alpha \in R^k : \alpha_j \ne \rho_0(f_j)
\text{ for all } j \,\}$, and $\msig = |\{\, \alpha \in Q : \score(P_0[S\!\to\!
\alpha]) > \vbase \,\}|$.
\end{remark}

\begin{definition}[$(\delta, \Delta)$-coupling set, register-machine form]\label{def:coupling-rm-relaxed}
With $P^*, \rho_0, \vbase$ as above and $\Delta > \delta \ge 0$, a set
$S = \{f_1, \dots, f_k\}$ is $(\delta, \Delta)$-$k$-coupled if
$\score(P_0[S'\!\to\!\alpha'], s) \le \vbase + \delta$ for every proper
$S' \subsetneq S$ and every $\alpha'$, and there is $\alpha^*$ with
$\score(P_0[S\!\to\!\alpha^*], s) \ge \vbase + \Delta$. This definition instantiates
\Cref{def:coupling-set-relaxed} with $X = \mathrm{Fields}(F, \tau^*)$ and $D = R$.
\end{definition}

\subsection{Score Functions and Syntactic Minimality}

\begin{definition}[Binary score for type discovery]\label{def:binary-score}
For a type assignment $\tau$ to frame $F^*$ and target $s \in \Sigma^N$,
the binary score is $\score_{\mathrm{bin}}(\tau, s) = 1$ if there exists $\rho$ with
$\exptr(F^*, \tau, \rho, N) = s$, and $0$ otherwise.
\end{definition}

\begin{definition}[Syntactic minimality]\label{def:syntactic-min}
Fix a distinguished type \textsc{identity} (no read, no write). A type
$\tau(f)$ is \emph{syntactically nonidentity} if $\tau(f) \ne \textsc{identity}$.
The \emph{syntactic minimality} $\Knt$ of $F^*$ with respect to $s$ is
\[
  \Knt := \min\bigl\{\, |S| : S \subseteq \mathrm{slots}(F^*),\ \exists\,
  \tau : S \to D \setminus \{\textsc{identity}\},\ \exists\, \rho :
  \exptr(F^*, \tau \cup \textsc{identity}_{\mathrm{rest}}, \rho, N) = s \,\bigr\},
\]
the least number of non-\textsc{identity} slots in $F^*$ able to generate $s$.
\end{definition}

\begin{remark}[Syntactic vs.\ semantic minimality]\label{rem:syn-sem}
Syntactic minimality counts non-\textsc{identity} slots regardless of whether
they affect $\exptr$ on the reachable trajectory; a semantic notion counting only
slots whose removal changes $\exptr$ can be smaller (e.g., $\textsc{add}(r_{\mathrm{acc}},
r_0, r_{\mathrm{acc}})$ with $r_0 \equiv 0$). \Cref{thm:binary-coupling} and its
corollaries use the syntactic notion, which, for a $\tau^*$ that realizes $\Knt$, is read off
$\tau^*$ by counting its non-\textsc{identity} slots, and which corresponds to the $\Knt \log_2 T$ description
cost. On a typing that realizes $\Knt$, the two notions coincide, since setting to
\textsc{identity} a slot whose removal leaves $\exptr$ unchanged would give a generating
typing with fewer non-\textsc{identity} slots.
\end{remark}

\section{Structural Coupling Analysis}

With the machine fixed, we bound its coupling width structurally,
from the typed dataflow of a program rather than from its dynamics. The bound ties the
width to the syntactic minimality $\Knt$ and, through it, to Kolmogorov
complexity.

\subsection{Binary Coupling Width and Kolmogorov Complexity}\label{sec:binary-coupling}

We first bound the binary coupling width from below by the syntactic minimality.

\begin{theorem}[Lower bound on the binary coupling width]\label{thm:binary-coupling}
Let $F^*$ be a frame with catalog $D$ of size $T \ge 2$, let $s$ be a target that the typing $\tau^*$ generates, and let the baseline $\alpha_0$
assign \textsc{identity} to every slot. Let $\Knt$ be the syntactic minimality
of $F^*$ with respect to $s$, assume $\Knt \ge 1$, and suppose that $\tau^*$ realizes it, so
$S := \{\, \mathrm{slot}_j : \tau^*(\mathrm{slot}_j) \ne \textsc{identity} \,\}$
has $|S| = \Knt$. Then $S$ is $\Knt$-coupled with respect to
$(\score_{\mathrm{bin}}, \alpha_0)$, and the coupling width $\kappa_{\mathrm{bin}}$ of $(\score_{\mathrm{bin}}, \alpha_0)$ satisfies
\[
  \kappa_{\mathrm{bin}} \ge \Knt.
\]
\end{theorem}

\begin{proof}
Since $\Knt \ge 1$, the all-\textsc{identity} baseline, which performs no computation, does not
generate $s$, so $\vbase = \score_{\mathrm{bin}}(\alpha_0) = 0$.
For every proper subset $S' \subsetneq S$, assigning catalog types to only the
slots of $S'$ leaves fewer than $\Knt$ non-\textsc{identity} slots, which by the
minimality of $\Knt$ cannot generate $s$ under any register assignment, so
$\score_{\mathrm{bin}} = 0 = \vbase$ and condition~(a) holds. The full assignment
$\tau^*$ generates $s$, so $\score_{\mathrm{bin}} = 1 > \vbase$
and condition~(b) holds. Hence $S$ is $\Knt$-coupled and
$\kappa_{\mathrm{bin}} \ge |S| = \Knt$. Since $\vbase = 0$, the local-minimum
hypothesis \textup{(\ref{eq:lmh})} is automatic and the rigid and $(0,1)$ notions
coincide here.
\end{proof}

\begin{proposition}[Equality under composition-permutation invariance]\label{prop:cpi}
Under the hypotheses of \Cref{thm:binary-coupling}, if $F^*$ and the catalog
satisfy composition-permutation invariance (the per-iteration function is
invariant under permuting the slot index set, as for a commutative
accumulator catalog), then $\kappa_{\mathrm{bin}} = \Knt$.
\end{proposition}
\begin{proof}
Under composition-permutation invariance the per-iteration function, and hence
the output, is unchanged by permuting the slot index set. The generating
assignment on the minimal $\Knt$-subset therefore transports, by a permutation of
indices, to a generating assignment on every $\Knt$-subset of slots.
Let $S$ be a candidate coupled set with $|S| > \Knt$, and pick a proper subset
$S' \subsetneq S$ with $|S'| = \Knt$. The transported assignment on $S'$
achieves $\score_{\mathrm{bin}} = 1 > 0 = \vbase$, so condition~(a) fails for
$S$. Hence no set of size exceeding $\Knt$ is coupled, and with
$\kappa_{\mathrm{bin}}\ge \Knt$ from \Cref{thm:binary-coupling},
$\kappa_{\mathrm{bin}}=\Knt$.
\end{proof}

\begin{corollary}[ISA-relative upper bound on $K(s)$]\label{cor:kol-upper}
Let $K(s)$ be the Kolmogorov complexity of $s$ on the fixed universal prefix machine.
For a register machine with catalog size $T$ and frame $F^*$, let $\tau^*$ be a minimal type
assignment, with $\Knt$ non-\textsc{identity} slots, that generates $s$ with register assignment $\rho^*$. Then
\begin{equation}\label{eq:kol-gen}
  K(s) \le \Knt \log_2 T + K(F^* \mid \Knt, N) + K(\rho^* \mid F^*, \tau^*) + O(\log N),
\end{equation}
where $K(F^* \mid \Knt, N)$ is the description length of the frame given the slot
count and length, and $K(\rho^* \mid F^*, \tau^*)$ that of the register
assignment.
\end{corollary}
\begin{proof}
Concatenate descriptions. A minimal generating program is specified by its
frame ($K(F^*\mid \Knt,N)$ bits), its $\Knt$ non-\textsc{identity} slot types ($\Knt\log_2 T$ bits), and its register assignment ($K(\rho^*\mid F^*,\tau^*)$ bits), plus
$O(\log N)$ for lengths; running it outputs $s$, so $K(s)$ is at most this
total.
\end{proof}

\begin{corollary}[Simplified form]\label{cor:kol-simplified}
If $K(F^* \mid \Knt, N) = O(1)$ and $K(\rho^* \mid F^*, \tau^*) = O(\log N)$, as when the frame is fixed by
the slot count and the length and the register assignment by the frame and the
types, then
\begin{equation}\label{eq:kol-simple}
  \kappa_{\mathrm{bin}} \log_2 T \ge K(s) - O(\log N). \tag{$*$}
\end{equation}
\end{corollary}
\begin{proof}
Specializing \Cref{cor:kol-upper} to these two terms gives $K(s)\le \kappa_{\mathrm{bin}}\log_2 T+O(\log N)$
since $\Knt\le\kappa_{\mathrm{bin}}$ by \Cref{thm:binary-coupling}.
\end{proof}

\begin{remark}[The relationship is one-directional]\label{rem:one-directional}
The converse $\Knt \log_2 T \le O(K(s))$ does not hold, even for fixed $T$. In the
class of \Cref{prop:zp}, the target is determined by $x_0$, $\Knt(T-1)$, $p$ and $N$, and by Bertrand's postulate $p$ can be
taken below $2\Knt(T-1)$, so $K(s) = O(\log \Knt + \log N)$ while $\Knt \log_2 T$ is linear in
$\Knt$. For $N$ polynomial in $\Knt$, the ratio of $\Knt \log_2 T$ to $K(s)$ is unbounded.
The representation cost $\Knt \log_2 T$ bounds $K(s)$ from above (more structure to
describe means a larger description), and $K(s)$ bounds search from below,
so the two compose in one direction only. Equality $\kappa_{\mathrm{bin}} \log_2 T
\asymp K(s)$ fails in general, already with a single witness as in that class. Witness multiplicity also forces a gap. When $\msig$ distinct minimal programs compute the same compressible
$s$ and $K(F^*, \rho^* \mid \Knt, N) = O(\log N)$, \Cref{rem:Ms} and \Cref{thm:binary-coupling} give $\kappa_{\mathrm{bin}} \log_2 T - K(s) \ge \log_2 \msig - O(\log N)$.
\end{remark}

\begin{remark}[The match-count score in program discovery]\label{rem:3-2f-3}
For the register-machine match score
$\score_{\mathrm{match}}(\tau, \rho) := |\{\, t \in [N] :
\exptr(F^*, \tau, \rho, N)[t] = s[t] \,\}|$, the score is computed from the target
given the frame, so
$K(\score_{\mathrm{match}} \mid F^*, N, \alpha_0) \le K(s \mid F^*, N) + O(\log N)$.
Kolmogorov-typicality is therefore not automatic. In the unique-witness regime
$\msig = 1$ it requires
$K(s \mid F^*, N) \ge \kappa \log_2(d-1) - c - O(\log N)$, i.e., a target at least
as complex, given its frame, as the coupling budget. Every result that uses it
states it as a hypothesis.
\end{remark}

\begin{remark}[Connection to Levin's bound]\label{rem:levin}
Under the hypotheses of \Cref{thm:cons-budget}, the binary score with
$\kappa_{\mathrm{bin}} \ge \Knt$ (\Cref{thm:binary-coupling}), and the ISA-relative
Kolmogorov bound \eqref{eq:kol-gen} with $K(F^* \mid \Knt, N) = O(1)$ and $K(\rho^* \mid F^*, \tau^*) = O(\log N)$, the coupling budget satisfies
$B \ge \Knt \log_2(T-1) \ge K(s) - \Knt \log_2\tfrac{T}{T-1} - O(\log N)$. The
bound \eqref{eq:cons-unique} then gives
\[
  K(\varphi, \psi) + \log_2 \operatorname{search}
  \ge K(s) - \Knt \log_2\tfrac{T}{T-1} - O(\log N),
\]
i.e., the total bit cost of solving the type-discovery CSP is at least $K(s)$ up to a
base correction of $\log_2\tfrac{T}{T-1}$ bits per nonidentity slot, at most
$\Knt/((T-1)\ln 2)$ bits in total, regardless of reformulation. Since $B = \kappa_{\mathrm{bin}} \log_2(T-1) = \tfrac{\log_2(T-1)}{\log_2 T}\,\kappa_{\mathrm{bin}} \log_2 T$,
the simplified bound \eqref{eq:kol-simple} also gives the floor $\tfrac{\log_2(T-1)}{\log_2 T}\,K(s) - O(\log N)$
directly, the form quoted in the main text. This floor needs no condition on how the
typing codes the target, and since $K(s) \le \Knt \log_2 T + O(\log N)$ it is never
weaker than the additive form beyond the $O(\log N)$ term. Levin search runs within a constant factor of
$2^{K(s)} t(p^*)$ steps, with $p^*$ a shortest program for $s$ and $t(p^*)$ its running time, and it saturates the floor asymptotically ($T \to \infty$) at $K(\varphi,
\psi) = O(1)$. For finite $T$ the floor is
$c_T^{\,K(s) - O(\log N)}$ with $c_T = (T-1)^{1/\log_2 T} < 2$, which Levin's
$2^{K(s)}$ exceeds by the factor $(2/c_T)^{K(s)}$.
\end{remark}

\subsection{Typed Dataflow and Structural Sufficiency}

Throughout this section the baseline program is $P_0 = (F^*, \tau_0, \rho^*)$, where
$F^*$ has exactly $\Knt$ searchable slots, $\tau_0$ assigns \textsc{identity} to each of
them, and the register assignment stays at $\rho^*$. The frame thus has no slot beyond a minimal
generating set, and every result of this section is stated for such frames. The program $P_0[S'\!\to\!\beta]$ types the
slots of $S'$ by $\beta$ and leaves the others at \textsc{identity}, and
$\exptr(\tau, N)[t]$ denotes the $t$th output symbol under typing
$\tau \in D^\Knt$ with registers $\rho^*$.

\begin{definition}[Typed dataflow graph]\label{def:dataflow}
For $\tau \in D^\Knt$, the \emph{typed dataflow graph} is
$\mathcal{D}(P_0, \tau) = (V, E_\tau)$ with $V = \{1, \dots, \Knt\} \cup \{\mathrm{OUT}\}$
and $(f, g) \in E_\tau$ if and only if some register read by $\tau(g)$ is written by
$\tau(f)$ under some execution trace (a static overapproximation). The edges
$(f, \mathrm{OUT})$ are defined analogously with $\mathrm{OUT}$ reading the
output register. Write $\mathcal{D}(P^*) := \mathcal{D}(P_0, \tau^*)$.
\end{definition}

\begin{definition}[Type-conservative dependency set]\label{def:dep}
Let $\mathcal{D}_N(P_0, \tau)$ be $\mathcal{D}(P_0, \tau)$ unrolled over $N$ iterations with
loop-carried edges from iteration $i$ to $i+1$ whenever a written register is
read next. For $t \in [N]$,
\[
  \mathrm{dep}(t) = \{\, f \in \{1, \dots, \Knt\} : \exists\, \tau \in D^\Knt
  \text{ such that some } f^{(i)},\, i \le t, \text{ has a directed path in }
  \mathcal{D}_N(P_0, \tau) \text{ to } \mathrm{OUT}^{(t)} \,\},
\]
i.e., the slots that some catalog typing connects to the output at position $t$.
The $\tau^*$-specific set $\mathrm{dep}^{\tau^*}(t)$, defined as $\mathrm{dep}(t)$ with $\tau$ fixed to $\tau^*$, is contained in
$\mathrm{dep}(t)$.
\end{definition}

\begin{proposition}[Single-slot stability]\label{prop:single-slot}
Let $f \notin \mathrm{dep}(t)$. For every pair $\tau, \tau' \in D^\Knt$ agreeing
off $\{f\}$, we have $\exptr(\tau, N)[t] = \exptr(\tau', N)[t]$. In particular, every replacement of
$\tau^*(f)$ by a catalog type preserves $\exptr(P^*, N)[t] = s[t]$.
\end{proposition}

\begin{proof}
Since $f \notin \mathrm{dep}(t)$, no catalog typing gives $f$ a path to
$\mathrm{OUT}^{(t)}$. For every $\sigma$ fixing the slots $g \ne f$ to the common
value and choosing $\sigma(f)$ arbitrarily, $\exptr(\sigma, N)[t]$ depends only on
register writes along paths reaching $\mathrm{OUT}^{(t)}$, none through $f$;
hence $\sigma(f)$ is immaterial. Applying this to $\sigma = \tau$ and
$\sigma = \tau'$ gives the claim.
\end{proof}

\begin{lemma}[Structural sufficiency]\label{lem:struct-suff}
For proper $S' \subsetneq \{1, \dots, \Knt\}$ and $\beta : S' \to D$, suppose
(i) $\mathrm{dep}(t) \subseteq S'$ and (ii)
$\beta|_{\mathrm{dep}(t)} = \tau^*|_{\mathrm{dep}(t)}$. Then
$\exptr(P_0[S'\!\to\!\beta], N)[t] = s[t]$.
\end{lemma}

\begin{proof}
Let $\tau_\beta$ type $P_0[S'\!\to\!\beta]$ ($\beta$ on $S'$, \textsc{identity}
off $S'$) and let $A := \{1, \dots, \Knt\} \setminus \mathrm{dep}(t)$. By (i),
$\{1, \dots, \Knt\} \setminus S' \subseteq A$; by (ii), $\tau_\beta$ and $\tau^*$
agree off $A$. Enumerate $A = \{g_1, \dots, g_{|A|}\}$ and form the hybrids
$\tau^{(0)} = \tau^*$ and $\tau^{(j)}$, which is $\tau^{(j-1)}$ with $g_j$ reset to
$\tau_\beta(g_j)$. Each $g_j$ lies in $A$, so $g_j \notin \mathrm{dep}(t)$, and
\Cref{prop:single-slot} gives
$\exptr(\tau^{(j-1)}, N)[t] = \exptr(\tau^{(j)}, N)[t]$. Chaining,
$s[t] = \exptr(\tau^*, N)[t] = \exptr(\tau^{(|A|)}, N)[t] = \exptr(\tau_\beta, N)[t]$,
since $\tau^{(|A|)} = \tau_\beta$.
\end{proof}

\begin{definition}[Match decomposition]\label{def:match-decomp}
For a set of slots $S'$ and $\beta : S' \to D$, the match count is $\mathrm{mc}(P^*, \beta, S') = |\{\, t :
\exptr(P_0[S'\!\to\!\beta], N)[t] = s[t] \,\}|$, its structural part is
\[
  \mathrm{mc}_{\mathrm{struct}}(P^*, \beta, S') = \bigl|\{\, t : \mathrm{dep}(t)
  \subseteq S' \text{ and } \beta|_{\mathrm{dep}(t)} =
  \tau^*|_{\mathrm{dep}(t)} \,\}\bigr|,
\]
i.e., the number of positions $t$ that satisfy conditions (i) and (ii) of \Cref{lem:struct-suff}, and its coincidental part is
$\mathrm{coinc}(P^*, \beta, S') = \mathrm{mc} - \mathrm{mc}_{\mathrm{struct}}$.
\end{definition}

\subsection{Essential Engagement}

\begin{definition}[Essential engagement]\label{def:essential}
The program $P^*$ is \emph{essentially engaged with parameter} $\varepsilon \in [0, 1)$ if
every slot $f$ satisfies $|\{\, t \in [N] : f \notin \mathrm{dep}(t) \,\}| \le
\varepsilon N$.
\end{definition}

\begin{proposition}[Structural corruption bound]\label{prop:struct-corruption}
If $P^*$ is essentially engaged with parameter $\varepsilon$, then every
proper $S' \subsetneq \{1, \dots, \Knt\}$ satisfies
$|\{\, t : \mathrm{dep}(t) \subseteq S' \,\}| \le \varepsilon N$.
\end{proposition}

\begin{proof}
Pick $f^* \in \{1, \dots, \Knt\} \setminus S'$ (nonempty as $S'$ is proper). If
$\mathrm{dep}(t) \subseteq S'$ then $f^* \notin \mathrm{dep}(t)$, so the count
is at most $|\{\, t : f^* \notin \mathrm{dep}(t) \,\}| \le \varepsilon N$.
\end{proof}

\begin{definition}[Target frequency ceiling]\label{def:freq-ceiling}
The target frequency ceiling of $s$ is $\mu(s) = \max_{\sigma \in \Sigma} |\{\, t \in [N] : s[t] = \sigma \,\}| / N$,
i.e., the largest fraction of positions of $s$ that hold the same symbol.
\end{definition}

\begin{definition}[Coincidence-boundedness]\label{def:coincidence}
The program $P^*$ is \emph{coincidence-bounded against $s$ with residual} $\gamma : \mathbb{N}
\to \mathbb{N}$ if $\mathrm{coinc}(P^*, \beta, S') \le \mu(s) N + \gamma(N)$ for every proper
$S'$ and every $\beta$, with $\gamma(N) = o(N)$.
\end{definition}

\begin{lemma}[Wire Corruption Lemma]\label{lem:wire-corruption}
Let $P^*$ be essentially engaged with parameter $\varepsilon$ and coincidence-bounded against $s$ with residual $\gamma$. For every proper $S' \subsetneq \{1, \dots, \Knt\}$ and every $\beta : S' \to D$,
the match count satisfies $\mathrm{mc}(P^*, \beta, S') = \mathrm{mc}_{\mathrm{struct}} + \mathrm{coinc} \le
\varepsilon N + \mu(s) N + \gamma(N)$, where the structural term is bounded by
\Cref{prop:struct-corruption} and the coincidence term by
\Cref{def:coincidence}.
\end{lemma}
\begin{proof}
By \Cref{def:match-decomp}, $\mathrm{mc} = \mathrm{mc}_{\mathrm{struct}} + \mathrm{coinc}$. The structural term
counts the iterations with both $\mathrm{dep}(t)\subseteq S'$ and agreement on
$\mathrm{dep}(t)$, so it is at most $|\{\, t : \mathrm{dep}(t)\subseteq S'
\,\}| \le \varepsilon N$ by \Cref{prop:struct-corruption}; the coincidence term
is at most $\mu(s)N + \gamma(N)$ by \Cref{def:coincidence}. Summing gives the
bound.
\end{proof}

\subsection{The Accumulator-Driven Class}

\begin{definition}[Accumulator-driven program]\label{def:accumulator}
The program $P^*$ is \emph{accumulator-driven} on register $r_{\mathrm{acc}}$ if
\begin{enumerate}[label=(\roman*)]
\item every slot $f$ implements an update $r_{\mathrm{acc}} \leftarrow g_f(r_{\mathrm{acc}},
\dots)$ that (i.a) reads and writes $r_{\mathrm{acc}}$ and (i.b) induces a
nonconstant single-variable map on $r_{\mathrm{acc}}$ for some reachable
auxiliary context;
\item \textsc{out} reads $r_{\mathrm{acc}}$;
\item slots execute in canonical order $1, \dots, \Knt$ then \textsc{out}, once per iteration.
\end{enumerate}
\end{definition}

\begin{proposition}[Essential engagement of accumulator-driven programs]\label{prop:ee-accumulator}
An accumulator-driven program satisfies \Cref{def:essential} with
$\varepsilon = 0$.
\end{proposition}

\begin{proof}
Fix a slot $f$. We show that $f \in \mathrm{dep}(t)$ for every $t$, via a path in $\mathcal{D}_N(P_0, \tau^*)$ from some $f^{(i)}$ with $i \le t$ to $\mathrm{OUT}^{(t)}$. For $t = 0$, by
(i.a) every consecutive pair $(g, g+1)$ with $f \le g < \Knt$ gives an edge
$g^{(0)} \to (g+1)^{(0)}$ through $r_{\mathrm{acc}}$, and by (ii)
$\Knt^{(0)} \to \mathrm{OUT}^{(0)}$; concatenating yields
$f^{(0)} \to \cdots \to \Knt^{(0)} \to \mathrm{OUT}^{(0)}$, so
$f \in \mathrm{dep}(0)$. For $t \ge 1$, the loop-carried edge
$\Knt^{(t-1)} \to 1^{(t)}$ (slot $\Knt$ writes $r_{\mathrm{acc}}$ at $t-1$, slot $1$
reads it at $t$) extends the path across iterations to $\mathrm{OUT}^{(t)}$.
Hence $|\{\, t : f \notin \mathrm{dep}(t) \,\}| = 0$, and the uniform bound is
$\varepsilon = 0$.
\end{proof}

\begin{definition}[Register-connected program]\label{def:register-connected}
The program $P^*$ is \emph{register-connected} if in $\mathcal{D}(P^*)$ every slot $f$ has a directed
path to \textsc{out}.
\end{definition}

\begin{remark}\label{rem:reg-conn}
Every accumulator-driven program is register-connected; the converse fails (a
program with three slots routing through distinct registers can be connected without a
common accumulator). Every minimal program (no slot replaceable by
\textsc{identity}) is register-connected, since a dead-code slot would
contradict minimality.
\end{remark}

\begin{proposition}[Essential engagement of register-connected programs]\label{prop:ee-connected}
A register-connected program with a uniform-footprint catalog, one whose types all read and write the same registers in each slot, satisfies
\Cref{def:essential} with $\varepsilon = \Knt/N$ whenever $N > \Knt$.
\end{proposition}
\begin{proof}
By register-connectedness every slot $f$ has a directed path to \textsc{out} in
$\mathcal{D}(P^*)$. Under a uniform-footprint catalog this path is present for every
typing (the footprint, hence the dataflow edges, is type-invariant). Take the path
simple. It has at most $\Knt$ edges, and in $\mathcal{D}_N(P_0, \tau^*)$ each edge stays in its
iteration or passes to the next, so the path runs from $f^{(t-\ell)}$ to $\mathrm{OUT}^{(t)}$
for every $t \ge \ell$, with a lag $\ell \le \Knt$, as in \Cref{prop:ee-accumulator}. Hence
$f \in \mathrm{dep}(t)$ for all but at most $\Knt$ values of $t$, and $\varepsilon = \Knt/N$.
\end{proof}

\subsection{Structural Coupling and Emergence}

\begin{definition}[Structurally coupled program]\label{def:struct-coupled}
The program $P^*$ is \emph{structurally coupled against $s$ with parameters
$(\varepsilon, \gamma)$} if it is essentially engaged with parameter
$\varepsilon$ (\Cref{def:essential}) and coincidence-bounded against $s$ with
residual $\gamma$ (\Cref{def:coincidence}).
\end{definition}

For structurally coupled programs with $\varepsilon + \mu(s) < 1$, the match score with the registers fixed at $\rho^*$ couples all $\Knt$ slots, with a margin linear in $N$.

\begin{theorem}[$k$-growth]\label{thm:k-growth}
Let $P^* = (F^*, \tau^*, \rho^*)$ generate $s \in \Sigma^N$ and be structurally
coupled against $s$ with parameters $(\varepsilon, \gamma)$, where
$\varepsilon + \mu(s) < 1$. Let $\vbase = \mathrm{mc}(P^*, \emptyset, \emptyset) =
|\{\, t \in [N] : s[t] = x_0 \,\}|$ be the match count of the
all-\textsc{identity} baseline $\alpha_0$ ($x_0$ the initial register value), and note that
$\vbase \le \mu(s) N$. Then the set $S$ of $\Knt$ non-\textsc{identity} slots is
$(\delta_0, \Delta_0)$-$\Knt$-coupled with respect to
$(\score_{\mathrm{match}}^{\rho^*}, \alpha_0)$, where
$\score_{\mathrm{match}}^{\rho^*}(\tau) = |\{\, t \in [N] : \exptr(F^*, \tau, \rho^*, N)[t] =
s[t] \,\}|$ is the match count with the register assignment fixed at $\rho^*$, with
\[
  \delta_0 = (\varepsilon + \mu(s)) N + \gamma(N) - \vbase, \qquad
  \Delta_0 = N - \vbase,
\]
and correction margin
\[
  \Delta_0 - \delta_0 = (1 - \varepsilon - \mu(s)) N - \gamma(N) = \Omega(N).
\]
In particular $\kappa_{\mathrm{match}}^{(\delta_0, \Delta_0)} \ge \Knt$.
\end{theorem}

\begin{proof}
For every proper $S' \subsetneq S$ and every $\beta$, the Wire Corruption Lemma
(\Cref{lem:wire-corruption}) bounds the absolute match count by
$\mathrm{mc}(P^*, \beta, S') \le \varepsilon N + \mu(s) N + \gamma(N) = \vbase + \delta_0$,
which is condition~(a$'$) at threshold $\vbase + \delta_0$. The full assignment
$\tau^*$ achieves $\score_{\mathrm{match}}^{\rho^*} = N = \vbase + \Delta_0$, which is
condition~(b$'$). Both thresholds are stated relative to $\vbase$, and the
margin $\Delta_0 - \delta_0 = (1 - \varepsilon - \mu(s)) N - \gamma(N)$ is
$\vbase$-invariant (the $-\vbase$ cancels). The margin is $\Omega(N)$ since
$\varepsilon + \mu(s) < 1$ and $\gamma(N) = o(N)$. So $S$ is
$(\delta_0, \Delta_0)$-$\Knt$-coupled. Two regimes arise. For baselines
whose output never agrees with $s$, the baseline count is $\vbase = 0$ and $\delta_0 = (\varepsilon+\mu(s))N +
\gamma(N)$. For the graded-translation class of
\Cref{prop:zp}, where $\varepsilon = 0$, $\gamma = 0$ and $\mu(s)N = \lceil N/p
\rceil$, the baseline count is $\vbase = \lfloor N/p \rfloor$ and $\delta_0 = \lceil N/p \rceil - \lfloor N/p \rfloor \le 1$. The margin is the same in both.
\end{proof}

\begin{remark}[The treewidth bound needs a stronger margin than the search bound]\label{rem:tw-margin}
Applying \Cref{thm:coupling-tw} ($\kappa \le \twscore + 1$) to the relaxed coupling of the match score requires $\Delta_0 > 2\delta_0$, which is stronger than the $\Omega(N)$ margin above. With $\vbase = 0$ (so $\Delta_0 = N$ and $\delta_0 = (\varepsilon + \mu(s)) N + \gamma(N)$), the condition $\Delta_0 > 2\delta_0$ holds if and only if $\varepsilon + \mu(s) < \tfrac12 - \gamma(N)/N$, which is strictly stronger than the condition $\varepsilon + \mu(s) < 1$ used by the search lower bound (\Cref{cor:search-struct}). The proven classes, the binary score and graded translation (\Cref{prop:zp}), have $\delta = 0$, so $\Delta > 2\delta$ holds and the treewidth bound is unconditional there. Only the general match-score statement carries the stronger threshold.
\end{remark}

\begin{corollary}[Exponential search lower bound for structurally coupled programs]\label{cor:search-struct}
Under the hypotheses of \Cref{thm:k-growth} with $\varepsilon + \mu(s) < 1$, score-oracle search
over type assignments requires $\Omega\bigl((T-1)^\Knt / \msig\bigr)$ queries
against the worst-case $(\delta_0, \Delta_0)$-consistent score, by clause~(III$'$) of \Cref{thm:barrier}
and \Cref{cor:yao}.
\end{corollary}
\begin{proof}
\Cref{thm:k-growth} gives a $(\delta_0,\Delta_0)$-$\Knt$-coupled set with margin
$\Omega(N)$; apply clause~(III$'$) and \Cref{cor:yao} with $d=T$.
\end{proof}

\begin{corollary}[Coupling emergence for minimal programs]\label{cor:emergence}
Let $P^* = (F^*, \tau^*, \rho^*)$ be a syntactically minimal program
(\Cref{def:syntactic-min}), let $\vbase$ be the all-\textsc{identity} match count
of \Cref{thm:k-growth}, and let $m_{\max}$ be the largest match count under $\rho^*$
of an assignment that gives catalog types to a proper subset of the $\Knt$
slots and leaves the others at \textsc{identity}. Then the fixed-$\rho^*$ match-score coupling
width satisfies
\[
  \kappa_{\mathrm{match}}^{\rho^*, (\delta, \Delta)} \ge \Knt, \qquad
  \delta = \max(m_{\max}, \vbase) - \vbase, \quad \Delta = N - \vbase,
\]
using only minimality. Its content is that no proper subset of the slots reaches
match count $N$, so every generating typing lies in $Q$. The margin
$\Delta - \delta$ can be as small as $1$; a search bound needs the margin of
\Cref{thm:k-growth} and control of $\msig$.
\end{corollary}
\begin{proof}
By syntactic minimality an assignment with fewer than $\Knt$ non-\textsc{identity}
slots generates $s$ under no register assignment, in particular not under
$\rho^*$, so its match count is below $N$. The same holds for the baseline, so $\max(m_{\max}, \vbase) < N$ and the
margin $\Delta - \delta \ge 1$. Conditions~(a$'$) and~(b$'$) of
\Cref{def:coupling-set-relaxed} then hold for $S = \{1, \dots, \Knt\}$ at the
all-\textsc{identity} baseline, as in \Cref{thm:k-growth}, with $\tau^*$ the
witness of~(b$'$).
\end{proof}

\section{Dynamical Coupling Analysis}

This section establishes ISA-specific dynamical properties. It covers adversarial
decorrelation, the provable class, the search lower bound, and the connection
to algorithmic information theory.

\subsection{Adversarial Decorrelation}

Every register machine whose register values are bounded in absolute value by a
constant $\mathrm{SAT}$ operates on a finite state space of size at most $(2\,\mathrm{SAT}+1)^W$. The deterministic
update makes every trajectory eventually periodic, so $s = \exptr(P^*, N)$ has
a transient and a period, and so does the output
$q$ of each partial program. Both transients are bounded by $(2\,\mathrm{SAT}+1)^W$, a constant
independent of $N$.

\begin{definition}[Adversarial decorrelation, per-pair]\label{def:ad-perpair}
An accumulator-driven minimal program $P^*$ with output $s = \exptr(P^*, N)$ is
\emph{per-pair adversarially decorrelated} if for every proper $S' \subsetneq
\{1, \dots, \Knt\}$ and every $\beta : S' \to D$, writing $q = \exptr(P_0[S'\!\to\!
\beta], N)$,
\[
  \limsup_{N \to \infty} \Bigl( \frac{1}{N} \bigl|\{\, t \in [N] : s[t] = q[t] \,\}\bigr|
  - \mu(s) \Bigr) \le 0.
\]
\end{definition}

\begin{definition}[Uniform adversarial decorrelation]\label{def:ad-uniform}
The program $P^*$ is \emph{uniformly adversarially decorrelated} with residual
$\gamma : \mathbb{N} \to \mathbb{N}$, where $\gamma(N) = o(N)$, if for every $N$, every
proper $S'$, and every $\beta$,
\[
  \bigl|\{\, t \in [N] : s[t] = \exptr(P_0[S'\!\to\!\beta], N)[t] \,\}\bigr|
  \le \mu(s) N + \gamma(N),
\]
with $\gamma(N)$ depending only on $N$.
\end{definition}

\begin{remark}[Per-pair does not imply uniform]\label{rem:ad-gap}
Per-pair adversarial decorrelation (AD) gives one $\limsup$ per pair; uniform AD requires a single
$\gamma(N)$ bounding all $2^\Knt T^\Knt$ pairs simultaneously. Without a structural
bound on the convergence rate the supremum over pairs can grow like $N$,
and then $\gamma(N) = o(N)$ fails. Uniform AD is the hypothesis that
\Cref{thm:coincidence} needs, and it does not follow from the per-pair version.
\end{remark}

\begin{theorem}[Coincidence bound via uniform AD]\label{thm:coincidence}
Let $P^*$ be accumulator-driven, minimal, and uniformly adversarially
decorrelated with residual $\gamma(N) = o(N)$. Then for every proper $S'$ and
every $\beta$,
\[
  \mathrm{coinc}(P^*, \beta, S') \le \mu(s) N + \gamma(N),
\]
so $P^*$ is coincidence-bounded (\Cref{def:coincidence}) with the same residual.
\end{theorem}

\begin{proof}
Fix $(S', \beta)$ and let $\mathrm{mc} = \mathrm{mc}(P^*, \beta, S')$. Uniform AD
(\Cref{def:ad-uniform}) gives $\mathrm{mc} \le \mu(s) N + \gamma(N)$. Since
$\mathrm{mc}_{\mathrm{struct}} \ge 0$ and $\mathrm{coinc} = \mathrm{mc} - \mathrm{mc}_{\mathrm{struct}} \le \mathrm{mc}$, the same
bound holds for $\mathrm{coinc}$.
\end{proof}

\subsection{The Provable Class of Graded Translations on $\mathbb{Z}/p$}

We now exhibit a class of programs that is uniformly adversarially decorrelated with zero residual.

\begin{proposition}[Graded translation on $\mathbb{Z}/p$ satisfies AD for every $T \ge 2$]\label{prop:zp}
Let $\Knt \ge 2$, $T \ge 2$, and let $p$ be a prime with $p > \Knt(T-1)$. Fix the
catalog $D = \{\textsc{add}_0, \textsc{add}_1, \dots, \textsc{add}_{T-1}\}$,
where $\textsc{add}_j$ sends $r_{\mathrm{acc}} \mapsto r_{\mathrm{acc}} + \xi_j
\bmod p$ with shift $\xi_j = j$. The catalog has a uniform footprint, and
$\textsc{add}_0$ is semantically \textsc{identity} (it computes $v \mapsto v + 0
= v$ while reading and writing $r_{\mathrm{acc}}$ with the same footprint as
every other $\textsc{add}_j$), so we take it as the distinguished baseline type
of \Cref{def:syntactic-min}, and the count of nonbaseline types is $T-1$. Fix
$\tau^* = (\textsc{add}_{T-1}, \dots, \textsc{add}_{T-1})$ of length $\Knt$ and an
initial value $x_0 \in \mathbb{Z}/p$. Since the output follows the slots,
$s[t] = x_0 + (t+1)\Knt(T-1) \bmod p$ for $t \in [N]$. Then $P^*$ is
accumulator-driven, essentially engaged with $\varepsilon = 0$, minimal, and
uniformly adversarially decorrelated with residual $\gamma(N) = 0$. Every proper
perturbation has the same agreement count,
\[
  \bigl|\{\, t \in [N] : s[t] = q(S', \beta)[t] \,\}\bigr|
  = \lfloor N/p \rfloor \le \mu(s) N
  \quad\text{for every proper } S' \text{ and every } \beta.
\]
Moreover $\msig = 1$, and the match score belongs to the family
$\mathcal{A}(S, 1)$ of \Cref{cor:yao}, with $S$ the $\Knt$ slots. Over this family,
match-score-oracle search requires
$\mathbb{E}[\#\text{queries}] \ge \Omega\bigl((T-1)^\Knt\bigr)$ (\Cref{cor:yao}(b)),
exponential in $\Knt$ for $T \ge 3$.
\end{proposition}

\begin{proof}
First, the composed single-iteration map is $x \mapsto x + \xi^* \bmod p$ with
$\xi^* = \Knt(T-1)$
(coefficients sum in the abelian setting of pure translations), so
$s[t] = x_0 + (t+1)\Knt(T-1) \bmod p$. Since $p$ is prime and $0 < \Knt(T-1) < p$, we have
$\gcd(\Knt(T-1), p) = 1$, so the orbit visits every residue once per period $p$.
Over $[N]$ each residue is hit $\lfloor N/p \rfloor$ or $\lceil N/p \rceil$
times, so $\mu(s) = \lceil N/p \rceil / N$, with asymptotic rate $1/p$.

Second, $P^*$ is accumulator-driven and essentially engaged with $\varepsilon = 0$. Each slot applies
$\textsc{add}_{T-1}$, which reads and writes $r_{\mathrm{acc}}$
(\Cref{def:accumulator}(i.a)) and induces the bijection
$v \mapsto v + (T-1) \bmod p$, nonconstant for $T \ge 2$ (i.b). With
\textsc{out} reading $r_{\mathrm{acc}}$ each iteration and canonical order,
\Cref{prop:ee-accumulator} gives $\varepsilon = 0$.

Third, $P^*$ is minimal. A partial program $P_0[S'\!\to\!\beta]$ with $|S'| < \Knt$ has
composed map $x \mapsto x + \xi' \bmod p$ with $\xi' = \sum_{f \in S'} \xi_{\beta(f)}$.
Since each $\xi_{\beta(f)} \le T-1$,
\[
  \xi' \le (\Knt-1)(T-1) < \Knt(T-1) = \xi^*,
\]
and both $\xi^*, \xi' \in [0, \Knt(T-1)] \subset [0, p)$, so $\xi^* \equiv \xi' \pmod p$ if and only if
$\xi^* = \xi'$, which fails; hence the partial map differs from the full map
$x \mapsto x + \xi^*$.

Fourth, we prove uniform AD. With $q[t] = x_0 + (t+1)\xi' \bmod p$, the condition $s[t] = q[t]$
is $(t+1)(\xi^* - \xi') \equiv 0 \pmod p$ where $\xi^* - \xi' \in [T-1, \Knt(T-1)]$. In all
cases $0 < \xi^* - \xi' < p$ and $p$ is prime, so $\xi^* - \xi'$ is a unit and the condition reduces to $t + 1 \equiv 0 \pmod p$. Over
$[N]$ the agreement count is $\lfloor N/p \rfloor \le \lceil N/p \rceil = \mu(s) N$,
independent of $(S', \beta)$. Hence uniform AD holds with
$\gamma(N) = 0$, and the $\limsup$ rate is $1/p$.

Fifth, we count the witnesses. A full assignment $\alpha \in D^\Knt$ has composed shift
$\sum_{f=1}^\Knt \xi_{\alpha(f)}$, equal to $\xi^* = \Knt(T-1)$ in $[0,p)$ if and only if every
$\xi_{\alpha(f)} = T-1$, i.e., $\alpha = \tau^*$; so $\msig = 1$.

Finally, every proper perturbation scores exactly $\vbase$ (uniform
AD above) and only $\tau^*$ scores above it, so the score belongs to
$\mathcal{A}(S, 1)$, and \Cref{cor:yao}(b) gives
$\mathbb{E}[\#\text{queries}] \ge \Omega((T-1)^\Knt / \msig) = \Omega((T-1)^\Knt)$
over that family, exponential in $\Knt$ for $T \ge 3$.
\end{proof}

\begin{remark}[Rigid coupling on this class]\label{rem:zp-rigid}
Since $\vbase = \lfloor N/p \rfloor$ and every proper perturbation
scores exactly $\vbase$, condition~(a) of \Cref{def:coupling-set} holds, so the
coupling on this class is rigid, a stronger property than relaxed coupling at $\delta = 0$. The match
score is two-valued (every nongenerating assignment scores $\vbase$, only
$\tau^*$ scores $N$), so it is a member of the family $\mathcal{A}(S, 1)$ of
\Cref{cor:yao}. The lower bound is a statement about that family, not about a
single instance, since an algorithm told $\tau^*$ in advance needs one query.
This class shows that the flat basin the bound assumes occurs in
actual programs. The class itself is not hard, since within it the
generating assignment is always the same, every slot taking the largest shift.
For register inference the worst case is attained by a class of actual programs
(\Cref{prop:horner}), and at every coupling width by a second class
(\Cref{prop:block}).
\end{remark}

\begin{corollary}[Lower bound for an oracle on partial programs]\label{cor:partial-oracle}
For $P^*$ as in \Cref{prop:zp}, every algorithm identifying $\tau^*$ via an oracle
returning $\mathrm{mc}(P^*, \beta, S')$ for arbitrary $(S', \beta)$ (including proper $S'$)
requires $\mathbb{E}[\#\text{queries}] \ge \Omega((T-1)^\Knt)$ against the
worst-case oracle in the family $\mathcal{A}(S, 1)$ to which this class
belongs. On this class partial-assignment queries return no more information than full
ones.
\end{corollary}
\begin{proof}
On this class every nongenerating query, full or partial, returns exactly
$\lfloor N/p \rfloor$ and only $\tau^*$ returns $N$ (\Cref{rem:zp-rigid}), so the
partial-assignment oracle is two-valued and information-equivalent to the
full-assignment one; the $\Omega((T-1)^\Knt)$ bound of \Cref{cor:yao}(b) therefore
applies unchanged.
\end{proof}

\begin{remark}[Scope and extension]\label{rem:zp-scope}
The discharge covers the additive sub-ISA $\{\textsc{add}_j\}$ on
$\mathbb{Z}/p$. Extension to catalogs with multiplicative, modular, or bitwise
operations is open. In the full affine fragment every affine target $x \mapsto a_1 x + a_0$ factors as
a composition of two slots, so minimality fails for $\Knt \ge 3$. The average-case route is also obstructed.
Collapsing catalogs have exponentially large typical $\msig$, so the
unconditional bound survives only at extremal instances such as this one.
\end{remark}

For register inference, where the operation types are fixed and the register
fields are searched (domain $R$, so $d = W$), the worst case is attained by actual
programs.

\begin{proposition}[Binary Horner programs]\label{prop:horner}
Fix $k, N \ge 1$. Take $W = 3$ registers $R_0, R_1, r_{\mathrm{acc}}$, set by
preloop loads to $0$, $1$ and $x_0$, and the loop body
$\textsc{dbl}\ \textsc{sub}_1\ \textsc{dbl}\ \textsc{sub}_2 \cdots \textsc{dbl}\
\textsc{sub}_k\ \textsc{out}$, where $\textsc{dbl}$ sets $r_{\mathrm{acc}} \leftarrow
2 r_{\mathrm{acc}}$, $\textsc{sub}_j$ sets $r_{\mathrm{acc}} \leftarrow r_{\mathrm{acc}} -
R[f_j]$ for its source field $f_j$, and $\textsc{out}$ emits $r_{\mathrm{acc}}$, with
arithmetic in $\mathbb{Z}/p$ as in \Cref{prop:zp}. The operation types are fixed;
the register fields are the destination and source fields of every operation and
the destination fields of the loads, $F$ of them. Let $S = \{f_1, \dots, f_k\}$, and for $\alpha^* \in \{R_0, R_1\}^k$
let $s_{\alpha^*}$ be the output over $N$ iterations of the program with
$f_j = \alpha^*_j$ and every other field as written. If $p$ is a prime with
$p > 2^{kN}$ and $p > 3^F 2^k N$, there is an $x_0$ such that for every
$\alpha^*$
\begin{enumerate}[label=(\roman*)]
\item with baseline $f_1 = \dots = f_k = r_{\mathrm{acc}}$, the set $S$ is
$k$-coupled for the match score against $s_{\alpha^*}$, with $\vbase = 0$ and
$\alpha^*$ the only witness;
\item for $\alpha^*$ uniform, an algorithm that reads only the match score needs,
in expectation, at least $(2^k+1)/2$ queries to reach $\alpha^*$ if it queries
assignments of $S$, and at least $2^k/(2N)$ queries before the first query that scores above
$0$ if it queries assignments of all $F$ register fields;
\item $\alpha^*$ is computed from $s_{\alpha^*}$ with $O(k)$ arithmetic
operations.
\end{enumerate}
\end{proposition}
\begin{proof}
All operations are affine and the loads are constants, so the output of every
assignment $g$ of the register fields at position $i$ is $A_g(i) + B_g(i)\, x_0$
for some $A_g(i), B_g(i) \in \mathbb{Z}/p$. A query output and a target output whose
expressions differ agree for at most one $x_0$. There are at most $3^F 2^k N$
triples of an assignment, a target and a position, so a uniformly random $x_0$
fails with probability at most $3^F 2^k N/p < 1$. Fix an $x_0$ at which a query
output agrees with a target output only where their expressions are identical.

Let $e_j = 1$ if $\alpha^*_j = R_1$ and $e_j = 0$ otherwise, and
$v = \sum_j 2^{k-j} e_j$, an integer in $[0, 2^k)$ that determines $\alpha^*$.
One iteration maps $r_{\mathrm{acc}}$ to $\lambda r_{\mathrm{acc}} - v$ with
$\lambda = 2^k$, so
\[
  s_{\alpha^*}[i] = \lambda^{i+1} x_0 - v\, \frac{\lambda^{i+1} - 1}{\lambda - 1},
  \qquad i = 0, \dots, N-1.
\]
Since $p > 2^{kN}$, we have $\lambda^m \ne 1$ for $1 \le m \le N$, so this expression is
injective in $v$ at every $i$, and for a given assignment and position at most
one target has the same expression.

(i) An assignment to $S$ that sets some $f_j$ to $r_{\mathrm{acc}}$ makes
$\textsc{sub}_j$ compute $r_{\mathrm{acc}} - r_{\mathrm{acc}} = 0$, after which the
output is a constant; its expression has no $x_0$ term while the target's has
$\lambda^{i+1} \ne 0$, so it scores $0$. The baseline is such an assignment, so
$\vbase = 0$, and condition~(a) of \Cref{def:coupling-set} holds. An assignment in
$\{R_0, R_1\}^k$ other than $\alpha^*$ has a different $v$ and scores $0$, and
$\alpha^*$ scores $N$.

(ii) On assignments of $S$ every query other than $\alpha^*$ scores $0$, so a
deterministic algorithm visits a fixed sequence until it reaches $\alpha^*$, which
takes $(2^k+1)/2$ queries on average. On assignments of all fields, each query
scores above $0$ for at most $N$ of the $2^k$ targets, one per position. The
queries made before the first positive score form a fixed sequence, and after $q$
of them at most $qN$ targets have scored above $0$. The number $T$ of queries
up to the first positive score therefore satisfies $\Pr[T > q] \ge 1 - qN/2^k$, and
$\mathbb{E}[T] = \sum_{q \ge 0} \Pr[T > q] \ge 2^k/(2N)$. Since $\alpha^*$ scores
$N$, reaching it takes at least as long. A randomized algorithm is an average of
deterministic ones.

(iii) The value $v$ is the representative in $[0, p)$ of $\lambda x_0 - s_{\alpha^*}[0]$,
and its binary digits are the $e_j$.
\end{proof}

A second class, with the same registers and operations, has coupling width equal to its block length.

\begin{proposition}[Block Horner programs]\label{prop:block}
Let $m, \kappa \ge 1$ be integers and $k = m\kappa$. Take the registers, operations and arithmetic of
\Cref{prop:horner}, a prime $p > 2^{k+1}$, and the loop body
\[
  \textsc{load}(r_{\mathrm{acc}}, 1)\
  \bigl(\textsc{dbl}\ \textsc{sub}\bigr)^{\kappa}\ \textsc{out}\
  \cdots\
  \bigl(\textsc{dbl}\ \textsc{sub}\bigr)^{\kappa}\ \textsc{out},
\]
with $m$ blocks of $\kappa$ steps, the $j$th $\textsc{sub}$ having source field
$f_j$, and one output per $\textsc{out}$. Let $S = \{f_1, \dots, f_k\}$ with baseline
$f_j = r_{\mathrm{acc}}$ and the other fields fixed as written, and for
$\alpha^* \in \{R_0, R_1\}^k$ let the target be the output of the program with
$f = \alpha^*$, the $mN$ symbols it emits in $N$ iterations. Then
\begin{enumerate}[label=(\roman*)]
\item for every $\alpha^*$, the score restricted to $S$ has $\vbase = 0$ and
coupling width $\kappa$, its only coupled set being the first block;
\item over a uniform $\alpha^*$, the least expected number of queries on $S$ with
which an algorithm that reads only the score reaches $\alpha^*$ is
$1 + m(2^\kappa - 1)/2$;
\item $\alpha^*$ is computed from the target with $O(k)$ arithmetic operations.
\end{enumerate}
\end{proposition}
\begin{proof}
Let $e_j = 1$ if $f_j = R_1$ and $e_j = 0$ if $f_j = R_0$. If none of
$f_1, \dots, f_{b\kappa}$ is $r_{\mathrm{acc}}$, the $b$th output of an iteration is the
integer $2^{b\kappa} - V_b$ with $V_b = \sum_{j \le b\kappa} 2^{b\kappa - j} e_j$,
which lies in $[1, 2^{b\kappa}]$; otherwise a $\textsc{sub}$ with source
$r_{\mathrm{acc}}$ has reset the accumulator to $0$ and the output is an integer in
$(-2^{b\kappa}, 0]$. Since $p > 2^{k+1}$, two of these integers are equal modulo
$p$ only if they are equal, so the $b$th output matches the target exactly when the
first $b$ blocks of the assignment equal those of $\alpha^*$, where $V_b$ is the
binary number they spell. Every iteration emits the same values, so an assignment
whose first $L$ blocks are correct scores $NL$, and $\alpha^*$ scores $mN$. The
baseline resets at the first $\textsc{sub}$, so its outputs are at most $0$ while
the target's are at least $1$, and $\vbase = 0$.

(i) The first block is coupled. An assignment to a proper subset of it leaves
some field at $r_{\mathrm{acc}}$ and scores $0 = \vbase$, while $\alpha^*$ on it
scores $N$. A coupled set $T$ needs an assignment on $T$, the rest at baseline,
with $L \ge 1$, so $T$ contains the first block. If $T$ were larger, $\alpha^*$ on
the first block alone would be an assignment to a proper subset of $T$ scoring
$N$, against condition~(a).

(ii) Fix the blocks of $\alpha^*$ other than the $b$th. Until a query is correct
on blocks $1, \dots, b$, no response depends on $\alpha^*_b$, so the distinct
values of block $b$ carried by queries correct on blocks $1, \dots, b-1$ form a
fixed sequence. It contains $\alpha^*_b$, since the last query is $\alpha^*$. If
$\alpha^*_b$ is its $r$th term, at least $r - 1$ queries fail on block $b$, and as
$\alpha^*_b$ is uniform and independent of the other blocks this is at least
$(2^\kappa - 1)/2$ on average. Every query other than $\alpha^*$ fails on exactly
one block, so summing over blocks and adding the last query gives at least
$1 + m(2^\kappa - 1)/2$. Consider the algorithm that keeps the blocks it has confirmed and
sets every other block to a value not yet ruled out for it, in a fixed order.
Each of its failing queries rules out a new value, and a block
with one value left never fails, so this algorithm attains the bound. A randomized algorithm is an average of
deterministic ones.

(iii) The $m$th output of an iteration is $2^k - V_m$, and $V_m$ spells
$\alpha^*$ in binary.
\end{proof}

\begin{remark}[The cost follows the coupling width]\label{rem:block}
In \Cref{prop:block} the searched fields, the registers and the operations are the
same for every $\kappa$; the blocks are marked by $m = k/\kappa$ outputs per
iteration. As $\kappa$ grows from $1$ to $k$ the optimal expected cost grows
strictly from $1 + k/2$ to $(2^k + 1)/2$, and it lies between $(2^\kappa + 1)/2$
and $m(2^\kappa + 1)/2$. The first block has the single witness $\alpha^*_1$, so
the order $(d-1)^\kappa/\msig = 2^\kappa$ of \Cref{cor:yao}(b) is attained, up to
the factor $m$, at every coupling width. On $\{R_0, R_1\}^k$, after
relabeling so that $\alpha^*$ is the all-correct string, the score is $N$ times
the number of leading correct blocks. This score is the royal staircase fitness of van
Nimwegen and Crutchfield (2000) with block size $\kappa$, less one and scaled by $N$. The proposition adds that the staircase arises
from a fixed register program in which $\kappa$ is the coupling width.
\end{remark}

\subsection{Search Lower Bound and the Connection to $K(s)$}

The coupling of \Cref{thm:k-growth} gives a lower bound on the number of score queries that type discovery needs.

\begin{theorem}[Match-score-oracle lower bound]\label{thm:search}
Let $P^*$ be accumulator-driven, minimal, and structurally coupled against $s$,
with $\varepsilon + \mu(s) < 1$, residual $\gamma(N) = o(N)$ and $\msig$
witnesses. Then match-score-oracle search for the type-discovery CSP on
$(F^*, s)$ requires
\[
  \mathbb{E}[\#\text{queries}] \ge \Omega\!\left(\frac{(T-1)^\Knt}{\msig}\right)
  = \Omega\!\left(c_T^{\,\Knt \log_2 T - \log_{c_T} \msig}\right), \qquad
  c_T = (T-1)^{1/\log_2 T},
\]
against the worst-case score in $\mathcal{A}_{\delta_0, \Delta_0}$.
\end{theorem}

\begin{proof}
\Cref{thm:k-growth} gives $(\delta_0, \Delta_0)$-$\Knt$-coupling with margin
$\Omega(N)$; clause~(III$'$) and \Cref{cor:yao} then yield
$\Omega((T-1)^\Knt / \msig)$. The second form is the first rewritten in base
$c_T = (T-1)^{1/\log_2 T}$, since $(T-1)^\Knt = c_T^{\,\Knt \log_2 T}$, the
multiplicity entering the exponent as $\log_{c_T} \msig$, which is
$O(\log \Knt)$ for $\msig$ polynomial in $\Knt$ and zero in the unique-witness
regime. When $K(F^* \mid \Knt, N) = O(1)$ and $K(\rho^* \mid F^*, \tau^*) = O(\log N)$,
the bound \eqref{eq:kol-gen} gives $\Knt \log_2 T \ge K(s) - O(\log N)$, so the exponent
is at least $K(s) - O(\log N) - \log_{c_T} \msig$.
\end{proof}

\begin{remark}[Relation to Levin's $2^{K(s)}$]\label{rem:levin-base}
For fixed $T$ the base satisfies $c_T < 2$ strictly (e.g., $c_3 \approx 1.549$, $c_{12} \approx
1.952$), so the lower bound $\Omega(c_T^{K(s)})$ falls below Levin's upper bound $2^{K(s)}$ by a
factor compounding to $\exp(-\Knt/T)$ in $\Knt$. As $T \to \infty$, the base $c_T$ tends to $2$ and
the bound approaches Levin's. The gap reflects the catalog's
\textsc{identity}-degeneracy. Levin's bound depends on the instruction set only through a constant factor, while this lower bound depends on it through the base $c_T$.
\end{remark}

\begin{remark}[Connection to algorithmic probability]\label{rem:Ms}
Solomonoff's algorithmic probability is $M(x) = \sum_{p : U(p) = x} 2^{-|p|}$.
Each of the $\msig$ witnesses in $Q$ is a typing of the $\Knt$ slots that
generates $s$ together with the frame $F^*$ and the register assignment $\rho^*$,
and is described by $\Knt \log_2 T$ bits given $(F^*, \rho^*)$. When
$K(F^*, \rho^* \mid \Knt, N) = O(\log N)$,
\[
  M(s) \ge \msig \cdot 2^{-(\Knt \log_2 T + O(\log N))},
\]
and since $K(s) \le -\log_2 M(s) + O(\log N)$, this gives
$\log_2 \msig \le \Knt \log_2 T - K(s) + O(\log N)$. The bound
$(T-1)^\Knt / \msig$ of \Cref{thm:search}, read in $K(s)$, is therefore
$\Omega\bigl(2^{\,K(s) - \Knt \log_2\frac{T}{T-1} - O(\log N)}\bigr)$,
an ISA-relative Levin bound whose base correction, $\log_2\frac{T}{T-1}$ bits per slot and at most
$\Knt/((T-1)\ln 2)$ bits in total, does not depend on the number of witnesses and
vanishes as the catalog grows. For an optimally coded typing, where
$\Knt \log_2 T = K(s) + O(\log N)$, the bound reads $\Omega(c_T^{K(s) - O(\log N)})$.
\end{remark}

\subsection{Match-Count Coupling Width}

The rigid bound $\kappa_{\mathrm{bin}} \ge \Knt$ (\Cref{thm:binary-coupling}) is stated
for the binary score $\score_{\mathrm{bin}}$. The match-count
score behaves differently in the rigid regime, and the two results below
delimit it. We use two conventions. The \emph{max-$\rho$} score
$\score_{\mathrm{match}}(\tau) = \max_\rho \mathrm{mc}(F^*, \tau, \rho, N)$ is the
adversarial convention, with the strongest partial-perturbation scores. The
\emph{fixed-$\rho^*$} score $\score_{\mathrm{match}}^{\rho^*}(\tau) =
\mathrm{mc}(F^*, \tau, \rho^*, N)$ holds the registers at the generating assignment $\rho^*$, matching the CSP analyzed in
\Cref{thm:k-growth}. Since the max-$\rho$ score is at least the fixed-$\rho^*$ score, condition~(a$'$) is
harder to satisfy under max-$\rho$, so $\kappa_{\mathrm{match}}^{\max_\rho} \le
\kappa_{\mathrm{match}}^{\rho^*}$. On the accumulator-driven class the unique
admissible assignment $\rho \equiv r_{\mathrm{acc}}$ makes the two coincide.

\begin{theorem}[Rigid match-count coupling width can be below $\Knt$]\label{thm:rigid-match}
Let $\score_{\mathrm{match}}(\tau) = \max_\rho |\{\, t : \exptr(F^*, \tau, \rho,
N)[t] = s[t] \,\}|$ with all-\textsc{identity} baseline $\alpha_0$, for a
nontrivial target $s$ (one the baseline does not generate, so $\vbase < N$),
and let $\kappa_{\mathrm{match}}^{(0,1)}$ be the $(0,1)$-relaxed match-count
coupling width. If $F^*$ and the catalog satisfy the composition-permutation
invariance of \Cref{prop:cpi}, then
\[
  \kappa_{\mathrm{match}}^{(0,1)} \le \Knt \le \kappa_{\mathrm{bin}},
\]
and the left inequality can be strict, with $\kappa_{\mathrm{match}}^{(0,1)} < \Knt$
whenever some proper subset of slots, suitably typed, achieves a match count
strictly above the all-\textsc{identity} baseline.
\end{theorem}
\begin{proof}
The inequality $\Knt \le \kappa_{\mathrm{bin}}$ is \Cref{thm:binary-coupling}. For the left
inequality, let $S$ be a slot set with $|S| > \Knt$. Under
composition-permutation invariance the generating assignment on the minimal
$\Knt$-subset transports, by a permutation of the slot index set, to a
generating assignment on every $\Knt$-subset (as in \Cref{prop:cpi}). In
particular, some proper subset $S' \subsetneq S$ with $|S'| = \Knt$ carries
types $\beta$ with $\exptr(F^*, \beta \cup \textsc{identity}_{\mathrm{rest}},
\rho, N) = s$ for some $\rho$, so
$\score_{\mathrm{match}}(\restrS{\alpha_0}{S'}{\beta}) = N > \vbase$ and
condition~(a$'$) of \Cref{def:coupling-set-relaxed} fails at $\delta = 0$.
Hence no set of size exceeding $\Knt$ is $(0,1)$-coupled and
$\kappa_{\mathrm{match}}^{(0,1)} \le \Knt$. For the strict case, suppose some
proper subset $S' \subsetneq S$ of a $\Knt$-set $S$, with types $\beta$, has
$\score_{\mathrm{match}}(\restrS{\alpha_0}{S'}{\beta}) > \vbase$. Then
condition~(a$'$) fails at $S'$ for $\delta = 0$ and $S$ is not
$(0,1)$-coupled; since $|S'| < \Knt$, the same transport places a violating
proper perturbation inside every $\Knt$-subset, so none is $(0,1)$-coupled and
$\kappa_{\mathrm{match}}^{(0,1)} < \Knt$. All provable instances
(\Cref{prop:zp}) satisfy the invariance.
\end{proof}

\begin{theorem}[Relaxed match-count coupling width for structurally coupled programs]\label{thm:relaxed-match}
Let $P^* = (F^*, \tau^*, \rho^*)$ generate $s$ and be structurally coupled
against $s$ (\Cref{def:struct-coupled}) with parameters $(\varepsilon, \gamma)$,
where $\varepsilon + \mu(s) < 1$, and let $\vbase$ be the all-\textsc{identity} match count of
\Cref{thm:k-growth}.
Under the fixed-$\rho^*$ match score $\score_{\mathrm{match}}^{\rho^*}(\tau) =
\mathrm{mc}(F^*, \tau, \rho^*, N)$, the type-discovery CSP has
\[
  \kappa_{\mathrm{match}}^{\rho^*, (\delta, \Delta)} \ge \Knt,
  \qquad \delta = (\varepsilon + \mu(s)) N + \gamma(N) - \vbase,\quad
  \Delta = N - \vbase,
\]
with correction margin $\Delta - \delta = (1 - \varepsilon - \mu(s)) N -
\gamma(N) = \Omega(N)$.
\end{theorem}
\begin{proof}
The theorem is \Cref{thm:k-growth} read for the fixed-$\rho^*$ score; condition~(a$'$)
is the Wire Corruption Lemma (\Cref{lem:wire-corruption}) bound, condition~(b$'$)
is $\mathrm{mc}(F^*, \tau^*, \rho^*, N) = N = \vbase + \Delta$, with coupled set
$S = \{1, \dots, \Knt\}$ and witness $\tau^*$.
\end{proof}

\section{Register Inference Algorithms}

We now specify the register-inference algorithms themselves, the
verification step and the two enumeration approaches, whose costs are the
quantities the coupling analysis prices. In this section the target is the output of the ground truth, $s = \exptr(F^*, \tau^*, \rho^*, N)$.

\subsection{Verification}
To verify a candidate assignment, initialize $W$ registers from preloop \textsc{load}s in $O(L)$; execute the
abstract syntax tree (AST) for $N$ iterations of the outer loop, each traversing the frame with $O(1)$ work per
node; compare against $s$ in $O(N)$. We assume that the frame has $O(L)$ nodes, which
holds for the frames OMNIS enumerates, since there every loop is itself an
operation slot (the \textsc{loop} entries of \Cref{tab:isa}). The total cost is $O(NL)$.

\subsection{Pattern Enumeration (Approach A)}

At iteration $t = 0$ the register state $r(0)$ is determined by the preloop
\textsc{load}s, computable from $(F^*, \tau^*)$ alone. Write $|\Fields|$ for the number of register fields of $(F^*, \tau^*)$, and maintain the set of
surviving complete patterns $\mathcal{S} \subseteq R^{|\Fields|}$, initialized
to $\mathcal{S}_0 = R^{|\Fields|}$; at each iteration eliminate patterns whose
simulated output disagrees with $s[t]$, so that
$\mathcal{S}_t = \{\, \hat\rho \in \mathcal{S}_{t-1} :
\exptr(F^*, \tau^*, \hat\rho, N)[t] = s[t] \,\}$.
The \emph{distinguishing depth} $d^* \le N-1$ is the least $d$ such that every two
register assignments whose outputs differ somewhere in $[N]$ already differ at
some position $t \le d$.

\begin{lemma}[Safety of Approach A]\label{lem:A-safe}
The ground truth satisfies $\rho^* \in \mathcal{S}_t$ for all $t$.
\end{lemma}
\begin{proof}
By hypothesis, $\exptr(F^*, \tau^*, \rho^*, N)[t] = s[t]$ for all $t$, so
$\rho^*$ is never eliminated.
\end{proof}

\begin{lemma}[Completeness of Approach A]\label{lem:A-complete}
After processing iterations $0$ through $d^*$, every surviving
$\hat\rho \in \mathcal{S}_{d^*}$ satisfies $\exptr(F^*, \tau^*, \hat\rho, N) = s$.
\end{lemma}
\begin{proof}
Every $\hat\rho$ with $\exptr(F^*, \tau^*, \hat\rho, N) \ne s$ has a
distinguishing position $t \le d^*$, by the definition of the distinguishing depth, and
Approach~A eliminates it at that position.
\end{proof}

Approach~A starts from $|\mathcal{S}_0| = W^{|\Fields|}$ patterns, and each iteration
costs $O(|\mathcal{S}_t| L)$, so its total cost is $O(W^{|\Fields|} NL)$, unconditional but
exponential in $|\Fields|$.

\subsection{Independent Field Testing (Approach B)}

OMNIS does not use Approach~B. We analyze it to show that polynomial inference is structurally achievable despite unbounded
treewidth, and to expose how coupling manifests on the inference side. For each
field $f$ maintain a candidate set $\mathcal{C}(f) \subseteq R$, initialized to $R$. At
each iteration a representative $\rho_{\mathrm{rep}}$ (chosen from the current
candidate sets) gives an approximate state by forward simulation; each candidate
$r \in \mathcal{C}(f)$ is tested with $r$ at $f$ and $\rho_{\mathrm{rep}}$ elsewhere, and
eliminated if the simulated output disagrees with $s[t]$.

\begin{definition}[Representative accuracy]\label{def:rep-accurate}
The representative $\rho_{\mathrm{rep}}$ is \emph{accurate} if $\rho_{\mathrm{rep}}(f) = \rho^*(f)$
for every field $f$.
\end{definition}

\begin{lemma}[Conditional soundness of Approach B]\label{lem:B-sound}
If $\rho_{\mathrm{rep}}$ is accurate (\Cref{def:rep-accurate}), then $\rho^*(f)$
is never eliminated from $\mathcal{C}_t(f)$.
\end{lemma}
\begin{proof}
If $\rho_{\mathrm{rep}} = \rho^*$ pointwise, the simulation with $\rho^*(f)$ at
$f$ and $\rho_{\mathrm{rep}}$ elsewhere equals $\exptr(P^*, N)[t]$, which is
$s[t]$; so $\rho^*(f)$ survives.
\end{proof}

\begin{lemma}[Monotonicity]\label{lem:B-mono}
The candidate sets satisfy $\mathcal{C}_{t+1}(f) \subseteq \mathcal{C}_t(f)$ for all $f, t$.
\end{lemma}
\begin{proof}
Each iteration only removes candidates that fail a test; none are added, so
$\mathcal{C}_{t+1}(f)\subseteq \mathcal{C}_t(f)$.
\end{proof}

\begin{lemma}[Termination]\label{lem:B-term}
The algorithm terminates after one pass through $t = 0, \dots, N-1$, with
$O(|\Fields|\, W N)$ elimination tests.
\end{lemma}
\begin{proof}
The algorithm performs one pass over $t=0,\dots,N-1$ and tests each of the
at most $W$ candidates at each of the $|\Fields|$ fields once per iteration, so it makes
$O(|\Fields|\,WN)$ tests and terminates after the pass.
\end{proof}

\begin{theorem}[Convergence of Approach B under representative accuracy]\label{thm:B-converge}
If $\rho_{\mathrm{rep}}$ is accurate at every iteration, Approach~B eliminates
every candidate $r$ at every field $f$ with
$\exptr(P^*[\{f\}\!\to\!r], N) \ne s$, within $d^* + 1$ iterations, at total
cost $O(|\Fields|\, W (d^*+1) L) = O(L^2 W N)$; the ground truth $\rho^*$
survives at every field.
\end{theorem}
\begin{proof}
By hypothesis, $\rho_{\mathrm{rep}} = \rho^*$ pointwise, so
$\rho^*$ survives by \Cref{lem:B-sound}. For a candidate $r$ at $f$ with
$\exptr(P^*[\{f\}\!\to\!r], N) \ne s$, a distinguishing position $t \le d^*$
exists by the definition of $d^*$. At iteration $t$ the simulation with $r$ at $f$ and $\rho^* = \rho_{\mathrm{rep}}$
elsewhere is exactly $\exptr(P^*[\{f\}\!\to\!r], N)$, which disagrees with
$s[t]$, so Approach~B eliminates $r$.
\end{proof}

\begin{remark}[The convergence gap]\label{rem:B-gap}
The convergence hypothesis (pointwise equality with $\rho^*$) is not available at the start, because it is the
conclusion. Approach~B may eliminate
$\rho^*(f)$ at an intermediate iteration when the representative is far from
$\rho^*$, because end-to-end simulation with $\rho^*(f)$ at $f$ and a wrong
representative elsewhere can disagree with $s[t]$. An unconditional polynomial
guarantee for Approach~B is open. This gap motivates WSBP, which trades
conditional soundness for unconditional soundness at the cost of conditional
completeness.
\end{remark}

\section{Wire-Space Backward Propagation}

Given $(F^*, \tau^*)$ and target $s$, WSBP searches for a register assignment by
constructing an abstract wire topology through demand-driven backward
propagation in the AST, then mapping wires to physical registers and verifying
by forward simulation. The types are drawn from a restricted catalog in which
every \textsc{mod} and \textsc{divmod} uses a constant divisor; consequently
every register field of every catalog type appears in the operation's port
signature. Throughout, the frame $(F^*, \tau^*)$ fixes the program's structural and semantic skeleton,
namely the loop-body operation sequence, which operation's output the \textsc{out} emits,
whether that operation is loop-carried, and the model of execution and
output interpretation under which the body is run. The engine's search over
this skeleton is an enumeration over frames, so the per-frame statements below
range, for a fixed $(F^*, \tau^*)$, only over the routing, demand-resolution,
operand-ordering, realization, and preloop initializer choices.

\begin{definition}[Wire]\label{def:wire}
A \emph{wire} is an element of a countably infinite set $\mathcal{W}$, allocated
sequentially (each allocation returns a previously unused wire).
\end{definition}

\begin{definition}[Wire topology]\label{def:topology}
A \emph{wire topology} for a loop body with operations $l_1, \dots, l_J$ is a
pair $\mathcal{T} = (\eta, U)$, where
$\eta : \{1, \dots, J\} \times \{\mathrm{in}_0, \mathrm{in}_1, \mathrm{out}_0,
\mathrm{out}_1\} \to \mathcal{W} \cup \{\bot\}$ assigns each port a wire or
$\bot$ (unused) and $U$ is an equivalence relation on allocated wires
(maintained by union-find). Ports in the same class must refer to the same physical register in every valid
realization.
\end{definition}

\begin{definition}[Register realization]\label{def:realization}
A \emph{register realization} of $\mathcal{T} = (\eta, U)$ is a map
$\rho : \Fields(F^*, \tau^*) \to R$ such that every two ports whose wires lie in $\mathcal{W}$ and are $U$-equivalent have their
fields mapped to the same register.
\end{definition}

\begin{definition}[Functional equivalence of register assignments at length $N$]\label{def:func-equiv-rho}
For a fixed frame $(F^*, \tau^*)$, assignments $\rho$ and $\rho'$ satisfy $\rho \sim \rho'$ if $\exptr(F^*, \tau^*, \rho, N) =
\exptr(F^*, \tau^*, \rho', N)$; write $[\rho]$ for the \emph{functional class} of $\rho$, its class under $\sim$.
\end{definition}

\begin{definition}[Structural minimality]\label{def:struct-min}
A program $(F^*, \tau^*, \rho^*)$ is \emph{structurally minimal} if no register field is
dead under $\rho^*$, that is, every written value is read by some subsequent visible
operation before being overwritten, and every output contributes (transitively)
to the \textsc{out} sequence.
\end{definition}

\begin{lemma}[Functional invariance under register relabeling]\label{lem:class-bound}
For every permutation $\pi$ of $R$, the assignment $\rho_\pi = \pi \circ \rho^*$
satisfies $\exptr(F^*, \tau^*, \rho_\pi, N) = \exptr(F^*, \tau^*, \rho^*, N)$.
Hence register relabeling maps each functional class to itself.
\end{lemma}
\begin{proof}
We induct on execution steps and write $\mathrm{val}_k^\pi$ and $\mathrm{val}_k^*$ for the register state after
step $k$ under $\rho_\pi$ and under $\rho^*$. The initial state (set by the program with its constants held fixed)
satisfies $\mathrm{val}_0^\pi = \mathrm{val}_0^* \circ \pi^{-1}$. If
$\mathrm{val}_k^\pi = \mathrm{val}_k^* \circ \pi^{-1}$, then each operation reads
$\mathrm{val}_k^\pi(\pi(\rho^*(\mathrm{src}))) = \mathrm{val}_k^*(\rho^*(\mathrm{src}))$,
the value it reads under $\rho^*$, writes it (relabeled by $\pi$), and
\textsc{out} reads identical values. The output sequences therefore agree.
\end{proof}

\subsection{The Algorithm}

WSBP takes as input a frame $F^*$, a type assignment $\tau^*$ from the restricted catalog,
and a target $s \in \Sigma^N$. It returns the candidate with the highest match count against $s$, together with that count. Let $l_1, \dots, l_J$ be the
loop-body leaves in execution order and $b$ the number of $\textsc{divmod}_C$
operations. The algorithm proceeds in six steps:
\begin{enumerate}[label=Step~\arabic*.,leftmargin=*,wide]
\item (Routing) Each $\textsc{divmod}_C$ has two outputs; a routing
$\theta$ selects, per $\textsc{divmod}_C$, which output receives the downstream
demand wire. Steps~2--6 run independently for each of the $2^b$ routings, and WSBP keeps
the candidate with the highest match count.
\item (Backward trace) Initialize $\eta \equiv \bot$. From the
\textsc{out} input, allocate a fresh wire and push a demand. While demands
remain, pop a demand $(w, j)$, a request for the producer of the wire $w$ that operation $l_j$ reads,
select the nearest unassigned non-\textsc{out},
non-\textsc{load} producer, and assign wires to the producer's ports per
its type signature. The search for the producer goes backward, then wraps for loop-carried
values, and otherwise defers the demand to Step~3. Read-modify-write ports share $w$, binary operations
allocate fresh input wires and push their demands, and $\textsc{divmod}_C$ routes
$w$ to the selected output.
\item (Structural resolution) Resolve deferred demands (merging with the
accumulator wire when one exists), encode the inner-loop $\textsc{divmod}$
feedback $\eta(i,\mathrm{in}_0) \equiv_U \eta(i,\mathrm{out}_0)$, detect the
accumulator class, wire in-loop \textsc{load}s, and wire control conditions.
At most $W$ deferred demands are kept, and those beyond the $W$th remain classes of
their own. A demand-resolution configuration $u$ merges each kept demand with one
of the wire classes the trace produced or leaves it alone, and a configuration
leaving more than $W$ classes cannot be realized and is discarded.

Every register field is assigned a wire by the end of this step. The backward
trace wires each producer port it visits, and structural resolution wires the
deferred demands, the in-loop \textsc{load}s, and the control conditions. A
field that no port ever reads or writes is dead under $\rho^*$, hence absent by
structural minimality (\Cref{def:struct-min}). Step~4 therefore enumerates
register assignments only over the wire classes the trace produced and never over
fresh choices for unwired fields. The image bound of the next section
(\Cref{lem:image-size}) counts every realization WSBP can produce under this convention.
\item (Realization) Map each wire equivalence class to a physical register,
enumerating the injective assignments of the $n_c$ classes to the $W$ registers
($(W)_{n_c}$ of them).
By \Cref{lem:class-bound} these assignments are register relabelings within one
functional class, so they exhaust that class's behavior without enlarging the set
of functions WSBP can realize.
\item (Preloop initialization and scoring) Pair each realized assignment
with a preloop initialization, drawn from a finite set $\mathcal{I}$ of
initializing prefixes (instruction sequences of bounded length that fix the
initial register state), and score the resulting program against $s$ by prefix
match. The target $s$ enters at this step. The wire topology of Steps~2--3 and the
realization of Step~4 are determined by $(F^*, \tau^*)$ alone and never read
$s$; the target is consulted only to score the candidates already realized.
\item (Verification) Simulate the realized assignment forward in $O(NL)$,
report its match count against $s$, and keep the best candidate across the
$2^b$ routings.
\end{enumerate}

We now prove that WSBP is sound and bound its running time.

\begin{theorem}[Always-correct]\label{thm:correct}
WSBP never reports a false perfect match; if it returns a candidate
$\rho_{\mathrm{cand}}$ with match count $N$, then
$\exptr(F^*, \tau^*, \rho_{\mathrm{cand}}, N) = s$.
\end{theorem}
\begin{proof}
Step~6 computes $\exptr(F^*, \tau^*, \rho_{\mathrm{cand}}, N)$ by forward
simulation and counts matching symbols; forward simulation is deterministic, so
a match count of $N$ means every symbol matches. The guarantee is independent of
the backward-trace procedure, the routing, the structural merges, and the
\textsc{load} wiring. Verification is therefore unconditional, and the only way WSBP can fail is by missing the correct wiring (a partial or no match).
\end{proof}

\begin{theorem}[Polynomial in the target length]\label{thm:poly}
For a fixed frame $(F^*, \tau^*)$, WSBP performs $N_{\mathrm{cand}}$ candidate evaluations, a number
that depends only on the frame and is independent of the target length $N$; each is a
forward simulation of cost $O(NL)$. Hence WSBP runs in
$O\bigl(2^b J^2 + N_{\mathrm{cand}} \cdot NL\bigr)$ time, where $b$ is the number of
$\textsc{divmod}_C$ operations, $J$ the number of loop-body operations, $N$ the
target length, and $L$ the number of operation slots, the frame having $O(L)$
nodes. The count $N_{\mathrm{cand}}$ is the product of the numbers of
routings ($2^b$), demand-resolution configurations, register
assignments ($(W)_{n_c}$), operand orderings of the noncommutative operations, and
preloop initializers, so that
$N_{\mathrm{cand}} \le 2^{b+n_{\mathrm{nc}}} c$ with $c = N_U\, W!\, N_P$. Here $n_{\mathrm{nc}}$ is the number of noncommutative binary operations, and $N_U$
and $N_P$ count the demand-resolution configurations and the preloop initializers. For bounded frame size
($b, J, W = O(1)$), $N_{\mathrm{cand}} = O(1)$ and WSBP runs in $O(NL)$.
\end{theorem}
\begin{proof}
The backward trace issues at most $2J+1$ demands, each scanning at most $J$
operations, so $O(J^2)$ per routing and $O(2^b J^2)$ over the $2^b$ routings;
structural resolution touches $O(J)$ ports with union-find in nearly constant time.
The realization enumerates register assignments and preloop initializers and,
together with the routings, the demand-resolution configurations, and the operand
orderings of the noncommutative operations,
produces $N_{\mathrm{cand}}$ realized candidates. Each is scored, and a full match confirmed, by
a forward simulation of cost $O(NL)$.

The count $N_{\mathrm{cand}}$ is fixed by the frame. The
routings number $2^b$. At most $W$ demands are deferred and each merge removes at most
one class, so a configuration realized on $W$ registers starts from at
most $2W$ classes, each deferred demand has at most $2W+1$ choices, and at most
$(2W{+}1)^{W}$ demand-resolution configurations reach realization. The register assignments number
$(W)_{n_c} \le W!$, and the operand orderings at most $2^{n_{\mathrm{nc}}} \le 2^{J}$. The preloop
initializers number at most $1 + m + m^2$, because the prefixes have length at most two. Here $m$, the number
of single instructions over $W$ registers with constants from the fixed catalog
ranges, is polynomial in $W$. All are independent of $N$, so the total is
$O(2^b J^2 + N_{\mathrm{cand}} \cdot NL)$. For $b, J, W = O(1)$ with $N \ge J$, the count $N_{\mathrm{cand}}$ and the routing count $2^b$ are
constants and $J^2 \le NL$, so the running time is $O(NL)$.
\end{proof}

\begin{corollary}\label{cor:poly-simple}
For $b, J, W = O(1)$ and $N \ge J$, WSBP runs in $O(NL)$.
\end{corollary}
\begin{proof}
By \Cref{thm:poly}, bounded frame size makes the candidate count $N_{\mathrm{cand}}$ and the
routing count $2^b$ constant, and $N \ge J$ gives $J^2 \le NL$, so WSBP runs in $O(NL)$.
\end{proof}

\section{WSBP Incompleteness}

WSBP is sound and runs in polynomial time. We now show what that
speed costs. Frames with enough wiring freedom carry programs whose register assignment WSBP cannot recover
(\Cref{thm:incomplete}), and on such frames these programs are
generic (\Cref{thm:typical-incomplete} and \Cref{cor:quant-incomplete}).

By \Cref{lem:image-size} below, WSBP realizes at most $2^{b+n_{\mathrm{nc}}} N_U N_P$ functional classes, where $n_{\mathrm{nc}}$ is the number of noncommutative binary
operations, and $N_U$ and $N_P$, the numbers of demand-resolution configurations and of
preloop initializers, are functions of the register
count $W$. The target space $R^{\Fields(F^*, \tau^*)}$, by contrast, has size $W^{|\Fields|}$.

\begin{definition}[WSBP image]\label{def:wsbp-image}
Fix $(F^*, \tau^*)$ from the restricted catalog, with $b$ the number of
$\textsc{divmod}_C$ operations, $n_{\mathrm{nc}}$ the number of noncommutative binary
operations, $N_U$ the number of demand-resolution configurations, and $N_P$ the
number of preloop initializers. Let $\mathcal{U}$ be the set of demand resolutions,
$\mathcal{O}$ the set of the $2^{n_{\mathrm{nc}}}$ operand orderings of the noncommutative operations,
and $\mathcal{I}$ the set of preloop initializers. Taking realizations up to
register relabeling (which preserves the computed function by
\Cref{lem:class-bound}), define
\[
  \mathcal{B}_\rho(F^*, \tau^*) := \{\, [\rho] : \rho = \mathrm{WSBP}(F^*, \tau^*, \theta, u, o, \iota),\;
  (\theta, u, o, \iota) \in \{0,1\}^b \times \mathcal{U} \times \mathcal{O} \times \mathcal{I} \,\},
\]
the set of functional classes WSBP can realize. It depends only on
$(F^*, \tau^*)$ and is the same for every target $s$ and every $\rho^*$.
\end{definition}

\begin{lemma}[Image size bound]\label{lem:image-size}
WSBP realizes at most $2^{b+n_{\mathrm{nc}}} N_U N_P$ functional classes, i.e., $|\mathcal{B}_\rho(F^*, \tau^*)|
\le 2^{b+n_{\mathrm{nc}}} N_U N_P$.
\end{lemma}
\begin{proof}
The wire topology is fixed by the routing $\theta$, the demand-resolution
configuration $u$, and the operand ordering $o$ of the noncommutative operations
(commutative operations contribute no distinct ordering). Realization then chooses
register relabelings, which by \Cref{lem:class-bound} leave the computed function
unchanged, and a preloop initializer from a set of size $N_P$. Hence the realized
class is determined by the tuple $(\theta, u, o, \iota)$, of which there are at most
$2^b N_U \, 2^{n_{\mathrm{nc}}} N_P = 2^{b+n_{\mathrm{nc}}} N_U N_P$.
\end{proof}

\begin{definition}[Frame function count]\label{def:frame-functions}
The set $\mathcal{F}(F^*, \tau^*)$ consists of the distinct output sequences the frame
realizes as the register assignment $\rho \in R^{\Fields(F^*, \tau^*)}$ and the preloop
initializer $\iota \in \mathcal{I}$ both vary, the initializer fixing the initial
register state via preloop \textsc{load}s (\Cref{def:execution});
$|\mathcal{F}(F^*, \tau^*)|$ is its cardinality. Every WSBP realization is one such
configuration, so $\mathcal{B}_\rho(F^*, \tau^*) \subseteq \mathcal{F}(F^*, \tau^*)$.
\end{definition}

\begin{definition}[Independently free ports]\label{def:free-ports}
Ports $p_1, \dots, p_j$ are \emph{independently free} (with respect to a base
assignment $\rho_0$ and alternative sources $a_1, \dots, a_j$) if the $2^j$
assignments $\{\rho_S : S \subseteq \{1, \dots, j\}\}$ produce $2^j$
pairwise distinct output sequences $\exptr(F^*, \tau^*, \rho_S, N)$, where $\rho_S$ routes
$p_i$ to $a_i$ for $i \in S$ and to its base source otherwise.
\end{definition}

\begin{lemma}[Function-count lower bound]\label{lem:function-lower}
If the frame has $j$ independently free ports, then $|\mathcal{F}(F^*, \tau^*)|
\ge 2^j$.
\end{lemma}
\begin{proof}
The $2^j$ assignments $\rho_S$ of \Cref{def:free-ports} have pairwise distinct
outputs, hence lie in $2^j$ pairwise distinct functional classes.
\end{proof}

\begin{theorem}[Incompleteness]\label{thm:incomplete}
Let $(F^*, \tau^*)$ be a frame from the restricted catalog with $j$
independently free ports satisfying $j > b + n_{\mathrm{nc}} + \log_2(N_U N_P)$. Then there
exists a program $P^* = (F^*, \tau^*, \rho^*)$ for which WSBP finds no
assignment achieving match count $N$.
\end{theorem}
\begin{proof}
By \Cref{lem:function-lower}, the frame realizes $|\mathcal{F}(F^*,\tau^*)| \ge 2^j$ functional classes,
while, by \Cref{lem:image-size} and the hypothesis on $j$, WSBP realizes at most $2^{b+n_{\mathrm{nc}}} N_U N_P < 2^j$
of them, so there is a class
$h^* \in \mathcal{F}(F^*,\tau^*) \setminus \mathcal{B}_\rho(F^*,\tau^*)$. By
\Cref{def:frame-functions}, $h^*$ is realized by some register assignment
$\rho^*$ and preloop initializer, which give a program $P^*$ whose output
$s := \exptr(F^*,\tau^*,\rho^*,N)$ lies in class $h^*$. Every candidate WSBP
realizes lies in $\mathcal{B}_\rho$, hence in a class distinct from $h^*$, and by
\Cref{def:func-equiv-rho} distinct classes differ on the output of length $N$,
so no candidate matches $s$ at all $N$ positions.
\end{proof}

In words, when a frame has more independently free ports than the number of bits that index WSBP's enumerated choices, some program in the frame lies outside WSBP's reach. The next theorem bounds the probability that WSBP realizes a function drawn from a probability measure on $\mathcal{F}(F^*, \tau^*)$.

\begin{theorem}[Measure-robust typical incompleteness of WSBP]\label{thm:typical-incomplete}
Fix $(F^*, \tau^*)$ from the restricted catalog. For every probability
measure $\lambda$ on $\mathcal{F}(F^*, \tau^*)$, with $\mathrm{Success}(h) := \{\, h
\in \mathcal{B}_\rho(F^*, \tau^*) \,\}$ the event that WSBP realizes the function $h$,
\[
  \Pr_{h \sim \lambda}[\mathrm{Success}] = \lambda(\mathcal{B}_\rho)
  \le |\mathcal{B}_\rho(F^*, \tau^*)| \cdot \|\lambda\|_\infty
  \le 2^{b+n_{\mathrm{nc}}} N_U N_P \cdot \|\lambda\|_\infty,
\]
where $\|\lambda\|_\infty = \max_{h \in \mathcal{F}} \lambda(h)$. Equivalently, in terms of
the min-entropy $H_\infty(\lambda) = -\log_2 \|\lambda\|_\infty$,
\[
  \Pr[\mathrm{Success}] \le 2^{\,b + n_{\mathrm{nc}} + \log_2(N_U N_P) - H_\infty(\lambda)}.
\]
The factor $2^{b+n_{\mathrm{nc}}} N_U N_P$ is independent of the target length $N$, and it depends on
$|\Fields|$ only through the exponent $b + n_{\mathrm{nc}}$. The counts $N_U \le
(2W{+}1)^W$ and $N_P$ depend only on $W$ and the fixed initializer space
(\Cref{thm:poly}). The bound
$b + n_{\mathrm{nc}} \le |\Fields|/3$ holds because each $\textsc{divmod}_C$ and each noncommutative
binary operation carries three register fields (a $\textsc{divmod}_C$ one input
and two outputs, a binary operation two inputs and one output) and the restricted
catalog admits no other operation with two outputs. The bound therefore depends on the measure only through its min-entropy, and the
success probability tends to $0$ along any family with
$H_\infty(\lambda) - b - n_{\mathrm{nc}} - \log_2(N_U N_P) \to \infty$.
\end{theorem}
\begin{proof}
The event $\mathrm{Success}(h)$ holds if and only if $h \in \mathcal{B}_\rho$, and $\mathcal{B}_\rho \subseteq
\mathcal{F}(F^*, \tau^*)$ (\Cref{def:frame-functions}), so $\Pr_{h \sim \lambda}[h \in
\mathcal{B}_\rho] = \lambda(\mathcal{B}_\rho) = \sum_{h \in \mathcal{B}_\rho} \lambda(h) \le |\mathcal{B}_\rho| \cdot
\|\lambda\|_\infty$. \Cref{lem:image-size} bounds $|\mathcal{B}_\rho| \le 2^{b+n_{\mathrm{nc}}} N_U N_P$, and
$\|\lambda\|_\infty = 2^{-H_\infty(\lambda)}$ gives the min-entropy form.
\end{proof}

\begin{corollary}[Quantitative incompleteness]\label{cor:quant-incomplete}
If the frame has $j$ independently free ports and $\lambda$ assigns mass at most
$2^{-j}$ to every function (equivalently $H_\infty(\lambda) \ge j$), then
\[
  \Pr_{h \sim \lambda}[\mathrm{Success}] \le 2^{\,b + n_{\mathrm{nc}} + \log_2(N_U N_P) - j}.
\]
The uniform measure on $\mathcal{F}$ is the canonical such measure. By
\Cref{lem:function-lower}, $|\mathcal{F}| \ge 2^j$, so the uniform measure has $\|\lambda\|_\infty =
1/|\mathcal{F}| \le 2^{-j}$. The bound tends to $0$ along every family with $j - b -
n_{\mathrm{nc}} - \log_2(N_U N_P) \to \infty$. Since $b + n_{\mathrm{nc}} \le |\Fields|/3$ and $\log_2(N_U N_P)
= O(W \log W)$, it does so whenever $j$ exceeds
$|\Fields|/3$ by $\omega(W \log W)$. It also does so along every family of frames in general position, that is, frames with
$j - b - n_{\mathrm{nc}} = \Omega(|\Fields|)$ and $|\Fields| =
\omega(W \log W)$.
\end{corollary}
\begin{proof}
With $H_\infty(\lambda) \ge j$, \Cref{thm:typical-incomplete} gives $\Pr \le 2^{\,b + n_{\mathrm{nc}} +
\log_2(N_U N_P) - j}$; for the uniform measure \Cref{lem:function-lower} gives
$\|\lambda\|_\infty = 1/|\mathcal{F}| \le 2^{-j}$.
\end{proof}

\begin{remark}[The description length prior and the Solomonoff idealization]\label{rem:solomonoff}
The engine orders a frame's programs by description length, so the prior it places
on $\mathcal{F}(F^*, \tau^*)$ is $\lambda_{\mathrm{len}}(h) \propto 2^{-\mathrm{len}(h)}$, where $\mathrm{len}(h)$ is the length of the
shortest description of $h$ by its wiring and initializer. The engine computes this length, which is an upper bound on $K(h)$, and $\lambda_{\mathrm{len}}$ is the finite stand-in for the Solomonoff
prior $M$. \Cref{thm:typical-incomplete} applies verbatim and gives incompleteness
whenever $H_\infty(\lambda_{\mathrm{len}}) > b + n_{\mathrm{nc}} + \log_2(N_U N_P)$.

Two cautions separate the description length prior from
the uniform case. First, the free-port count does not transfer; $j$ independently
free ports give $|\mathcal{F}| \ge 2^j$, but $\lambda_{\mathrm{len}}$ weights those functions
unequally and concentrates on the few of shortest description, so $H_\infty(\lambda_{\mathrm{len}})$
can fall well below $j$. Second, the min-entropy condition is only sufficient. The length-weighted success probability is $\lambda_{\mathrm{len}}(\mathcal{B}_\rho)$, the
prior mass of WSBP's reachable image, which is small precisely when that image
misses the bulk of the functions of low description length. On a frame whose prior is
dominated by one trivial low-complexity wiring, the simplest function carries most
of the mass. A structural method that reaches this function is then, correctly, not incomplete, and
the bound is vacuous. Incompleteness under a length-weighted prior is therefore a
strictly stronger property, and it is still open. It asks whether the functions with short
descriptions themselves lie outside WSBP's structural reach.
\end{remark}

\begin{remark}[Operating regimes]\label{rem:regimes}
The bound $2^{b+n_{\mathrm{nc}}} N_U N_P / |\mathcal{F}|$ is small only when the frame's function
count outgrows the routing, resolution, operand-ordering, and initializer axes
WSBP already enumerates. The surplus must come from wiring freedom that its structured
backward trace cannot reach, and \Cref{lem:function-lower} measures that freedom through the
free-port count $j$. Two extremes bracket this condition. A frame in general position (\Cref{cor:quant-incomplete}) has $j =
\Theta(|\Fields|)$ independently free ports, $\Omega(|\Fields|)$ more than $b + n_{\mathrm{nc}}$, so $|\mathcal{F}| \ge
2^{j}$ while $2^{b+n_{\mathrm{nc}}} N_U N_P \le 2^{j - \Omega(|\Fields|)}$,
and $\Pr[\mathrm{Success}] \le 2^{-\Omega(|\Fields|)}$. A degenerate frame, for
instance a chain of $\textsc{inc}$ operations on a single register, has $j =
O(1)$; its distinct outputs come only from the initializer and routing axes that
$2^{b+n_{\mathrm{nc}}} N_U N_P$ already counts, WSBP enumerates them, and there is no incompleteness.
The theorem separates the two cases. Incompleteness depends on the frame as well as on the
machine, and it is a property of frames whose wiring freedom exceeds what the backward trace reaches.
\end{remark}

\begin{remark}[Compatibility with benchmark performance]\label{rem:benchmark}
OMNIS's benchmarks are curated targets generated by structurally minimal,
compressible programs with pipeline dataflow, a vanishing fraction of the
frame's function count. \Cref{thm:typical-incomplete} bounds WSBP's coverage of
$\mathcal{F}(F^*, \tau^*)$ as a whole and says nothing about the
failure rate within the curated subclass, so empirical success on the
benchmark is fully compatible with typical incompleteness.
\end{remark}

\section{The Search-Inference Coupling}

The two costs developed above, score-oracle search and structural inference,
are not independent. This section proves the capstone theorem, which compares
the two costs in score evaluations. They are also linked through the constraint
graph, whose treewidth sets the cost of exact inference and also caps the
coupling width, and with it the search exponent (\Cref{thm:coupling-tw},
\Cref{lem:gscore-primal}), so easing inference also lowers that cap.

\begin{theorem}[Search-Inference Coupling]\label{thm:coupling-capstone}
Let $P^* = (F^*, \tau^*, \rho^*)$ have coupling width $\kappa \ge 1$ with
respect to baseline $\rho_0$ and the match score.
\begin{enumerate}[label=(\alph*)]
\item The score-oracle cost on the worst-case score in the coupling family
$\mathcal{A}_{\delta, \Delta}$ is $\Omega\bigl((W-1)^\kappa / \msig\bigr)$ expected
evaluations (\Cref{cor:yao}(b) with $d = W$, which rests on clauses~(III) and~(III$'$) of the Coupling Barrier Theorem, \Cref{thm:barrier}). The
match score lies in this family, with $(\delta, \Delta) = (0, 1)$ (\Cref{prop:def-equiv}(i)).
\item The wire-space inference cost on the register-machine match-score CSP is
$O\bigl(2^b J^2 + N_{\mathrm{cand}} \cdot NL\bigr)$ (\Cref{thm:poly}), where $N_{\mathrm{cand}}$ is the
candidate count, which the frame bounds; the backward trace and structural merges are
score-free, and the target is evaluated only $O(N_{\mathrm{cand}})$ times, to select among the
realized candidates. For $b, J, W = O(1)$ and $N \ge J$ this is $O(NL)$ with
$O(1)$ score evaluations.
\item Counted in score evaluations, a $k$-coupled set $S$ with $\msig(S)$ witnesses
gives the ratio $\Omega\bigl((W-1)^k / (\msig(S)\, 2^{b+n_{\mathrm{nc}}} c)\bigr)$ of (a)
to (b), with $c$ as in \Cref{thm:poly}. For $W \ge 3$ it grows exponentially along
every family of frames in which $k \to \infty$ while $W$, $\msig(S)$,
$b + n_{\mathrm{nc}}$ and $c$ stay bounded, and \Cref{cor:capstone-family} gives one.
The ratio compares evaluation counts; it does not assert that WSBP succeeds.
\end{enumerate}
\end{theorem}
\begin{proof}
(a) Instantiate the Coupling Barrier Theorem and \Cref{cor:yao}(b) with $D = R$ and $d = W$. (b)
Apply \Cref{thm:poly}. The backward trace and structural merges run union-find on the
wire topology without executing the program or computing match scores, and the
realization's permutation enumeration is likewise score-free, so the target is
read only when the $O(N_{\mathrm{cand}})$ realized candidates are scored. (c) Apply (a) to $S$, which \Cref{cor:yao} allows for every coupled set, and divide
by the $N_{\mathrm{cand}} \le 2^{b+n_{\mathrm{nc}}} c$ evaluations of (b). The bound of (a) is a
worst case over the family $\mathcal{A}_{\delta,\Delta}$, while the count of (b) holds
for every target in the frame, whether or not WSBP succeeds on it.
\end{proof}

\begin{remark}[What the coupling says]\label{rem:coupling-meaning}
In \Cref{thm:coupling-capstone} the search cost is governed by the coupling
width (clause~(III) of \Cref{thm:barrier}). The cost of exact inference is governed by
the treewidth of the constraint graph (clause~(II) of \Cref{thm:barrier}), and
$\kappa \le \twscore + 1 \le \tw + 1$ (\Cref{thm:coupling-tw}, \Cref{lem:gscore-primal})
ties the first to the second. Score-oracle search
pays the interaction budget in time (exponential in $\kappa$, complete) and WSBP
pays it in coverage (polynomial, incomplete). If one-way functions exist, no
algorithm running in polynomial time covers all targets on a universal machine
(\Cref{prop:owf}); for a fixed frame the question is open (\Cref{rem:owf}).

The coupling-treewidth bound $\kappa \le \twscore + 1$
(\Cref{thm:coupling-tw}) holds unconditionally for coupling in the sense of
\Cref{def:coupling-set}. It also holds for the relaxed coupling of the match score wherever
the theorem's hypotheses, $\Delta > 2\delta$ together with the local
minimum~\textup{(\ref{eq:lmh})}, are met, and since the match score counts
constraints, \Cref{lem:gscore-primal} extends it there to $\kappa \le \tw + 1$.
Where the bound holds, a relaxation that eases
inference simultaneously weakens the search
barrier. The scalar-fitness
channel cannot extract the structural information the AST channel provides, and
even the stronger oracle that evaluates partial programs returns no more than the
score on the provable class (\Cref{cor:partial-oracle}). Qualitatively different structural access closes the gap, and its price for WSBP is incompleteness.
\end{remark}

\begin{remark}[Scope of the register coupling width]\label{rem:reg-kappa-scope}
The width $\kappa$ in \Cref{thm:coupling-capstone} is the coupling width of the
register match-score CSP (\Cref{def:coupling-rm}), and the theorem assumes $\kappa \ge 1$. The growth statement of clause~(c) is
attained on the family of \Cref{cor:capstone-family}, which has a $k$-coupled set
with a single witness, so that $\kappa \ge k$, while $W$, $b + n_{\mathrm{nc}}$ and $c$
stay bounded. For type discovery the
concrete per-instance lower bound is
$\kappa_{\mathrm{bin}} \ge \Knt$ (\Cref{thm:binary-coupling}); its register-field analog, a
lower bound on this width, is established on the class of \Cref{prop:horner},
where $W = 3$ and $\kappa \ge k$ at the stated baseline; in general $\kappa$
enters as a hypothesis.
\end{remark}

\begin{corollary}[A family for clause~(c)]\label{cor:capstone-family}
Replace each $\textsc{sub}_j$ of \Cref{prop:horner} by $\textsc{add}_j$, which sets
$r_{\mathrm{acc}} \leftarrow r_{\mathrm{acc}} + R[f_j]$, and write $\textsc{dbl}$ as
the addition $r_{\mathrm{acc}} + r_{\mathrm{acc}}$. Then (i) to (iii) of
\Cref{prop:horner} hold, with its conditions on $p$ for the field count $F = 6k + 4$
of this frame and an $x_0$ chosen as there for this frame. Every binary operation
of the loop body is a commutative addition, so $b = n_{\mathrm{nc}} = 0$, and
$W = 3$. The set $S$ is $k$-coupled with the single witness $\alpha^*$, and a search
that reads only the score needs at least $(2^k + 1)/2$ queries on $S$ on average.
WSBP evaluates at most $c$ candidates, a constant (\Cref{thm:poly}), so the
ratio of clause~(c) of \Cref{thm:coupling-capstone} is at least
$(2^k + 1)/(2c)$, exponential in $k$. WSBP builds its candidates from
$(F^*, \tau^*)$ before it reads the target, so on this family it is itself a
search that reads only the score, with at most $c$ queries, and it scores above
$0$ on at most $cN$ of the $2^k$ targets. On $S$ the score is a point function,
so the bound is that of unstructured search, and the family shows that the
hypotheses of clause~(c) can be met.
\end{corollary}
\begin{proof}
When no $f_j$ is $r_{\mathrm{acc}}$, one iteration maps $r_{\mathrm{acc}}$ to
$\lambda r_{\mathrm{acc}} + v$, with $\lambda = 2^k$ and $v$ as in \Cref{prop:horner},
so $s_{\alpha^*}[i] = \lambda^{i+1} x_0 + v\,(\lambda^{i+1} - 1)/(\lambda - 1)$. The expression is
injective in $v$ for the same reason as in the proof of \Cref{prop:horner}, and $v$ is the representative in
$[0, p)$ of $s_{\alpha^*}[0] - \lambda x_0$, which gives (iii). An assignment that sets
$a \ge 1$ of the $f_j$ to $r_{\mathrm{acc}}$ doubles the accumulator $a$ more times
per iteration, so the coefficient of $x_0$ in its output at position $i$ is
$2^{(k+a)(i+1)}$. It differs from the target's $2^{k(i+1)}$, since
$0 < 2^{a(i+1)} - 1 < 2^{kN} < p$, so by the choice of $x_0$ the assignment scores
$0$, the baseline included. The rest of the proof of \Cref{prop:horner} applies
unchanged. All binary operations are additions, so WSBP enumerates no routing and
no operand ordering, and $N_{\mathrm{cand}} \le c = N_U\, W!\, N_P$ with $W = 3$.
The candidates depend only on $(F^*, \tau^*)$ (\Cref{def:wsbp-image}), and the
targets differ at every position, so each candidate scores above $0$ on at most
$N$ targets.
\end{proof}

Combining the cost bounds of \Cref{thm:coupling-capstone} and \Cref{thm:poly} with the coverage bound of \Cref{cor:quant-incomplete} gives a Pareto frontier on the axes of coverage and cost.

\begin{corollary}[Sharpened coupling as a Pareto frontier]\label{cor:pareto}
Let $(F^*, \tau^*)$ be a frame in
general position (\Cref{cor:quant-incomplete}) with $W$ fixed and with
$b + n_{\mathrm{nc}}$ and $\log_2 c$ of order $O(\log |\Fields|)$, let $\nu$ be the uniform measure on $\mathcal{F}(F^*, \tau^*)$, and let $P^*$ be a
program in the frame whose coupling width makes $(W-1)^\kappa/\msig$
superpolynomial in $|\Fields|$. Then
(a) score-oracle search has universal coverage
at cost $\Omega((W-1)^\kappa / \msig)$, while
(b) WSBP has polynomial cost
$O(2^b J^2 + N_{\mathrm{cand}} \cdot NL)$, with $N_{\mathrm{cand}} \le 2^{b+n_{\mathrm{nc}}} c$
polynomial in $|\Fields|$, but realizes at most a
$2^{-\Omega(|\Fields|)}$ fraction of the frame's functions. The two bounds pin a
Pareto frontier on the axes (coverage of
$\mathcal{F}(F^*, \tau^*)$, cost). No score-oracle algorithm occupies the strictly dominating corner (polynomial cost,
universal coverage), since that would contradict clause~(a), and the strictly dominated
corner is Pareto-inferior to score-oracle search. WSBP attains the polynomial-cost corner only at restricted coverage. Whether some algorithm, score-oracle or structural, occupies the dominating corner is not settled here, since clause~(a) constrains score-oracle search alone. For programs on a universal machine, \Cref{prop:owf} answers this question if one-way functions exist. The two bounds answer to different adversaries, the cost bound (clause~(a)) to the worst-case score and the coverage bound to the uniform measure $\nu$.
\end{corollary}
\begin{proof}
Clause~(a) is \Cref{thm:coupling-capstone}(a); clause~(b) combines
\Cref{thm:poly} (time) with \Cref{cor:quant-incomplete} (coverage). A
polynomial-cost, universal-coverage score-oracle algorithm would violate
the $\Omega((W-1)^\kappa/\msig)$ lower bound of clause~(a), and an
exponential-cost, restricted-coverage one is dominated in both axes by
score-oracle search on the same coverage subset. The two operating points
therefore lie on a Pareto frontier within the score-oracle class. WSBP builds
its candidates from the dataflow and scores at most $N_{\mathrm{cand}}$ of them
(\Cref{thm:poly}), so by clause~(a) it keeps polynomial cost only by giving up
coverage.
\Cref{thm:typical-incomplete} bounds its coverage directly and places it at the
polynomial-cost, restricted-coverage point.
\end{proof}

\begin{remark}[Not a formal reformulation]\label{rem:not-reformulation}
WSBP is not a score-preserving reformulation in the sense of
\Cref{def:reformulation}; the maps between register space and wire space do not
preserve score for arbitrary assignments, the wire-to-register map is
many-to-one, and the round trip is not the identity. Consequently
\Cref{cor:3-2f} does not apply to WSBP. \Cref{thm:typical-incomplete} is a
separate argument that bounds the coverage of this specific algorithm of bounded
adaptivity, and it complements \Cref{cor:3-2f}'s lower bound on the coupling width
of every score-preserving reformulation.
\end{remark}

We close with a result for programs on a universal machine. It assumes that one-way functions exist and applies to every algorithm running in polynomial time.

\begin{proposition}[Under one-way functions, polynomial time forces incompleteness]\label{prop:owf}
Let $U$ be the universal prefix machine of $K(\cdot)$. As is standard for
time-bounded complexity, we take $U$ to simulate every fixed machine with polynomial overhead.
For a constant $a$, call a program $p$ a \emph{fast generator} of
$s \in \{0,1\}^N$ if $U(p)$ outputs exactly $s$ and halts within $N^{a}$ steps.
If one-way functions exist, there are a constant $a$ and, for every large $\ell$,
a distribution of targets $s \in \{0,1\}^N$ with $N = \ell^2$, samplable in time
polynomial in $\ell$, such that
\begin{enumerate}[label=(\roman*)]
\item every target has a fast generator of $\ell + O(\log \ell)$ bits, and
\item for every algorithm $\mathcal{E}$ running in probabilistic polynomial time
and every $e = e(\ell)$ computable in time polynomial in $\ell$, the probability,
over $s$ and the coins of $\mathcal{E}$, that $\mathcal{E}(s)$ is a fast generator
of $s$ with at most $N - e$ bits is at most $2^{1-e} + \mathrm{negl}(\ell)$.
\end{enumerate}
Here $\mathrm{negl}$, which depends on $\mathcal{E}$ and $e$, is eventually below
$\ell^{-c}$ for every constant $c$.
With $e = N/2$, no algorithm running in polynomial time finds, with more than
negligible probability, even a fast generator of $N/2$ bits, although every target
has one of about $\sqrt{N}$ bits.
\end{proposition}
\begin{proof}
The argument is the standard one by which one-way functions make time-bounded
Kolmogorov complexity hard (Liu and Pass, 2020); we state it in the form used here.
One-way functions give a pseudorandom generator
$\mathcal{G} : \{0,1\}^\ell \to \{0,1\}^{\ell^2}$ computable in polynomial time
(H\aa stad, Impagliazzo, Levin, and Luby, 1999, with the standard length
extension and seed padding). For every test running in probabilistic polynomial time, the
probability that it accepts $\mathcal{G}(\zeta)$, with $\zeta$ uniform, differs by $\mathrm{negl}(\ell)$ from the
probability that it accepts a uniform string. The targets are
$s = \mathcal{G}(\zeta)$.

(i) A fixed code for $\mathcal{G}$, followed by a self-delimiting code for $\ell$
and by the seed $\zeta$, is a program of $\ell + O(\log \ell)$ bits that outputs
$\mathcal{G}(\zeta)$. It runs in time polynomial in $\ell$, hence, under the
polynomial overhead of $U$, within $N^{a}$ steps for a suitable constant $a$ and
every large $\ell$.

(ii) Let the test $\mathcal{V}$ run $\mathcal{E}$ on its input $s$ and accept
exactly when the output $p$ has at most $N - e$ bits and $U(p)$ outputs $s$ within
$N^{a}$ steps; $\mathcal{V}$ runs in polynomial time. On $\mathcal{G}(\zeta)$ it
accepts with the success probability of $\mathcal{E}$. On a uniform $s$ it accepts
only if one of the fewer than $2^{N-e+1}$ programs of at most $N - e$ bits outputs
$s$; each program outputs one string, so this happens with probability below
$2^{1-e}$. The two acceptance probabilities differ by $\mathrm{negl}(\ell)$, so $\mathcal{E}$ succeeds with probability at most $2^{1-e} + \mathrm{negl}(\ell)$, which is (ii).
\end{proof}

\begin{remark}[Scope and the hypothesis]\label{rem:owf}
If one-way functions exist, \Cref{prop:owf} answers, for programs on a universal
machine, the analogue of the question left open in \Cref{cor:pareto}. In this analogue, coverage means a fast generator shorter than the target, since a literal copy of the
target is always one. On these targets an algorithm running in polynomial time
finds a fast generator of at most $N/2$ bits only with negligible probability,
although it sees the whole target. For a fixed frame
$\mathcal{F}(F^*, \tau^*)$ the question stays open, since the targets above need
not be outputs of that frame.

The hypothesis cannot be dropped without proving
$\mathrm{P} \ne \mathrm{NP}$. Whether $s$ has a fast generator within a given
length whose code begins with a given prefix is an NP question, the generator being
a certificate checked in $N^{a}$ steps. If $\mathrm{P} = \mathrm{NP}$, an algorithm
running in polynomial time takes the least such length and then extends the prefix one
bit at a time, stopping as soon as the prefix is itself a fast generator. It returns
a generator of $\ell + O(\log \ell)$ bits on every target, contradicting (ii).
\end{remark}

\section*{SI Appendix B. The OMNIS Architecture}

OMNIS is a deterministic program-search engine over a universal
register-machine instruction set. It approximates Solomonoff induction, with a universal
hypothesis space (the instruction set), a prior that decreases with description length
(the code described below), selection of the shortest program that compresses,
prediction by a length-weighted vote, and search by structured enumeration.

The engine has no learned state and no randomness, so no
result comes from training. Two runs can still differ through timing. Every phase stops at a
wall-clock deadline, the time budget or an earlier one of its own. When programs of equal score tie, thread order can decide
which one is reported, and the reported description length then changes by at most half a bit. Of the
2,095 distinct targets in this corpus, five appear more than once with different
outcomes. In each, one run made a nested-loop discovery between 275 and 378\,s and another
timed out.
The engine is written in C++17, compiled with Apple clang 17.0.0, and runs its search phases in parallel threads. The four sweeps ran on one Apple M2 computer under Darwin 25.3.0, in an x86\_64 environment translated by Rosetta 2, as their manifests in the deposited data record.

\subsubsection*{Instruction set} Programs are built from 18 canonical
operations (Table~\ref{tab:isa}), together
with one further metaoperation that calls a stored program as a subroutine. The 18 are increment and decrement; addition, subtraction, and multiplication;
multiply and modulo and divide by a constant or a register; load and copy;
output; the bitwise AND, OR, and XOR; a zero test; and a single loop.
Each instruction carries an opcode, a constant, up to four register arguments,
and an arity. The subroutine operation is the only route by which one program
could draw on another, and it is disabled throughout the experiments reported
here, which enforces the condition of no transfer (C3 of the main text).

One interpreter
executes the instructions of every program over registers of 64 bits, through one
dispatch over the opcodes. A version of it with registers of 512 bits, in which the product of two registers is
not available, executes the programs that output one bit of a register, such as the program of 74.5 bits for Rule 30. No
strategy has its own interpreter.

\subsubsection*{The 18 operations and the domain size $d = 54$} Table~\ref{tab:isa}
lists the instruction set as the type-discovery search enumerates it. Each
operation contributes one catalog entry for every admissible constant.
Thirteen operations carry no constant and contribute one entry each. The constant
multiply contributes nine ($c \in \{0,\dots,8\}$), the constant modulo eight
($c \in \{1,\dots,8\}$), the constant divide eight ($c \in \{1,\dots,8\}$), and the
load eleven ($c \in \{0,\dots,10\}$). The loop contributes one entry per body
length $\ell \in \{1,\dots,4\}$, four in all. The catalog the search draws from therefore has
\[
  \underbrace{13}_{\text{constant-free}}
        \;+\; \underbrace{9 + 8 + 8 + 11}_{36}
        \;+\; \underbrace{4}_{\text{loop lengths}}
        \;=\; 53
\]
entries, none of which is the identity. With the baseline (identity) type added,
the domain of the type-discovery CSP has size $d = 54$, the catalog size $T$ of
SI Appendix A, and the witness cube has
$d - 1 = 53$ choices per slot, which gives the barrier $53^{\kappa}$ of the main text. Two
pairs of entries coincide as functions, \texttt{MUL\_C} with $c = 1$ and
\texttt{COPY} (both copy the source register), and \texttt{MUL\_C} with $c = 0$
and \texttt{MOD\_C} with $c = 1$ (both write zero). This coincidence adds signal witnesses
and never removes one, so it is absorbed by the witness count $\msig$ in every
bound. The subroutine metaoperation is disabled in the empty-library runs reported
here and contributes no entry. The
register-inference barrier $7^{\kappa}$ uses a separate domain, the eight
registers $R_0, \dots, R_7$ ($d = 8$).

\begin{table}[ht]
\centering
\caption{The OMNIS instruction set, as the type-discovery search enumerates it.
``Arity'' is the number of register operands; operations with a constant range
also read that constant, and the loop reads a body length $\ell$. ``$c$'' is the
admissible constant range, with ``n/a'' for operations that carry no constant.
``Entries'' is the number of catalog entries the operation contributes; the
entries total $53$, which with the baseline type gives the type-discovery
domain $d = 54$. The subroutine metaoperation
is disabled throughout (empty library) and contributes no entry.}
\label{tab:isa}
\begin{tabular}{rlccr}
\toprule
\# & Mnemonic & Arity & $c$ & Entries \\
\midrule
0  & \texttt{INC}    & 1 & n/a                      & 1 \\
1  & \texttt{DEC}    & 1 & n/a                      & 1 \\
2  & \texttt{ADD}    & 3 & n/a                      & 1 \\
3  & \texttt{SUB}    & 3 & n/a                      & 1 \\
4  & \texttt{MUL}    & 3 & n/a                      & 1 \\
5  & \texttt{MUL\_C} & 2 & $\{0,\dots,8\}$          & 9 \\
6  & \texttt{MOD\_C} & 2 & $\{1,\dots,8\}$          & 8 \\
7  & \texttt{MOD\_R} & 3 & n/a                      & 1 \\
8  & \texttt{DIVC}   & 3 & $\{1,\dots,8\}$          & 8 \\
9  & \texttt{DIVR}   & 4 & n/a                      & 1 \\
10 & \texttt{LOAD}   & 1 & $\{0,\dots,10\}$         & 11 \\
11 & \texttt{COPY}   & 2 & n/a                      & 1 \\
12 & \texttt{OUT}    & 2 & n/a                      & 1 \\
13 & \texttt{AND}    & 3 & n/a                      & 1 \\
14 & \texttt{OR}     & 3 & n/a                      & 1 \\
15 & \texttt{XOR}    & 3 & n/a                      & 1 \\
16 & \texttt{ISZERO} & 2 & n/a                      & 1 \\
17 & \texttt{LOOP}   & 1 & $\ell \in \{1,\dots,4\}$ & 4 \\
\midrule
\multicolumn{4}{r}{total} & 53 \\
\bottomrule
\end{tabular}
\end{table}

\subsubsection*{Execution modes and program types}
The interpreter is driven in one of four ways,
an iterative mode that persists state, a functional mode that reloads the index each
step, an emission mode that collects outputs, and a context mode that loads
registers from the sequence's own history through a permutation. The choice of
mode is made by the harness, not by the program, and all four run the same
instruction set. Three further program types, named
by the type tag of the code described below, are fixed control forms with a few
parameters. A digit recursion applies one instruction to the digits of the
index in a chosen base (DARY), a digit concatenation writes the indices in a chosen base one after
another (CONCAT), and a step count counts the iterations of a branched body until
it halts (STEP).

\subsubsection*{Description length} Every candidate is measured by one
code, built from a universal integer code. In this code a type tag of two bits selects a
subdecoder, much as a universal machine's header selects a submachine, so the
description lengths of programs of different types, and of the families compared in the main
text, are directly comparable. This length is
the system's computable stand-in for the Kolmogorov complexity of the sequence,
an upper bound on it up to $O(\log N)$ bits for a sequence of length $N$, computed exactly for each
candidate. For a context-register program the $O(\log N)$ term also covers the
terms it starts from, as many as its context depth, which it takes from the
data.

\subsubsection*{Score and gate} During the search a candidate's score is the number of leading observed terms it
reproduces before its first mismatch, a prefix form of the match score
of SI Appendix A (\Cref{def:match-score}). The barrier holds for any score, and the bound with the treewidth of the
constraint graph is stated for the match score. A candidate compresses when it reproduces
all of the observed (training) terms and its description length is strictly below the bit
cost of the training data. It predicts when the held-out block is reproduced in full, that is, when the held-out score
equals the block length exactly.

The
prediction is a weighted vote, as in Solomonoff's mixture. Every program $p$ that the search finds
and that reproduces all training terms votes at each held-out position with weight
$2^{-|p|}$, where $|p|$ is its description length, and the symbol with the largest total weight is predicted. The two conditions define the
four classes used throughout, discovered, compressed only, predicted only, and neither. For the context-register programs the
held-out terms are produced autoregressively, each program reading its own output
back as input. A single early error then cascades through the rest of the
block and the prediction fails, which makes the prediction condition strict.

\subsubsection*{Enumeration phases} Search proceeds through a portfolio of enumeration
phases, each covering a different region of program space. The largest, by share of the full
corpus, is the context-register search (about 54\% of discoveries), which
enumerates short bodies over permutations of the sequence history. A closed-form
accelerator (about 34\%) checks whether, in a base $b$, the term at index $bn+r$ is fixed by
the term at $n$ and the digit $r$. If so, it deduces the instruction that computes
it, as ISA matching deduces an instruction from the table of observed transitions. It also tries every single instruction as
the step, and returns a digit recursion given by the base, an initial value and
that instruction. Backward propagation (about 6\% of discoveries) recovers the
nested-loop programs by backward demand-driven register inference. It is the
structural inference method of the main text and of SI Appendix A. It enumerates
operation sequences and builds the wiring of each from its dataflow, so it tests
only a bounded number of candidates per sequence.

The remaining
discoveries come from the flat sieve, the wide-integer sieve, the
digit-concatenation accelerator, and the branched cascade. These, like
the two large deductive subphases, inject some structure and search the rest by
score. Only backward propagation reads structure directly. The
library lookup and hierarchical synthesis complete the portfolio, but the
runs with an empty library leave them idle. Every phase
emits programs in the same set of 18 operations or in one of the three control
forms, every candidate is measured by the same description length, and every
candidate faces the same gate.

For reference, the search stages of Fig.~1 of the main text, the phases of the source code, and
the solver families of \Cref{fig:price} correspond as follows. In the source code,
ISA matching, the context-register search (family CTX\_X), the closed-form
accelerator (the DARY accelerator, family DARY), and the digit-concatenation
accelerator (family CONCAT) are subphases of a single deductive Phase~1.
Together they account for 2{,}103 of the 2{,}383 discoveries. The flat sieve
is Phase~2A and produces families FUNC and FUNC\_L together with the smaller
iterative variants ITER and ITER\_L. The branched cascade is Phase~2B and
produces family STEP. The wide-integer sieve is Phase~2C and produces family
WIDE\_BIT. Phase~2F is the backward propagation described above, the
unified wire-space backward propagation method (WSBP), which recovers
programs that contain loops by backward demand-driven register inference. In this
corpus it produces families NESTED\_LOOP and the smaller ITER\_F, 133
discoveries in all. Hierarchical
synthesis is Phase~2H; because the four sweeps run from an empty library, its
SUB\_CALL library has no entries to compose and it contributes no discoveries
to this corpus. The nine enumeration strategies of the main text are the four
deductive subphases (ISA matching, CTX\_X, DARY, CONCAT) and the five
enumerative phases (2A, 2B, 2C, 2F, 2H); the idle Phase~0 library lookup is not
counted among them.

\subsubsection*{Holdout sizing} The held-out block holds the larger of 20 terms and a quarter of the
available terms, which in practice is 20 terms for most sequences of the
On-Line Encyclopedia of Integer Sequences (OEIS), 50 for the Collatz step counts, and 125 for the
elementary cellular automata. Each of the 244 recovered elementary rules
reproduced all 125 held-out terms.

\subsubsection*{False acceptance} The prediction condition is a statistical test whose
type I error is bounded explicitly.

\begin{proposition}[Gate false-acceptance bound]\label{prop:gate-fap}
Fix a sequence over an alphabet $\Sigma$ and a held-out block $y \in \Sigma^{n_y}$
withheld from the search. Let $\hat{y}$ be the prediction of the gate, at each
held-out position the symbol with the largest total weight $2^{-|p|}$ over the
programs $p$ found by the search that reproduce every training term, where $|p|$ is the description length of $p$. Under the
null hypothesis $H_0$ that $y$ is drawn uniformly from $\Sigma^{n_y}$
independently of the training terms, the gate accepts with probability at most
\[
  \Pr\nolimits_{H_0}\bigl[\,\hat{y} = y\,\bigr] \;\le\; |\Sigma|^{-n_y}.
\]
\end{proposition}

\begin{proof}
The search never reads $y$, so $\hat{y}$ is fixed independently of $y$. A uniform
$y$ then equals $\hat{y}$ with probability at most $|\Sigma|^{-n_y}$, and the gate
accepts only if it does.
\end{proof}

\begin{remark}[Roles of compression and prediction]
Compression makes the returned program shorter than the data it explains, and
prediction supplies the test, $n_y \log_2 |\Sigma|$ bits the search never saw.
Since the prediction is committed before $y$ is read, the bound of
\Cref{prop:gate-fap} does not depend on how many programs the search examined or
on their length. For the elementary cellular automaton Rule 30 the held-out block
is $n_y = 125$ binary terms, so the false-acceptance probability is at most
$2^{-125}$.
\end{remark}

\begin{remark}[Empirical bound across the discovered set]
Over the 2,383 discoveries the held-out block carries $n_y \log_2 |\Sigma|$ bits,
at least $20$ and with median $198.1$ (Table~\ref{tab:sc6}), so every discovery
has false-acceptance probability at most $2^{-20}$ and the median is
$2^{-198.1}$. Summing each sequence's own bound $|\Sigma|^{-n_y}$ over all 3,914 sequences, the expected number of false
discoveries under $H_0$ is below $4 \times 10^{-4}$.
\end{remark}

\subsubsection*{Null calibration of the discovery gate} \Cref{prop:gate-fap} bounds the false-acceptance probability under its null of a uniform continuation by $|\Sigma|^{-n_y}$ per trial (at most $2^{-20}$, median $2^{-198.1}$ over the discoveries), whatever inference path produced the candidate. We
complement that analytic bound with a direct empirical calibration. We drew a
stratified sample of $321$ sequences across the four populations (S1 to S4) and
replaced each by a uniformly random permutation of its own terms, holding the
multiset exactly fixed. Per-sequence shuffle seeds are derived deterministically,
and multiset invariance is audited when the workload is built. Each shuffled sequence
was submitted to the unchanged engine at the same $600$\,s budget with a frozen
program library, so no state crosses trials, and discoveries are counted by the
same criterion as the main corpus. Unlike the main corpus, these runs loaded the engine's default library of five programs, read only, so the library lookup and hierarchical synthesis also ran; they supplied the best candidate of $60$ trials, all of which failed both gate conditions.

A permutation ablation removes structure only where the multiset does not already
encode it, so we read the false-alarm rate off the high-entropy cohort
(the $n = 201$ sequences with normalized multiset entropy $H \ge 0.7$), the regime in which
shuffling destroys the generating order. The gate fired on none of
these sequences, zero discoveries in $201$ trials, which gives a one-sided Clopper-Pearson
$95\%$ upper bound of $1.48\%$ on the per-sequence false-alarm rate. The zero is uniform across populations (S1,
$0$ of $25$; S2, $0$ of $30$; S3, $0$ of $67$; S4, $0$ of $79$), and no sequence
in this cohort was predicted or compressed.

The only discoveries arose in the low-entropy tail, where a permutation leaves a
near-constant or strongly biased multiset that still carries
the structure of a short program. These are detections of the structure the shuffle
keeps, a case the uniform null of \Cref{prop:gate-fap} does not describe. Six of the $321$ trials passed the gate, all at $H < 0.5$. Five were
near-constant ternary sequences ($H \approx 0.02$), four recovered as
context-register programs and one by backward propagation. The other was
a strongly biased quaternary sequence ($H \approx 0.14$), recovered as a short
functional loop. Two further biased binary sequences
($H \approx 0.12$) were predicted but not compressed. The engine reproduced their
held-out block exactly, but the recovered program did not meet the compression
condition, so the conjunctive gate correctly declined them. The compression
requirement is therefore not vacuous on this workload. \Cref{tab:snull} reports the
counts stratified by tier; the pooled rate over all $321$ trials does not estimate
the false-alarm rate, which the high-entropy cohort isolates.

\begin{table}[t]
\centering
\caption{Null calibration of the discovery gate on $321$ shuffled controls, stratified
by normalized multiset entropy $H$. Discoveries are counted by the criterion of the main
corpus, both gate conditions met. Predicted-only trials reproduce the held-out
block but do not meet the compression condition, so the conjunctive gate
declines them. Compressed-only trials meet the compression condition but miss part of the
held-out block, so the gate also declines them. The high-entropy cohort ($H \ge 0.7$) isolates the false-alarm
rate; the one-sided Clopper-Pearson $95\%$ upper bound there is $1.48\%$.}
\label{tab:snull}
\begin{tabular}{@{}lrrrrr@{}}
\toprule
Entropy tier & $n$ & discoveries & predicted only & compressed only & Clopper-Pearson $95\%$ upper \\
\midrule
high, $H \ge 0.7$           & $201$ & $0$ & $0$ & $0$ & $1.48\%$  \\
moderate, $0.5 \le H < 0.7$ & $1$   & $0$ & $0$ & $0$ & $95.00\%$ \\
low, $0.1 \le H < 0.5$      & $37$  & $1$ & $2$ & $7$ & $12.19\%$ \\
near-constant, $H < 0.1$    & $82$  & $5$ & $0$ & $16$ & $12.39\%$ \\
\bottomrule
\end{tabular}
\end{table}

\begin{figure}[htbp]
\centering
\includegraphics[width=0.85\linewidth]{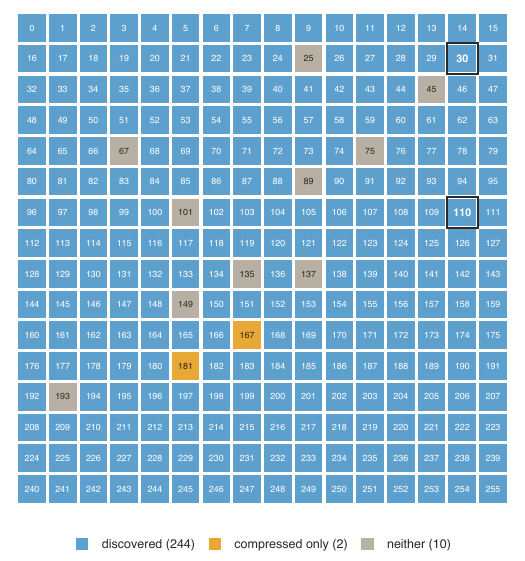}
\caption{Elementary cellular automata, full population. Each cell is one of the 256 elementary (two-state, three-neighbor) rules, shaded by outcome under the gate. The three shades mark the discovery of a generating program that both compresses and predicts the held-out continuation (244 rules), compression without prediction (2 rules, Rules 167 and 181), and neither (10 rules, Rules 25, 45, 67, 75, 89, 101, 135, 137, 149 and 193). Each target is the center column of the rule's evolution from a single live cell. A rule and its mirror image give the same column, so the 256 rules give 35 distinct sequences, 29 of them recovered, and the 12 failures are six sequences. Rules 45 and 101 and Rules 75 and 89 are chaotic. Rules 135 and 149 are the complements of Rule 30, and Rules 137 and 193 those of Rule 110; a complement is not a relabeling of its parent's output, because the seed is not complemented with the rule. Rules 167 and 181 give a column of ones except at times 2, 4, 8, 16, and so on, and compress without predicting. Rules 25 and 67 give a column that becomes periodic with period 2 after a short transient, and run out of time. Rules 30 and 110 are outlined. Rule 137, one of the failures, is Rule 110 complemented and started from a block of three cells, so its column is the periodic background of Rule 110 crossed by its gliders, the behavior that makes that rule complex. Discovery uses the same architecture and time budget as every other population, with no tuning specific to any rule.}
\label{fig:eca}
\end{figure}

\begin{figure}[htbp]
\centering
\includegraphics[width=0.85\linewidth]{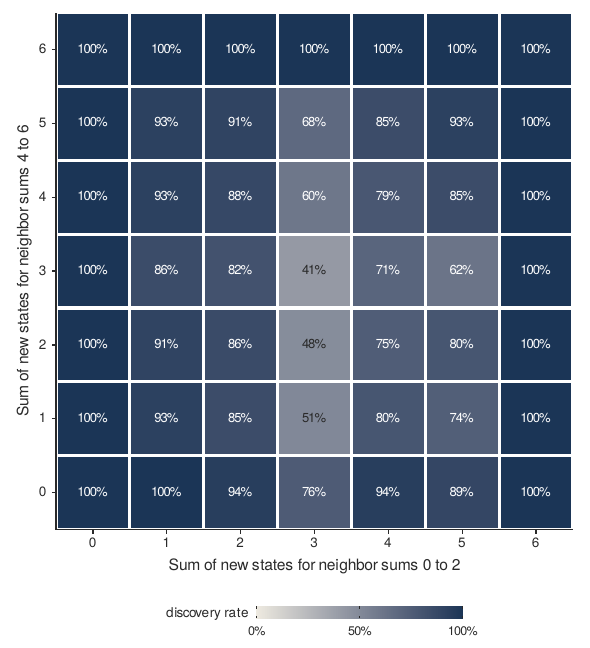}
\caption{Three-state totalistic cellular automata, discovery by output structure. The plane is indexed by two sums over a rule's output table, $d_0+d_1+d_2$ on the horizontal axis and $d_4+d_5+d_6$ on the vertical axis,
with $d_i = f(i)$ the new state for neighbor sum $i$. Each cell gives the fraction of rules with those values for which the engine discovered a generating program. Of the 2,187 rules, 1,683 (77\%) are discovered, covering 180 of the 666 distinct sequences these rules produce. The rate is nearly complete at the structural extremes and falls toward one half in the balanced interior.}
\label{fig:s4struct}
\end{figure}

\begin{figure}[htbp]
\centering
\includegraphics[width=0.85\linewidth]{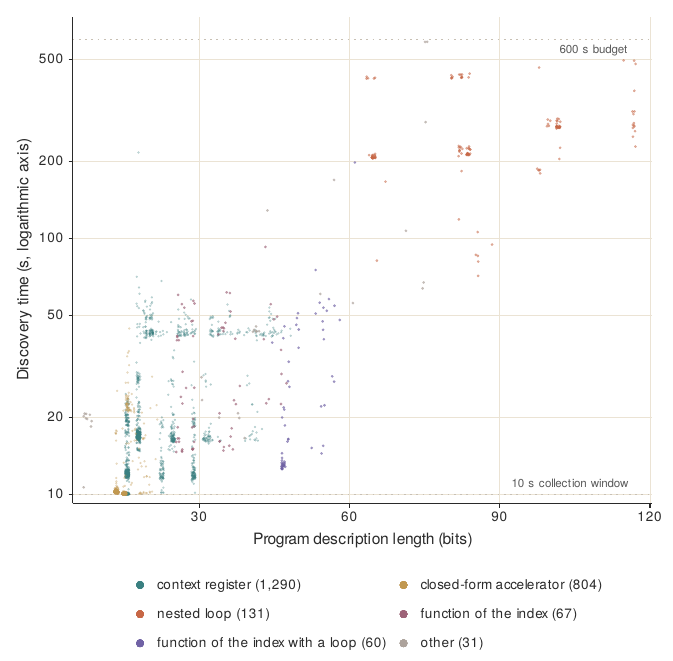}
\caption{Discovery time and program length (diagnostic). Discovery time, on a logarithmic axis, against program description length across all $2{,}383$ discoveries, with a small horizontal jitter separating the discrete bit lengths. Colors mark the solver family that produced each discovery, namely context-register programs ($n=1{,}290$), the closed-form accelerator ($n=804$), the nested-loop family ($n=131$), functions of the index computed by a body without a loop ($n=67$) or with one ($n=60$), and the other families ($n=31$). In the deposited data these families are CTX\_X, DARY, NESTED\_LOOP, FUNC and FUNC\_L. The two ITER\_F discoveries of backward propagation are folded into the other families, so its $133$ discoveries appear here as the $131$ nested-loop points. The dotted lines mark the floor and the ceiling of the time axis, the $10$\,s collection window that follows the first full match and the $600$\,s budget. Between them $\log_2$ time spans fewer than six doublings, so the figure shows which phase finds each program. The deductive families sit low, and the backward propagation family, which runs after the sieves and the branched cascade, sits high. The figure does not show the engine's search cost (SI Appendix B).}
\label{fig:price}
\end{figure}

\subsubsection*{The cellular-automaton targets and the center-column projection} Each
cellular-automaton target is the time evolution of the center column of a
one-dimensional automaton simulated on a lattice of width $1{,}512$ with a fixed zero
boundary, started from a single seeded cell (the center cell set to $1$, all
others zero) and run for $N = 500$ timesteps. At each timestep
$t = 0, 1, \dots, N-1$ the value of the center cell is emitted as the $t$th term
of the target, so the target has length $N = 500$. The lattice is wide enough
that the automaton's light cone never reaches the boundary within this horizon,
so the boundary condition is immaterial.

For the $256$ elementary automata (\Cref{fig:eca}) the
alphabet is $\{0, 1\}$ and the transition follows the Wolfram convention, the
output for the neighborhood $(l, c, r)$ being bit $4l + 2c + r$ of the
rule number, written in eight bits, over rule numbers $0$ to $255$. For the $2{,}187$ three-state
totalistic automata (\Cref{fig:s4struct}) the alphabet is $\{0, 1, 2\}$ and the transition depends only
on the neighbor sum $l + c + r \in \{0, \dots, 6\}$. Each rule is a
seven-tuple $(f(0), \dots, f(6)) \in \{0, 1, 2\}^7$, indexed by a code in
$[0, 3^7)$ through $f(i) = \lfloor \mathrm{code} / 3^i \rfloor \bmod 3$.
Recovering a cellular-automaton target therefore means discovering a
register-machine program whose first $500$ outputs match this center-column
sequence, that is, one that reproduces the $375$ training terms exactly and predicts
the $125$ held-out terms exactly within the per-instance budget. It does not mean
recovering the rule's behavior over infinite time or its action on other initial
conditions.

\subsubsection*{Length exclusion and the corpus size} A sequence is admitted only if its
length leaves a usable prefix of observed terms once the held-out block is removed. The
per-instance gate sets the held-out length to $n_y = \max(n_{\min}, \lfloor N \cdot k_{\mathrm{num}}/k_{\mathrm{den}} \rfloor)$ and requires the observed prefix $N - n_y \ge n_{\min}$. With the values used throughout, $n_{\min} = 20$ and $k_{\mathrm{num}}/k_{\mathrm{den}} = 1/4$, this requirement is
exactly $N \ge 40$. Of the $4{,}538$ raw candidates emitted across the four
populations, $612$ fail this criterion and are excluded before the search. None
are in the two cellular-automaton populations, whose targets are all of length
$N = 500$, so all $612$ come from the rows derived from the OEIS ($251$ of the $739$ raw
S1 candidates and $361$ of the $1{,}356$ raw S3 candidates). A further $12$ S3 rows are removed
by a deduplication step, the base-set rows of four OEIS sequences that the
core set of S3 already holds at the same alphabets. The analyzed corpus is therefore
$4{,}538 - 612 - 12 = 3{,}914$ (Table~\ref{tab:exclusion}).

The deduplication step works on
OEIS identifiers, and after reduction modulo the alphabet, different sequences can
coincide. The $3{,}914$ candidates hold $2{,}095$ distinct targets, $606$ of them with at least one discovery, and $43$ S3
rows, $29$ of them discoveries, share their target with an S1 row. The OEIS
category tagged hard loses the largest share, about 70\% of its raw pool, to
the length floor, which is why the alphabet analysis of Section~2.6 of the main text draws
on only 26 of its sequences.

\begin{table}[ht]
\centering
\caption{From the raw generated pool to the analyzed corpus. ``Excl.'' counts
sequences of length $N < 40$, removed before the search, and ``Dedup'' counts S3
base-set rows removed as duplicates of rows of the S3 core set. The four analyzed populations are
$488 + 256 + 983 + 2{,}187 = 3{,}914$.}
\label{tab:exclusion}
\begin{tabular}{llrrrr}
\toprule
Pop. & Categories & Raw & Excl.\ ($N<40$) & Dedup & Analyzed \\
\midrule
S1 & Collatz, OEIS morphic/cellular/hard & 739     & 251 & 0  & 488     \\
S2 & elementary cellular automata        & 256     & 0   & 0  & 256     \\
S3 & OEIS base/core                      & 1{,}356 & 361 & 12 & 983     \\
S4 & three-state totalistic cellular automata & 2{,}187 & 0   & 0  & 2{,}187 \\
\midrule
\multicolumn{2}{l}{total} & 4{,}538 & 612 & 12 & 3{,}914 \\
\bottomrule
\end{tabular}
\end{table}

\subsubsection*{Discovery time and program length} Table~\ref{tab:scaling} collects least-squares fits of $\log_2$ discovery time on program length. They describe the engine's phase schedule rather than its search cost. Every recorded time includes the $10$\,s collection window, so no time falls below $10$\,s and $755$ discoveries lie below $10.5$\,s. The $600$\,s budget caps the top. The $\log_2$ time therefore spans fewer than six doublings over program lengths from $7$ to $117$ bits, so wall time cannot show a cost that doubles with each added bit. The pooled slope is $0.054$. Within each of the four largest families $R^2$ is at most $0.18$, so most of the pooled fit reflects which phase finds a program. Three phases also carry deadlines, the context-register search at 30\,s (half the remaining time, clamped to $[5, 30]$\,s), the flat sieve at a third of the time left when it starts, and the branched cascade at half the remaining time, extended while it makes progress.

\section*{SI Appendix C. Sweep Data Tables}

Seven tables accompany the deposited sweep data and are rendered here as
Tables~\ref{tab:sc1} to~\ref{tab:sc7}. Table~\ref{tab:sc1} is the contingency table of gate outcomes per
population, the counts behind Fig.~2 of the main text. Table~\ref{tab:sc2} is the solver portfolio, the discoveries per family per population behind Fig.~S3 and the
per-population slopes of Section~2.4 of the main text. Table~\ref{tab:sc3} is the alphabet-asymmetry
table, the decay of Section~2.6 of the main text measured on the multialphabet cohort; the
hard subset of 26 sequences and the 30 inversion sequences are in the
deposited data. Across alphabets, the same sequence is often recovered by different solver families. Fewer than two in five of the sequences discovered at more than one alphabet are found by the same family at every alphabet where they are discovered, so the engine does not recover a single generator that is the same at every alphabet. Table~\ref{tab:sc4} gives the description-length distribution, and Table~\ref{tab:sc5} the
off-diagonal characterization, the 27 predicted-only candidates of Section~2.5
and the 17 compressed-only candidates of Section~2.2 of the main text, including the divisor-counting instance that misses
compression by 0.45 bits. The enumerative phases produce 26 of the 27 predicted-only candidates, since their programs tend to be longer than those of the deductive stage and so can cost more bits than short data. The one exception is the context-register recovery of the Fibonacci sequence, and the closed-form accelerator produces none across all 804 of its discoveries. Table~\ref{tab:sc6} gives the held-out test bits underlying the
false-acceptance bounds of SI Appendix B (at least 20 bits, median 198.1), and Table~\ref{tab:sc7}
the distribution of the total ratio, the quantity on the vertical axis of Fig.~2 of the main text. Table~\ref{tab:scaling}, which follows them, gives the fits of discovery time on program length discussed in SI Appendix B.
The machine-readable versions are provided in the deposited archive (Data,
Materials, and Software Availability of the main text).

\begin{table}[H]\centering
\caption{Gate outcomes by population. For each population S1 to S4, the number of candidates classified discovered (compresses and predicts), compressed only, predicted only, and neither; row sums match Table~\ref{tab:exclusion}, and column sums match the corpus totals of Fig.~2 of the main text (2,383 / 17 / 27 / 1,487).}\label{tab:sc1}
\begin{tabular}{lrrrrr}\toprule
Pop. & discovered & compressed only & predicted only & neither & total \\\midrule
S1 & 120 & 1 & 8 & 359 & 488 \\
S2 & 244 & 2 & 0 & 10 & 256 \\
S3 & 336 & 13 & 19 & 615 & 983 \\
S4 & 1,683 & 1 & 0 & 503 & 2,187 \\
\midrule
total & 2,383 & 17 & 27 & 1,487 & 3,914 \\
\bottomrule\end{tabular}\end{table}

\begin{table}[H]\centering
\caption{Solver portfolio. Discoveries per solver family and per population; family totals match Fig.~S3. The final row gives the per-population ordinary least squares slope of $\log_2$ discovery time against description length (Section~2.4 of the main text).}\label{tab:sc2}
\begin{tabular}{lrrrrr}\toprule
Family & S1 & S2 & S3 & S4 & total \\\midrule
CONCAT & 0 & 0 & 9 & 0 & 9 \\
CTX\_X & 38 & 80 & 108 & 1,064 & 1,290 \\
DARY & 57 & 147 & 123 & 477 & 804 \\
FUNC & 7 & 8 & 48 & 4 & 67 \\
FUNC\_L & 9 & 3 & 19 & 29 & 60 \\
ITER & 2 & 2 & 4 & 5 & 13 \\
ITER\_F & 0 & 0 & 1 & 1 & 2 \\
ITER\_L & 0 & 0 & 0 & 1 & 1 \\
NESTED\_LOOP & 5 & 0 & 24 & 102 & 131 \\
STEP & 1 & 0 & 0 & 0 & 1 \\
WIDE\_BIT & 1 & 4 & 0 & 0 & 5 \\
\midrule
slope & 0.048 & 0.060 & 0.045 & 0.058 & \\
\bottomrule\end{tabular}\end{table}

\begin{table}[H]\centering
\caption{Alphabet asymmetry. Discovery rate by alphabet size for the multialphabet OEIS cohort, the sequences encoded at two or more alphabets (rates are discoveries over candidates at that alphabet). This cohort holds every OEIS sequence of S1 and S3, not only the base keyword set, and so runs below the base-set rates of Section~2.6 of the main text (54\%, 33\% and 27\%), while showing the same drop from the binary alphabet. The 26 sequences of the hard keyword set and the 30 inversion sequences of S3 (of its 261 sequences encoded at two or more alphabets) that become easier at a larger alphabet are listed in the deposited data. The paired outcomes of the 26 sequences of the hard keyword set, at the binary and at a larger alphabet, feed the exact McNemar test ($p\approx0.016$).}\label{tab:sc3}
\begin{tabular}{lrrr}\toprule
Cohort & $|\Sigma| = 2$ & $|\Sigma| = 3$ & $|\Sigma| = 4$ \\\midrule
multialphabet cohort & 45\% & 24\% & 25\% \\
\bottomrule\end{tabular}
\end{table}

\begin{table}[H]\centering
\caption{Description-length distribution of discovered programs, in bits, per population.}\label{tab:sc4}
\begin{tabular}{lrrrr}\toprule
Pop. & min & median & mean & max \\\midrule
S1 & 13 & 15.6 & 24.6 & 117 \\
S2 & 13 & 15.0 & 18.8 & 75 \\
S3 & 7 & 17.5 & 27.5 & 117 \\
S4 & 13 & 17.9 & 24.6 & 117 \\
\bottomrule\end{tabular}\end{table}

\begin{table}[H]\centering
\caption{Off-diagonal characterization. The compressed-only (17) and predicted-only (27) candidates, with population, solver family, and the compression deficit $|p|-\text{train\_bits}$ in bits (negative for compressed-only, positive for predicted-only). The table includes the divisor-counting instance (A000005\_a2) that misses compression by about 0.45 bits (Section~2.5 of the main text).}\label{tab:sc5}
\begin{tabular}{llllr}\toprule
Candidate & pop. & family & class & deficit (bits) \\\midrule
A246318\_a2 & S1 & DARY & compressed only & -20.55 \\
eca\_167 & S2 & WIDE\_BIT & compressed only & -300.47 \\
eca\_181 & S2 & WIDE\_BIT & compressed only & -300.47 \\
A007093\_a4 & S3 & FUNC & compressed only & -52.91 \\
A043096\_a2 & S3 & NESTED\_LOOP & compressed only & -16.19 \\
A043096\_a3 & S3 & NESTED\_LOOP & compressed only & -64.16 \\
A043096\_a4 & S3 & NESTED\_LOOP & compressed only & -98.19 \\
A052382\_a3 & S3 & DARY & compressed only & -89.73 \\
A052382\_a4 & S3 & ITER\_L & compressed only & -72.67 \\
A055641\_a2 & S3 & FUNC & compressed only & -44.91 \\
A055641\_a3 & S3 & FUNC & compressed only & -91.71 \\
A055641\_a4 & S3 & FUNC & compressed only & -124.91 \\
A122840\_a2 & S3 & FUNC & compressed only & -36.47 \\
A122840\_a3 & S3 & FUNC & compressed only & -89.12 \\
A122840\_a4 & S3 & FUNC & compressed only & -115.47 \\
A002654\_a4 & S3 & NESTED\_LOOP & compressed only & -41.22 \\
tot3\_0331 & S4 & STEP & compressed only & -535.88 \\
A079314\_a3 & S1 & NESTED\_LOOP & predicted only & 3.04 \\
A085587\_a2 & S1 & FUNC\_L & predicted only & 15.73 \\
A147562\_a3 & S1 & FUNC\_L & predicted only & 11.43 \\
A151922\_a3 & S1 & FUNC\_L & predicted only & 16.20 \\
A169699\_a2 & S1 & FUNC\_L & predicted only & 14.78 \\
A267700\_a3 & S1 & NESTED\_LOOP & predicted only & 39.85 \\
A001606\_a2 & S1 & FUNC\_L & predicted only & 4.53 \\
A039951\_a2 & S1 & FUNC & predicted only & 2.73 \\
A003754\_a2 & S3 & ITER & predicted only & 10.78 \\
A004488\_a2 & S3 & FUNC\_L & predicted only & 6.54 \\
A007931\_a3 & S3 & STEP & predicted only & 1.97 \\
A010785\_a2 & S3 & ITER & predicted only & 6.09 \\
A051885\_a2 & S3 & FUNC & predicted only & 4.09 \\
A056964\_a2 & S3 & NESTED\_LOOP & predicted only & 35.56 \\
A059893\_a2 & S3 & FUNC\_L & predicted only & 9.58 \\
A075928\_a2 & S3 & FUNC\_L & predicted only & 28.96 \\
A000005\_a2 & S3 & WIDE\_BIT & predicted only & 0.45 \\
A000045\_a2 & S3 & CTX\_X & predicted only & 3.11 \\
A000069\_a3 & S3 & NESTED\_LOOP & predicted only & 45.46 \\
A000203\_a3 & S3 & NESTED\_LOOP & predicted only & 37.53 \\
A000203\_a4 & S3 & NESTED\_LOOP & predicted only & 16.78 \\
A000203\_a5 & S3 & NESTED\_LOOP & predicted only & 0.68 \\
A000292\_a2 & S3 & FUNC & predicted only & 1.73 \\
A000292\_a7 & S3 & NESTED\_LOOP & predicted only & 11.97 \\
A001157\_a4 & S3 & NESTED\_LOOP & predicted only & 56.78 \\
A001969\_a3 & S3 & NESTED\_LOOP & predicted only & 28.40 \\
A002487\_a3 & S3 & NESTED\_LOOP & predicted only & 7.42 \\
\bottomrule\end{tabular}\end{table}

\begin{table}[H]\centering
\caption{Held-out test bits $n_y\log_2|\Sigma|$ over the 2,383 discoveries, per population. Each discovery's false-acceptance probability under $H_0$ is at most $2^{-n_y\log_2|\Sigma|}$ (\Cref{prop:gate-fap}).}\label{tab:sc6}
\begin{tabular}{lrrrr}\toprule
Pop. & $n$ & min & median & max \\\midrule
S1 & 120 & 20.0 & 20.0 & 100.0 \\
S2 & 244 & 125.0 & 125.0 & 125.0 \\
S3 & 336 & 20.0 & 31.7 & 73.0 \\
S4 & 1,683 & 198.1 & 198.1 & 198.1 \\
\midrule
total & 2,383 & 20.0 & 198.1 & 198.1 \\
\bottomrule\end{tabular}\end{table}

\begin{table}[H]\centering
\caption{Total-ratio distribution. Description length over the bits of the whole sequence, held-out block included, per gate class across the corpus, the quantity on the vertical axis of Fig.~2 of the main text.}\label{tab:sc7}
\begin{tabular}{lrrrr}\toprule
Class & min & median & mean & max \\\midrule
discovered & 0.017 & 0.028 & 0.076 & 0.702 \\
compressed only & 0.074 & 0.249 & 0.281 & 0.604 \\
predicted only & 0.583 & 0.799 & 0.844 & 1.239 \\
neither & 0.043 & 0.324 & 0.392 & 2.212 \\
\bottomrule\end{tabular}\end{table}

\begin{table}[H]\centering
\caption{Discovery time against program length, the fits behind Section~2.4 of the main text. Ordinary least squares of $\log_2$ discovery time on description length in bits, over the 2,383 discoveries and the subsets named. The fits describe the phase schedule and do not measure search cost (SI Appendix B). The context-register family (\texttt{CTX\_X}) stops at a phase deadline of 30\,s, which appears at about 42\,s of wall time. The fit that excises the band (row ``all but \texttt{CTX\_X} at 30 to 60 s'') removes only the 274 \texttt{CTX\_X} discoveries whose time falls in the window from 30 to 60\,s around that deadline. On all discoveries, a Theil--Sen fit gives $0.066$ ($95\%$ CI $[0.0635, 0.0693]$) and the Spearman correlation is $0.76$. Within the nested-loop family ($n = 131$), $\log_2$ time on nesting depth gives $R^2 = 0.07$. The longest discovery times are 378, 587, 314 and 497\,s in populations S1 to S4. Of the 3,914 candidates, 15 ran past 700\,s (774 to 2,746\,s) through gaps in inner-loop deadline checks, all recorded as failures; the other candidates that passed the 600\,s budget stopped by 623\,s.}\label{tab:scaling}
\begin{tabular}{lrcc}\toprule
Subset & $n$ & slope [95\% CI] & $R^2$ \\\midrule
all discoveries & 2,383 & 0.0539 [0.0525, 0.0554] & 0.69 \\
\quad S1 & 120 & 0.0480 [0.0420, 0.0541] & 0.68 \\
\quad S2 & 244 & 0.0601 [0.0515, 0.0687] & 0.44 \\
\quad S3 & 336 & 0.0451 [0.0425, 0.0477] & 0.78 \\
\quad S4 & 1,683 & 0.0577 [0.0559, 0.0595] & 0.71 \\
\texttt{CTX\_X} only & 1,290 & 0.0420 [0.0371, 0.0468] & 0.18 \\
all but \texttt{CTX\_X} at 30 to 60 s & 2,109 & 0.0538 [0.0526, 0.0550] & 0.78 \\
all but \texttt{CTX\_X} & 1,093 & 0.0566 [0.0552, 0.0580] & 0.85 \\
\quad of which S4 & 619 & 0.0618 [0.0602, 0.0635] & 0.90 \\
\bottomrule\end{tabular}\end{table}

\section*{SI Appendix D. Prior Bounds and Prior Systems}

SI Appendix D.1
compares the coupling-barrier lower bound with prior search and complexity
lower bounds, and SI Appendix D.2 compares OMNIS as a discovery system with prior
program-discovery systems.

\subsection*{SI Appendix D.1. Prior Lower Bounds, and How the Coupling Barrier Differs}

This section groups the prior results that the main text compares with the coupling barrier into four lineages. It states each result precisely and names the respect in which the
barrier departs from it.

The first lineage gives upper bounds. Levin (1973) showed that universal search
solves every search problem in time optimal up to a constant factor and an additive term linear in the input length. In the standard analysis of this search (Li and Vit\'anyi, Chapter 7) it
finds a solution of an inversion problem, a generating program included, within a constant factor of
$2^{|p|}\cdot t(p)$ steps for every inverter $p$, where $t(p)$ is the number of steps $p$ takes to output a solution, verification included. This time bound is the canonical
upper bound on program search. Li and Vit\'anyi also formalize the
time-bounded (Levin) complexity $\Kt(x) = \min_p\bigl(|p| + \log t(p)\bigr)$, and
Levin search runs within a constant factor of $2^{\Kt(x)}$ steps, the
upper-bound direction of our conservation law. Hutter (2002) constructed an
algorithm whose running time is at most five times the provable time bound of
any program that provably solves the problem, plus a constant multiple of the time
to compute that bound and a further constant, both independent of the input. Hutter's bound is an upper bound on the
cost of the best algorithm, not a lower bound on the cost of all algorithms, so
it places no floor on score-oracle search.

The second lineage gives lower bounds for inversion and meta-complexity. Yao (1990) proved that
inverting a random function or permutation with $S$ bits of preprocessed advice and $T$
queries requires $S\cdot T = \widetilde{\Omega}(N)$. Yao's bound is the closest
published conservation between advice and queries, but it concerns a random function or permutation, prices
the tradeoff in advice bits rather than units of coupling width, and is
information-theoretic rather than tied to constraint structure. Corrigan-Gibbs
and Kogan (2019) showed that any improvement on Yao's bound would imply new lower
bounds on depth-two circuits with arbitrary gates, and that strong lower bounds
on nonadaptive function-inversion algorithms would imply breakthrough lower
bounds on linear-size, logarithmic-depth circuits. Their result explains why the
inversion bound resists strengthening, but again the object is a random function
under preprocessing. Brandt (2024) proved that $\MKtP$, the decision problem for
Levin complexity $\Kt$, is unconditionally not in $\DTIME[O(n)]$ (and not in
$\DTIME[O(n^2)]$ on machine models with linear-time universal simulation). Brandt's result
is a lower bound on computing $\Kt$ itself, a different problem from searching
for a generating program.

The third lineage concerns CSP and treewidth. Razgon (2006) showed that binary CSP
with domain size $d$ is solved in $O^*\bigl((d-1)^n\bigr)$ by forward checking
combined with conflict-directed backjumping and fail-first ordering; the coupling barrier's base $d-1$ matches this and
lifts it to $(d-1)^\kappa$. Traxler (2008) proved that, under the Exponential Time Hypothesis (ETH),
$(d,2)$-CSP with bounded variable frequency and $d$-UNIQUE-CSP require time
$\Omega(d^{cn})$ for a constant $c>0$ independent of $d$, which in particular rules out
$d^{o(n)}$-time algorithms. The bound depends on the domain but is conditional on the
ETH. Lokshtanov, Marx, and Saurabh (2011) showed that, under the Strong Exponential Time Hypothesis (SETH), the standard
algorithms parameterized by treewidth are optimal; for every $\varepsilon > 0$, Independent Set is not solvable
in $(2-\varepsilon)^{\tw} n^{O(1)}$ time, Dominating Set not in $(3-\varepsilon)^{\tw} n^{O(1)}$, and
$q$-Coloring, for every $q \ge 3$, not in $(q-\varepsilon)^{\tw} n^{O(1)}$, all conditional on SETH. Marx (2010)
proved the weaker, nontight form that binary CSP admits no
$|I|^{o(\tw/\log \tw)}$ algorithm under the ETH.

Marx (2013) introduced
submodular width and showed that, under the ETH, it characterizes fixed-parameter tractability (FPT) for
unbounded-arity CSP (bounded submodular width if and only if
$\mathrm{CSP}(H)$ is FPT when parameterized by the number of variables). For
bounded arity the corresponding parameter is treewidth, modulo homomorphic
equivalence and under $\mathrm{W}[1] \ne \mathrm{FPT}$ (Grohe, 2007). Freuder
(1982) and Dechter and Pearl (1989) are the classical results on structure versus
search. Freuder showed that search is backtrack-free when the constraint graph's width is small
relative to the level of local consistency enforced, and Dechter and Pearl showed that a network
decomposed into a tree of clusters is solved without backtracking, in time exponential in the size of the largest cluster, which gives $O(n\,d^{\tw+1})$ once a tree decomposition of width $\tw$ is given (Freuder, 1990). All of these CSP results are parameterized by the explicit constraints, and those that bound every algorithm from below are conditional, whereas the coupling barrier is unconditional for algorithms that read only the score and is governed by the score itself.

The fourth lineage covers black-box optimization and measures of problem difficulty. Closest in mechanism to the coupling barrier is the complexity
theory of black-box (query) optimization. It restricts an algorithm to the
value oracle, exactly the score-oracle class, and proves unconditional query
lower bounds (Droste, Jansen, and
Wegener, 2006; the more restricted unbiased model is due to Lehre and Witt, 2012). The
bound $\Omega((T-1)^{\Knt}/\msig)$ of \Cref{thm:search} is of this kind. The
$(\delta,\Delta)$-coupled set is a needle of $\msig$ improving assignments
inside the informative region of size $(d-1)^{\kappa}$, with the baseline
returned everywhere outside it. For a single improving assignment the bound reduces to the known bound for finding a needle, and it extends that bound to coupled fields with any number of improving assignments. 

The no-free-lunch theorems (Wolpert and
Macready, 1997) are themselves a conservation result, but of average
performance over the uniform distribution on all objective functions, whereas
the conservation law here is a per-instance budget priced in bits of
$K(s)$. In evolutionary computation the canonical hard instances are the
deceptive and trap functions, in which low-order schema information misleads
(Goldberg, 1987; Whitley, 1991; Deb and Goldberg, 1993). A fully deceptive
trap of order $\ell$ requires all subpatterns of order $(\ell-1)$ to mislead.
Deception is a property of the gradient, which points toward the local optimum,
whereas the flat basin is gradient-free. The genetic algorithm analyses show trap
hardness for particular heuristics, and Droste, Jansen, and Wegener (2006) show that a
trap whose optimum is hidden costs every black-box algorithm as many queries as a needle.
A distinct neutral case, the plateau of constant fitness that an algorithm
crosses by accepting moves of equal fitness, is treated by Jansen and Wegener
(2001). There the plateau is traversable, whereas the flat basin has an
exponentially rare exit and no neutral path to it. 

Closest in intent are the
measures of problem difficulty read off the fitness landscape. Kauffman's $\mathrm{NK}$
ruggedness parameter (Kauffman, 1993) and epistasis variance (Davidor, 1991) estimate hardness from the
interaction structure of the objective, and fitness-distance correlation (Jones and Forrest, 1995)
from how fitness correlates with distance to the global optimum. These measures are descriptive and correlational
and admit counterexamples (He, Reeves, Witt, and Yao, 2007), and \Cref{prop:incomp-sens} shows that coupling width
is not a function of single-coordinate sensitivity, which, like coupling width, is defined by changes at the baseline.
What coupling width adds is a proven query lower bound, the treewidth
characterization $\kappa \le \twscore + 1$, and the conservation law of
\Cref{cor:cw-unit}.

The barrier departs from the bounds
above in three respects, and from every one of them in the second and third. First, it is unconditional for the score-oracle model,
requiring neither SETH nor ETH, whereas each prior CSP lower bound that holds for every algorithm is conditional.
Second, it is governed by the coupling width $\kappa$ of the score-interaction
graph, which can be far below the primal treewidth that parameterizes the prior
bounds. The incomparability proposition (\cref{prop:incomp-tw}) exhibits instances of treewidth $n-1$ with coupling width $1$, and the reverse. Third, it comes with a
conservation law (\cref{cor:cw-unit}) that prices reformulation in coupling-width
units against the search it eliminates at a rate of at most one to one, which no prior
bound provides. Levin's $2^{K(s)}$ sits at the end of this law where nothing is reformulated,
and the lower bound approaches it as the catalog grows (\Cref{rem:levin-base}).
Table~\ref{tab:prior-bounds} summarizes the comparison.

\begin{table}[t]
\centering
\caption{Prior bounds and structural parameters, against the coupling barrier.
The barrier is the only entry that is at once a lower bound, unconditional,
parameterized by coupling width rather than primal treewidth, and equipped with
a conservation law. \textsuperscript{$\dagger$}Information-theoretic.
\textsuperscript{$\ddagger$}Unconditional within the score-oracle model. Hutter
(2002) and Corrigan-Gibbs and Kogan (2019) are discussed in the text but omitted
from the grid. The former bounds the cost of the single best algorithm rather
than all algorithms, and the latter is a barrier to strengthening inversion
bounds rather than a parameterized bound.}
\label{tab:prior-bounds}
\small
\setlength{\tabcolsep}{4pt}
\renewcommand{\arraystretch}{1.18}
\begin{tabular}{@{}>{\raggedright\arraybackslash}p{2.7cm}>{\raggedright\arraybackslash}p{2.85cm}>{\raggedright\arraybackslash}p{1.45cm}>{\raggedright\arraybackslash}p{1.45cm}>{\raggedright\arraybackslash}p{2.5cm}>{\raggedright\arraybackslash}p{2.3cm}@{}}
\toprule
Result & Object bounded & Type & Cond.\ on & Governing parameter & Conservation law? \\
\midrule
Levin 1973 & program search & upper & none & program length $/\ K(s)$ & endpoint (no reformulation) \\
Yao 1990 & random function or permutation inversion w/ preproc. & lower & none\textsuperscript{$\dagger$} & advice $S$, queries $T$ & partial ($S\cdot T$, in bits) \\
Brandt 2024 & computing $\Kt$ ($\MKtP$) & lower & none & input length $n$ & no \\
Razgon 2006 & binary CSP & upper & none & domain $d$ (base $d-1$) & no \\
Traxler 2008 & $(d,2)$-CSP & lower & ETH & domain $d$ & no \\
Lokshtanov, Marx, Saurabh 2011 & Ind.\ Set, Dom.\ Set, $q$-Coloring (bdd.\ tw) & lower & SETH & treewidth & no \\
Marx 2010 & binary CSP & lower & ETH & treewidth ($\tw/\log\tw$) & no \\
Marx 2013 & unbounded-arity CSP & FPT char. & ETH & submodular width & no \\
Grohe 2007 & bounded-arity CSP & FPT char. & $\mathrm{W}[1]\!\neq\!\mathrm{FPT}$ & treewidth (cores, mod.\ hom.\ equiv.) & no \\
Freuder 1982, 1990; Dechter--Pearl 1989 & CSP structure $\to$ search & qual. & none & graph width / treewidth & no \\
Wolpert \& Macready 1997 & optimization over all functions & lower & none\textsuperscript{$\dagger$} & uniform prior over $f$ & avg.\ performance \\
Droste--Jansen--Wegener 2006; Lehre \& Witt 2012 & black-box (query) opt. & lower & none & queries (unbiased: arity) & no \\
Deb \& Goldberg 1993; Jansen \& Wegener 2001 & trap / plateau hardness & lower & heuristic & deception order / plateau & no \\
Kauffman 1993; Jones \& Forrest 1995 & genetic algorithm difficulty (descriptive) & correl. & n/a & $\mathrm{NK}$ degree / fitness-distance correlation & no \\
This work & score-oracle program search & lower + cons. & none\textsuperscript{$\ddagger$} & coupling width $\kappa$ & yes (Cor.\ \ref{cor:cw-unit}) \\
\bottomrule
\end{tabular}
\end{table}

\subsection*{SI Appendix D.2. Comparison with Prior Program-Discovery Systems}

Section~2.1 of the main text states the four conditions under which every
reported discovery is obtained and summarizes how prior systems relax them. The conditions are (C1) compress the observed terms and predict an
isolated held-out block; (C2) one universal instruction set or three fixed control forms, the same for
every sequence; (C3) an empty library; and (C4) no outside knowledge imported,
whether as a closed-form fitter, a recurrence solver, a learned policy, or programs found for
other sequences.
Table~\ref{tab:c1c4} records the comparison condition by condition; the
details behind each entry follow.

LODA (Krause) mines the OEIS at
the scale of 150,000 programs, but its programs may call
previously discovered programs as subroutines, so a program can build on
programs found for other sequences. Part of its search mutates existing programs,
alongside generators that enumerate or sample new ones, and it validates by matching observed terms,
without a held-out prediction requirement. LODA therefore relaxes C1, C3, and C4.

Gauthier and Urban (2023) search a fixed language of basic arithmetic operators
and loops, so their system holds C2. Their tree search is guided by a policy trained on the
shortest programs found so far for other sequences. The search thus carries knowledge
between candidate sequences and uses a learned policy, so it relaxes C3 and C4. A program is
accepted when it reproduces the given terms, which relaxes C1. After the search the authors
checked 12,639 of their solutions on 100 further terms, and more than 92\%
matched. The block decomposition methods (BDM) in the algorithmic information
tradition (Zenil and collaborators; Zenil et al., 2018) estimate complexity from lookup tables of
complexity values precomputed for small objects, which is a different task from recovering and certifying a
single generating program, and they impose no gate that requires compression and prediction.

On the
OMNIS side, the deductive subphases test candidates from fixed families against the
target, directly or through its table of transitions or digits, and import no
outside solver. The programs they return
are built from the same instructions and control forms, measured by the same code
and judged by the same gate, so Table~\ref{tab:c1c4} marks C2 and C4 as held for OMNIS.

\begin{table}[t]
\centering
\caption{The four conditions under which OMNIS operates, against prior
program-discovery systems. A $\checkmark$ denotes that the condition is held, and a $\times$
that it is relaxed. OMNIS holds all four. \textsuperscript{$\dagger$}Zenil--BDM estimates an upper bound on Kolmogorov
complexity by composing precomputed complexity values for small objects (the
Coding Theorem Method lookup tables), rather than searching for and emitting a
single generating program. The criteria on the synthesis instruction set (C2) and on
imported solvers (C4) therefore have no program-discovery analogue and
are marked n/a. The reliance of Zenil--BDM on a precomputed library is recorded under
C3.}
\label{tab:c1c4}
\small
\setlength{\tabcolsep}{8pt}
\renewcommand{\arraystretch}{1.25}
\begin{tabular}{@{}>{\raggedright\arraybackslash}p{6.0cm}cccc@{}}
\toprule
Condition & OMNIS & LODA & Gauthier--Urban & Zenil--BDM \\
\midrule
(C1) compress and predict an isolated held-out block & $\checkmark$ & $\times$ & $\times$ & $\times$ \\
(C2) one universal instruction set or three fixed control forms, the same for every sequence & $\checkmark$ & $\checkmark$ & $\checkmark$ & n/a\textsuperscript{$\dagger$} \\
(C3) empty library, no transfer between candidate sequences & $\checkmark$ & $\times$ & $\times$ & $\times$ \\
(C4) no outside knowledge imported (closed-form fitter, recurrence solver, learned policy, programs of other sequences) & $\checkmark$ & $\times$ & $\times$ & n/a\textsuperscript{$\dagger$} \\
\bottomrule
\end{tabular}
\end{table}

\section*{SI References}\makeatletter\stepcounter{NAT@ctr}\makeatother

\noindent N. Brandt, ``Lower Bounds for Levin--Kolmogorov Complexity.'' In \emph{Theory of Cryptography (TCC 2024), Part I}, Lecture Notes in Computer Science, vol.\ 15364, pp.\ 191--221, Springer (2024). doi:10.1007/978-3-031-78011-0\_7. Also Cryptology ePrint Archive, Paper 2024/687.

\noindent H. Corrigan-Gibbs and D. Kogan, ``The Function-Inversion Problem: Barriers and Opportunities.'' In \emph{Theory of Cryptography (TCC 2019), Part I}, Lecture Notes in Computer Science, vol.\ 11891, pp.\ 393--421, Springer (2019). doi:10.1007/978-3-030-36030-6\_16.

\noindent Y. Davidor, ``Epistasis Variance: A Viewpoint on GA-Hardness.'' In \emph{Foundations of Genetic Algorithms} (G.\ J.\ E.\ Rawlins, ed.), pp.\ 23--35, Morgan Kaufmann (1991). doi:10.1016/B978-0-08-050684-5.50005-7.

\noindent K. Deb and D. E. Goldberg, ``Analyzing Deception in Trap Functions.'' In \emph{Foundations of Genetic Algorithms 2} (L.\ D.\ Whitley, ed.), pp.\ 93--108, Morgan Kaufmann (1993). doi:10.1016/B978-0-08-094832-4.50012-X.

\noindent R. Dechter and J. Pearl, ``Tree clustering for constraint networks.'' Artificial Intelligence 38(3), pp.\ 353--366 (1989).\newline doi:10.1016/0004-3702(89)90037-4.

\noindent S. Droste, T. Jansen, and I. Wegener, ``Upper and Lower Bounds for Randomized Search Heuristics in Black-Box Optimization.'' Theory of Computing Systems 39(4), pp.\ 525--544 (2006). doi:10.1007/s00224-004-1177-z.

\noindent E. C. Freuder, ``A Sufficient Condition for Backtrack-Free Search.'' Journal of the ACM 29(1), pp.\ 24--32 (1982).\newline doi:10.1145/322290.322292.

\noindent E. C. Freuder, ``Complexity of K-Tree Structured Constraint Satisfaction Problems.'' In \emph{Proc.\ 8th National Conference on Artificial Intelligence (AAAI-90)}, pp.\ 4--9 (1990).

\noindent T. Gauthier and J. Urban, ``Learning Program Synthesis for Integer Sequences from Scratch.'' \emph{Proceedings of the AAAI Conference on Artificial Intelligence} 37(6), pp.\ 7670--7677 (2023). doi:10.1609/aaai.v37i6.25930.

\noindent D. E. Goldberg, ``Simple Genetic Algorithms and the Minimal, Deceptive Problem.'' In \emph{Genetic Algorithms and Simulated Annealing} (L.\ Davis, ed.), pp.\ 74--88, Morgan Kaufmann (1987).

\noindent M. Grohe, ``The complexity of homomorphism and constraint satisfaction problems seen from the other side.'' Journal of the ACM 54(1), Art.\ 1, pp.\ 1:1--1:24 (2007). doi:10.1145/1206035.1206036.

\noindent J. H\aa stad, R. Impagliazzo, L. A. Levin, and M. Luby, ``A Pseudorandom Generator from any One-way Function.'' SIAM Journal on Computing 28(4), pp.\ 1364--1396 (1999). doi:10.1137/S0097539793244708.

\noindent J. He, C. R. Reeves, C. Witt, and X. Yao, ``A Note on Problem Difficulty Measures in Black-Box Optimization: Classification, Realizations and Predictability.'' Evolutionary Computation 15(4), pp.\ 435--443 (2007). doi:10.1162/evco.2007.15.4.435.

\noindent M. Hutter, ``The Fastest and Shortest Algorithm for All Well-Defined Problems.'' International Journal of Foundations of Computer Science 13(3), pp.\ 431--443 (2002). doi:10.1142/S0129054102001199.

\noindent T. Jansen and I. Wegener, ``Evolutionary Algorithms: How to Cope with Plateaus of Constant Fitness and When to Reject Strings of the Same Fitness.'' IEEE Transactions on Evolutionary Computation 5(6), pp.\ 589--599 (2001).\newline doi:10.1109/4235.974841.

\noindent T. Jones and S. Forrest, ``Fitness Distance Correlation as a Measure of Problem Difficulty for Genetic Algorithms.'' In \emph{Proc.\ 6th International Conference on Genetic Algorithms} (L.\ J.\ Eshelman, ed.), pp.\ 184--192, Morgan Kaufmann (1995).

\noindent S. A. Kauffman, \emph{The Origins of Order: Self-Organization and Selection in Evolution.} Oxford University Press (1993).

\noindent C. Krause, ``LODA: An assembly language, a computational model and a tool for mining integer sequence programs.'' Software, https://github.com/loda-lang (accessed 29 October 2026).

\noindent P. K. Lehre and C. Witt, ``Black-Box Search by Unbiased Variation.'' Algorithmica 64(4), pp.\ 623--642 (2012).\newline doi:10.1007/s00453-012-9616-8.

\noindent L. A. Levin, ``Universal sequential search problems.'' Problems of Information Transmission 9(3), pp.\ 265--266 (1973).

\noindent M. Li and P. M. B. Vit\'anyi, \emph{An Introduction to Kolmogorov Complexity and Its Applications.} 4th ed., Springer (2019). doi:10.1007/978-3-030-11298-1.

\noindent Y. Liu and R. Pass, ``On One-way Functions and Kolmogorov Complexity.'' In \emph{IEEE 61st Annual Symposium on Foundations of Computer Science (FOCS 2020)}, pp.\ 1243--1254 (2020). doi:10.1109/FOCS46700.2020.00118.

\noindent D. Lokshtanov, D. Marx, and S. Saurabh, ``Known algorithms on graphs of bounded treewidth are probably optimal.'' In \emph{Proc.\ 22nd Annual ACM--SIAM Symposium on Discrete Algorithms (SODA)}, pp.\ 777--789 (2011). doi:10.1137/1.9781611973082.61. Journal version: ACM Transactions on Algorithms 14(2), Art.\ 13 (2018). doi:10.1145/3170442.

\noindent D. Marx, ``Can you beat treewidth?'' Theory of Computing 6(1), pp.\ 85--112 (2010). doi:10.4086/toc.2010.v006a005.

\noindent D. Marx, ``Tractable hypergraph properties for constraint satisfaction and conjunctive queries.'' Journal of the ACM 60(6), Art.\ 42 (2013). doi:10.1145/2535926.

\noindent I. Razgon, ``Complexity Analysis of Heuristic CSP Search Algorithms.'' In \emph{Recent Advances in Constraints (CSCLP 2005, Revised Selected and Invited Papers)} (B.\ Hnich, M.\ Carlsson, F.\ Fages, F.\ Rossi, eds.), Lecture Notes in Computer Science (LNAI), vol.\ 3978, pp.\ 88--99, Springer (2006). doi:10.1007/11754602\_7.

\noindent L. J. Stockmeyer, ``The polynomial-time hierarchy.'' Theoretical Computer Science 3(1), pp.\ 1--22 (1976). doi:10.1016/0304-3975(76)90061-X.

\noindent P. Traxler, ``The Time Complexity of Constraint Satisfaction.'' In \emph{Parameterized and Exact Computation (IWPEC 2008)}, Lecture Notes in Computer Science, vol.\ 5018, pp.\ 190--201, Springer (2008). doi:10.1007/978-3-540-79723-4\_18.

\noindent E. van Nimwegen and J. P. Crutchfield, ``Optimizing epochal evolutionary search: Population-size independent theory.'' Computer Methods in Applied Mechanics and Engineering 186(2--4), pp.\ 171--194 (2000). doi:10.1016/S0045-7825(99)00383-7.

\noindent L. D. Whitley, ``Fundamental Principles of Deception in Genetic Search.'' In \emph{Foundations of Genetic Algorithms} (G.\ J.\ E.\ Rawlins, ed.), pp.\ 221--241, Morgan Kaufmann (1991). doi:10.1016/B978-0-08-050684-5.50017-3.

\noindent D. H. Wolpert and W. G. Macready, ``No Free Lunch Theorems for Optimization.'' IEEE Transactions on Evolutionary Computation 1(1), pp.\ 67--82 (1997). doi:10.1109/4235.585893.

\noindent A. C. Yao, ``Probabilistic computations: Toward a unified measure of complexity.'' In \emph{18th Annual Symposium on Foundations of Computer Science (SFCS 1977)}, pp.\ 222--227 (1977). doi:10.1109/SFCS.1977.24.

\noindent A. C. Yao, ``Coherent functions and program checkers.'' In \emph{Proc.\ 22nd Annual ACM Symposium on Theory of Computing (STOC)}, pp.\ 84--94 (1990). doi:10.1145/100216.100226.

\noindent H. Zenil, S. Hern\'andez-Orozco, N. A. Kiani, F. Soler-Toscano, A. Rueda-Toicen, and J. Tegn\'er, ``A Decomposition Method for Global Evaluation of Shannon Entropy and Local Estimations of Algorithmic Complexity.'' Entropy 20(8), 605 (2018). doi:10.3390/e20080605.